\documentclass[aps, prx, twocolumn, longbibliography, superscriptaddress,footinbib]{revtex4-2}

\usepackage{libertine}
\usepackage{amsmath, amsfonts, amssymb, bm}
\usepackage{braket} 
\usepackage{
	float
} 

\usepackage{bbold}

\usepackage{slashed}
\usepackage{enumerate}
\usepackage{
	enumitem
} 
\usepackage[all]{xy} 
\usepackage{tikz}
\usetikzlibrary{quantikz2}

\usepackage{amsthm} 

\newtheorem{theorem}{Theorem}
\newtheorem{assumption}{Assumption}
\newtheorem{lemma}[theorem]{Lemma}   
\newtheorem{corollary}[theorem]{Corollary}

\usepackage{graphics}
\graphicspath{{./}} 

\makeatletter
\newcommand{\apptocfile}{atoc}
\let\apptoc@orig@appendix\appendix
\renewcommand{\appendix}{%
	\apptoc@orig@appendix
	\let\apptoc@orig@addtocontents\addtocontents
	\long\def\addtocontents##1##2{%
		\def\apptoc@ext{##1}%
		\def\apptoc@toc{toc}%
		\ifx\apptoc@ext\apptoc@toc
		\apptoc@orig@addtocontents{\apptocfile}{##2}%
		\else
		\apptoc@orig@addtocontents{##1}{##2}%
		\fi
	}%
}
\newcommand{\appendixtableofcontents}{%
	\begingroup
	\setcounter{tocdepth}{3}%
	\phantomsection
	\let\addcontentsline\@gobblethree
	\section*{Appendix Contents}%
	\pdfbookmark[1]{Appendices}{apxcontents}%
	\@starttoc{\apptocfile}%
	\endgroup
}
\makeatother

\usepackage[dvipsnames]{xcolor}
\definecolor{mygreen}{rgb}{0.0, 0.65, 0.31} 
\definecolor{myred}{rgb}{1.0, 0.13, 0.32} 
\definecolor{mypurple}{rgb}{0.71, 0.49, 0.86} 
\definecolor{myblue}{rgb}{0.19, 0.55, 0.91} 
\definecolor{myorange}{rgb}{1.0, 0.56, 0.0} 

\usepackage[normalem]{ulem} 

\usepackage[
	colorlinks=true,
	urlcolor=violet,
	linkcolor=RoyalBlue,
	citecolor=RoyalBlue,
	hyperindex=true,
	linktocpage=true
]{hyperref}
\usepackage[nameinlink, capitalise]{cleveref}
\crefname{assumption}{assumption}{assumptions}
\Crefname{assumption}{Assumption}{Assumptions}

\allowdisplaybreaks[3] 

\def\Im{\mathop{\mathsf{Im}}}
\def\Re{\mathop{\mathsf{Re}}}

\def\tr{\mathop{\mathrm{tr}}}

\def\sgn{\mathop{\mathrm{sgn}}}

\begin{document}
\title{HyperDet Wavefunction: A Phase-Agnostic Ansatz for Strongly Correlated Systems}

\author{Xiaodong Hu}
\affiliation{Department of Materials Science and Engineering, University of Washington, Seattle, WA 98195, USA}

\author{Guan-Lin Lin}
\affiliation{Department of Physics, Boston College, Chestnut Hill, Massachusetts 02467, USA}

\author{Ying Ran}
\affiliation{Department of Physics, Boston College, Chestnut Hill, Massachusetts 02467, USA}

\author{Di Xiao}
\affiliation{Department of Materials Science and Engineering, University of Washington, Seattle, WA 98195, USA}
\affiliation{Department of Physics, University of Washington, Seattle, Washington 98195, USA}

\date{\today}

\begin{abstract}
	Competing phases of strongly correlated systems are often described by trial wavefunctions built on phase-specific assumptions. We propose the \emph{hyperdeterminant (HyperDet) wavefunction} as a phase-agnostic ansatz for both bosonic and fermionic quantum many-body systems exhibiting spontaneous symmetry-breaking, fractionalization, and/or topological order. The HyperDet structure emerges by fusing auxiliary fermionic parton Slater determinants into physical orbitals through a learnable \emph{fusion tensor $\mathcal F$}. Optimized using variational Monte Carlo, a single HyperDet architecture achieves overlaps exceeding $99.9\%$ with exact-diagonalization ground states for nearly all sampled parameters across fractional Chern insulator (FCI) phases and their neighboring competing phases in both bosonic and fermionic models. Beyond its variational expressiveness, the optimized fusion tensor also provides rich interpretability. We show that, (i) the singular-value spectra of the \emph{bipartite fusion matrix} can be used to track phase transitions; (ii) the reconstructed Bott indices reproduce the parton Chern numbers expected for the FCI states without parton-Hamiltonian input; and (iii) faithful lifts of physical translations to the parton degrees of freedom recover the expected translation fractionalization in the bosonic FCI without projective symmetry class assignment. These results establish HyperDet wavefunction as a promising framework for both variational ground-state searches and phase-diagram explorations across strongly correlated phases, and for providing theoretical insights by connecting the optimized wavefunction with parton microscopics and field-theory descriptions.
\end{abstract}

\maketitle


\section{Introduction}
As cornerstones of condensed matter physics, Laughlin's and Jain's trial wavefunctions~\cite{laughlin1983anomalous,jain1989incompressible,jain1989composite} have provided the foundation for our understanding of fractional quantum Hall (FQH) physics. However, extending this success to fractional Chern insulators (FCIs), lattice analogues of FQH states that can arise without an external magnetic field, has remained a theoretical challenge since their prediction~\cite{tang2011high,neupert2011fractional,sun2011nearly,sheng2011fractional,regnault2011fractional,xiao2011interface}. This challenge has since grown: FCI phases have now been realized in two-dimensional materials~\cite{cai2023signatures,park2023observation,xu2023observation,zeng2023thermodynamic,lu2024fractional}, and the tunability of these devices has rapidly exposed a rich landscape of competing phases nearby~\cite{park2023observation,lu2025extended,xu2025signatures,han2025signatures,han2026evidence,sun2026twist}. Describing this landscape calls for a similar wavefunction framework that addresses two challenges: (i) an accurate and robust ground-state description valid throughout an entire FCI phase, and (ii) a systematic, phase-agnostic account of the neighboring competing phases and the transitions between them.

The Laughlin and Jain ansatzes provide a natural starting point for challenge~(i), as they have been extensively examined in conventional Landau-level settings~\cite{haldane1985finite,yoshioka1983ground,jain1990theory,kamilla1996excitons,jain1997composite}. Considerable effort has gone into adapting them to lattice systems~\cite{qi2011generic,claassen2015position,jian2013crystal,wu2012gauge,wu2013bloch}, including wavefunction constructions based on ideal-band geometry and vortexability~\cite{parameswaran2012fractional,roy2014band,wang2021exact,ledwith2020fractional,ledwith2023vortexability}. These constructions, however, are built to mimic the lowest Landau level (LLL), and their accuracy can degrade rapidly away from the LLL limit. Even when ideal-band or vortexability conditions are satisfied, they do not fully determine the many-body physics. For example, the stability of the FCI can depend on features not captured by these conditions~\cite{shi2026effects,grushin2012enhancing,zhang2025beyond,sarkar2026similar}. A systematic understanding of the FCI phase therefore remains incomplete, motivating a flexible ansatz that can capture lattice effects beyond the Laughlin-Jain construction.

Challenge~(ii) seems to be even harder. The Laughlin and Jain ansatzes were designed for FQH systems and need not describe the phases neighboring an FCI, which may lack anyonic excitations and instead develop symmetry-breaking order. These neighboring phases can be gapped or gapless, with distinct symmetry and entanglement structures~\cite{chen2010local,wen2017colloquium,kitaev2006topological,levin2005string}, so an FCI can lose its topological order through many different routes. The competing phases include Fermi liquids~\cite{wang2024fractional,sheng2011fractional,park2023observation,lu2025extended}, Wigner crystals~\cite{han2026evidence}, charge-density waves~\cite{xie2021fractional,aronson2025displacement,wilhelm2021interplay,reddy2023fractional,qin2025stripe}, and more exotic ones such as anomalous composite Fermi liquids~\cite{anderson2024trion,goldman2023zero,dong2023composite}, anomalous Hall crystals~\cite{liu2025fractional,chen2026fractional}, and superconducting states~\cite{xu2025signatures,sun2026twist,seo2026family,nguyen2025hierarchy,hua2026multi}. The transitions themselves may involve deconfined quantum criticality beyond the Ginzburg--Landau paradigm~\cite{senthil2004deconfined,senthil2004quantum,wang2017deconfined}, as proposed for the continuous transition between a bosonic Laughlin state and a superfluid~\cite{barkeshli2014continuous,barkeshli2015continuous,song2024phase,lu2025continuous}. Describing this diversity of phases and transitions within a single wavefunction framework therefore remains a central challenge, both conceptually and computationally.

Recently, a \emph{hyperdeterminant (HyperDet)} structure~\cite{hu2024hyperdeterminants,hu2026composite} was recognized in the projective construction of composite-fermion states~\cite{jain1989incompressible,jain1989composite} for FCIs, with the Laughlin and Jain wavefunctions appearing as special cases. Within the FCI phase, these composite-fermion HyperDet wavefunctions provide an accurate description of the ground state beyond the LLL limit, but their quality degrades across FCI phase boundaries. The reason is structural. Their construction rests on a prescribed enlargement of the virtual Hilbert space motivated by flux attachment~\cite{pasquier1998dipole,read1994theory,shankar1999hamiltonian,murthy2003hamiltonian,murthy2012hamiltonian}, together with a fixed virtual-to-physical fusion map. Outside the FCI phase, these assumptions no longer provide an adequate description of the ground state.

In this work, we simultaneously address these two challenges by replacing the flux-attachment-based construction with a parton construction~\cite{lin2026hyperdeterminant}. Representing bosonic or fermionic physical particles through auxiliary fermionic partons, we promote the virtual-to-physical \emph{fusion map} to a fully learnable variational degree of freedom encoded in the \emph{fusion tensor $\mathcal F$}. We propose that the resulting \emph{HyperDet wavefunction $\psi_{\mathcal F}(\bm s)$} provides a phase-agnostic variational class for both bosonic and fermionic many-body systems, accommodating short-range or long-range entangled phases~\cite{wen2017colloquium,chen2010local,levin2005string,kitaev2006topological}, with or without anyonic excitations. This is demonstrated through both the expressiveness of the HyperDet wavefunction ansatz and its interpretability in terms of parton-level microscopics, including both the parton Chern numbers and the symmetry-fractionalization data.

We validate the expressiveness of the HyperDet wavefunction by benchmarking it against both bosonic and fermionic lattice models. By direct variational Monte Carlo (VMC) optimization of the fusion tensor itself, we can reach wavefunction overlaps above $99.95\%$ with ED ground states on clusters of up to 36 sites throughout the entire FCI phase in both cases. More surprisingly, the HyperDet wavefunction maintains this performance as we sweep across multiple nearby phases, with over $99.9\%$ overlap with ED ground states for almost all sampled points, including two distinct superfluid phases with condensates at the $\Gamma$ (SF@$\Gamma$) and M (SF@M) points for the bosonic model~\cite{lu2025continuous}, and ED-resolved candidate anomalous Hall crystal (AHC) and charge-density wave (CDW) phases for the fermionic model. These wide cross-phase sweeps substantiate that the same HyperDet architecture can represent states with fundamentally different topological order and symmetry-breaking orders.

Beyond this quantitative accuracy, we further show that the optimized fusion tensor can be used to probe phase transitions by exposing the internal fusion structure.  We examine the singular-value spectrum of the bipartite fusion matrix, obtained by reshaping the fusion tensor according to bipartitions of the parton species. In the bosonic model, we find a single dominant singular value characterizes the FCI phase, while additional singular values gain weight as the system enters the neighboring superfluid phases.
More broadly, sweeps across competing phases in both bosonic and fermionic models show that the redistribution of spectral weight tracks the physical phase boundaries, with changes consistent with the known first-order or continuous character of the transitions. These spectra thus serve as structural diagnostics of the optimized many-body states, revealing changes in their internal fusion structure across phase transitions without evaluating phase-specific physical observables.

The optimized fusion tensor is itself a direct source of parton-level topological data.  We extract the parton Chern numbers $(\mathcal C^{(1)},\mathcal C^{(2)})=(1,1)$ for the 1/2-filled bosonic FCI and $(\mathcal C^{(1)},\mathcal C^{(2)},\mathcal C^{(3)})=(1,1,1)$ for the 1/3-filled fermionic FCI, in agreement with parton topological contents of the semion topological order~\cite{arovas1984fractional,wen1995topological} and the Laughlin topological order~\cite{laughlin1983anomalous}, respectively. The HyperDet wavefunction also features a clear intrinsic gauge structure, which permits physical symmetries to be lifted onto the virtual parton legs, with stabilizer-valued ambiguities of the lifted group relations encoding parton symmetry-fractionalization data~\cite{wen2002quantum,essin2013classifying,barkeshli2019symmetry}. We recover the expected $\pi$-flux translation fractionalization without imposing any projective-symmetry data for the bosonic FCI model; however, no faithful translation virtual lift is found for the fermionic FCI within the tested representation. These results show that in our framework, the parton topology and, where faithful symmetry lifts exist, the symmetry-fractionalization data are learned from HyperDet optimization rather than imposed \emph{a priori}. This access to the underlying parton structure makes HyperDet more than a flexible variational ansatz: it also provides a bridge between numerical ground-state searches and parton field-theory descriptions.

The paper is organized as follows. We begin by introducing the HyperDet wavefunction and its variational fusion tensor in~Sec.~\ref{sec:hyperdet_ansatz}. In Sec.~\ref{sec:fusion_channels}, we show how factorizable fusion tensors reduce the HyperDet to a product of determinants, and develop the bipartite fusion matrix singular-value spectrum as a practical structural proxy for both factorizability and many-body phase transitions. To establish the basis for the parton analysis, Sec.~\ref{sec:determinant expansion} develops an exact weighted Tucker reconstruction of the fusion tensor and a corresponding exact determinant expansion of the wavefunction. We then formulate direct VMC optimization in Sec.~\ref{sec:vmc_optimization} and benchmark the ansatz across bosonic and fermionic FCIs and their competing phases in Sec.~\ref{sec:expressiveness}. Turning from accuracy to interpretation, Sec.~\ref{sec:interpretability} examines the parton topology and translation-fractionalization data encoded in the optimized fusion tensor. We conclude in Sec.~\ref{sec:discussion} with the physical implications, scalability, and future directions of the HyperDet framework. Technical derivations and extended numerical results are provided in the Appendix~\cite{Appendix}.

\section{HyperDet Wavefunction Ansatz}\label{sec:hyperdet_ansatz}

In this section we present a self-contained introduction of the HyperDet wavefunction ansatz for $N$ physical particles occupying $N_s$ orbitals. The physical single-particle Hilbert space is spanned by $\ket{\phi_s}$ with $s = 1, \dots, N_s$. We focus on the hard-core case for simplicity, where $N\leq N_s$. For systems with additional internal degrees of freedom, such as spin, or with soft bosons that permit multiple occupancy, these degrees of freedom can be incorporated by enlarging the single-particle orbital index $N_s$.

We introduce $m$ copies of fermionic partons as the auxiliary degrees of freedom and fill exactly the same number of parton orbitals $N_\alpha=N$ for each species $\alpha=1,2,\ldots,m$. We assume that the parton species are mutually disentangled, and the fermionic partons within each species $\alpha$ can be described by a single Slater determinant
\begin{equation}\label{eq:parton single Slater determinant}
	|\Psi^{(\alpha)}\rangle = \sum_{P}(-1)^{P}\bigotimes_{\ell=1}^{N}|\varphi_{P(\ell)}^{(\alpha)}\rangle_\ell,
\end{equation}
where the sum runs over all permutations $P\in S_N$ of the $N$ occupied parton orbitals, and $\ell\in\{1,2,\ldots,N\}$ is a first-quantized particle-slot index: $|\varphi_{P(\ell)}^{(\alpha)}\rangle_\ell$ places the $P(\ell)$-th parton orbital into the $\ell$-th particle slot. Unlike the projective construction of parton/ancilla qubits in Ref.~\cite{jain1989incompressible,zhang2020pseudogap,zhou2025variational,wen1999projective}, where auxiliary degrees of freedom are defined on local physical sites, the $N$ occupied fermionic parton orbitals $|\varphi^{(\alpha)}_{P(\ell)}\rangle_\ell$ introduced here are not necessarily local.

The full parton mean-field state is the direct product over species $|\Psi^{\text{parton}} \rangle = \bigotimes_{\alpha=1}^m|\Psi^{(\alpha)}\rangle$. We introduce a canonical isomorphism $U_\pi^{\text{slot}}$ to reorder this species-major state into particle-slot-major form~\cite{Appendix}:
\begin{equation}
	U_\pi^{\text{slot}}|\Psi^{\text{parton}}\rangle = \sum_{P_1,\ldots,P_m}\Big[\prod_{\alpha}(-1)^{P_\alpha}\Big]\bigotimes_{\ell=1}^N\bigotimes_{\alpha=1}^m|\varphi_{P_{\alpha}(\ell)}^{(\alpha)}\rangle_\ell.\label{eq:particle-slot-major parton mean-field state}
\end{equation}
All relative fermionic signs are already accounted for in the first-quantized Slater antisymmetrizers.

Because both the parton mean-field state and the physical states are arranged in particle-slot order, we can introduce a particle-slot-wise fusion map $\hat F_\ell$ to fuse the parton species together and map them to the $\ell$-th physical state $|\phi_{s_\ell}\rangle$. In terms of the occupied parton orbital basis $\{|\varphi^{(\alpha)}_{\ell_\alpha}\rangle\}$, it reads
\begin{equation}\label{eq:fusion map in parton orbital basis}
	\hat F_\ell = \sum_{\ell_1,\ldots,\ell_m} \mathcal F_{s_\ell; \ell_1,\ldots,\ell_m} |\phi_{s_\ell}\rangle\langle\varphi^{(1)}_{\ell_1}\cdots\varphi^{(m)}_{\ell_m}|
\end{equation}
where the fusion tensor coefficients are
\begin{equation}\label{eq:fusion tensor coefficients in parton orbital basis}
	\mathcal F_{s_\ell; \ell_1,\ldots,\ell_m} = \langle\phi_{s_\ell}|\hat F_{\ell}|\varphi^{(1)}_{\ell_1}\cdots\varphi^{(m)}_{\ell_m}\rangle.
\end{equation}
The generically nonlocal fusion map $\hat F_\ell$ differs significantly from the on-site Gutzwiller projection~\cite{gutzwiller1963effect,gutzwiller1965correlation} widely used in studies of high-$T_c$ superconductivity and spin liquids~\cite{wen2002quantum,ran2007projected,anderson1987resonating,zhang1988renormalised,edegger2007gutzwiller}. This freedom is essential and should be distinguished from the possible nonlocality of the occupied parton orbitals.

Fusing the parton mean-field state in Eq.~\eqref{eq:particle-slot-major parton mean-field state} according to Eq.~\eqref{eq:fusion map in parton orbital basis} gives the physical state $\textstyle|\Psi^{\text{phys}}\rangle = \big(\bigotimes_{\ell=1}^N \hat F_\ell\big) U_\pi^{\text{slot}}|\Psi^{\text{parton}}\rangle$. Its amplitude for the physical configuration $|\bm s\rangle=|\phi_{s_1}\cdots\phi_{s_N}\rangle$ takes the mathematical form of the \emph{anchored combinatorial hyperdeterminant}~\cite{barvinok1995new} constructed from the fusion tensor coefficients:
\begin{align}
	\langle\bm s|\Psi^{\text{phys}}\rangle &= \sum_{P_1,\cdots,P_m}\Big[\prod_{\alpha=1}^{m}(-1)^{P_\alpha}\Big] \prod_{\ell=1}^{N} \mathcal F_{s_\ell; P_1(\ell),\ldots,P_m(\ell)}\nonumber\\
	& \equiv \mathop{\mathsf{HyperDet}}[\mathcal F_{\bm s}].\label{eq:hyperdeterminant wavefunction amplitude}
\end{align}

Equation~\eqref{eq:hyperdeterminant wavefunction amplitude} is the HyperDet wavefunction ansatz.
For $m=1$ it reduces to an ordinary Slater determinant. Every additional species introduces another independent antisymmetrization. Therefore, unlike ordinary Slater determinants, HyperDet can describe both symmetric/bosonic (even-$m$) and antisymmetric/fermionic (odd-$m$) wavefunctions~\cite{lin2026hyperdeterminant}. This statistical theorem explicitly separates exchange statistics from the detailed fusion maps: auxiliary particles are always fermions, but the number of independently antisymmetrized species fixes the exchange parity.
The physical one-body basis used to label $|\bm s\rangle$ is likewise unrestricted: it can be real-space or momentum-space physical orbitals, being either localized or non-local. This freedom facilitates the practical use of the HyperDet wavefunction.

The full hyper-rectangular fusion tensor $\mathcal F$ in Eq.~\eqref{eq:fusion tensor coefficients in parton orbital basis} is of size $N_s\times N^m$ since there are $N_s$ possible physical orbitals. Each physical configuration $|\bm s\rangle$ of $N$ occupied physical orbitals slices out a hypercubic, configuration-restricted fusion tensor $\mathcal F_{\bm s}$ of size $N^{m+1}$ for evaluating the hyperdeterminant in Eq.~\eqref{eq:hyperdeterminant wavefunction amplitude}. This process is illustrated in Fig.~\ref{fig:HyperDet_illustration}.
\begin{figure}[btp!]
	\centering
	\includegraphics[width=1.0\linewidth]{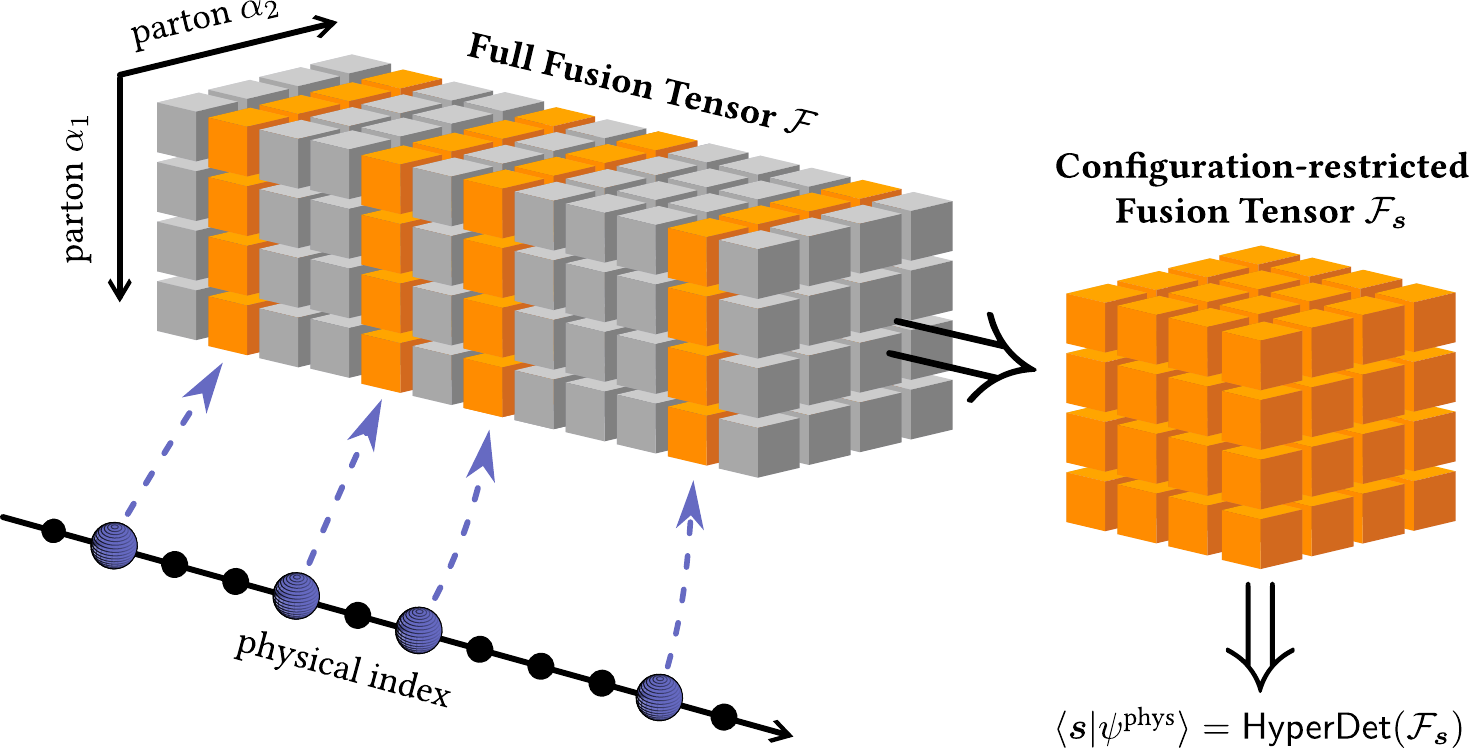}
	\caption{Illustration of a HyperDet wavefunction with $N=4$ filled physical particles and $m=2$ parton species for a given configuration $|\bm s\rangle$ on a $12$-site graph (obtained by flattening the 2D lattice, for example). Physical degrees of freedom live at the bottom and parton degrees of freedom live above. Each parton species occupies the same number of orbitals as there are physical particles: $N_\alpha=N_\beta=N=4$. Each configuration $|\bm s\rangle$ specifies the $N\times N_\alpha\times N_\beta=4^3$ hypercubic, configuration-restricted fusion tensor $\mathcal F_{\bm s}$ sliced from the full hyper-rectangular fusion tensor $\mathcal F$ (which contains all variational parameters) of size $N_s\times N_\alpha\times N_\beta=12\times 4^2$.}
	\label{fig:HyperDet_illustration}
\end{figure}

\section{Factorization Simplification and SVD Spectrum of Bipartite Fusion Matrix}\label{sec:fusion_channels}
To appreciate the many-body structure and correlations encoded in the fusion map, it is helpful to switch from the occupied parton orbital basis $\{|\varphi^{(\alpha)}_{\ell_\alpha}\rangle\}$ to a local parton basis $\{|\phi^{(\alpha)}_u\rangle\}$ labeled by the vertices $u$ of a graph $\mathcal V$.  We take the same vertices to label the physical orbitals, and refer to each vertex as a ``site''.
Locality is defined with respect to this orbital graph; for localized real-space orbitals, this reduces to ordinary spatial locality.

The occupied orbitals are generally linear combinations of these local basis states, $|\varphi^{(\alpha)}_\ell\rangle = \sum_{u\in\mathcal V} \varphi^{(\alpha)}_\ell(u) |\phi^{(\alpha)}_u\rangle$. After the basis transformation, the fusion tensor in the parton orbital basis takes the form:
\begin{align}\label{eq:fusion tensor connection relation}
	\mathcal F_{s_\ell;\ell_1,\ldots,\ell_m} &=  \sum_{u_1,\ldots,u_m\in\mathcal V} \mathcal F^{\text{loc}}_{s_\ell; u_1,\ldots,u_m} \prod_{\alpha=1}^{m}\varphi_{\ell_\alpha}^{(\alpha)}(u_{\alpha}),
\end{align}
where we introduce the local fusion tensor coefficients
\begin{equation}\label{eq:fusion tensor coefficients in parton local basis}
	\mathcal F^{\text{loc}}_{s_\ell; u_1,\ldots,u_m} = \langle\phi_{s_\ell}|\hat F_\ell|\phi^{(1)}_{u_1}\cdots\phi^{(m)}_{u_m}\rangle.
\end{equation}

One particular situation is when the fusion map is strictly \emph{on-site} at the occupied physical orbitals
\begin{equation}\label{eq:on-site fusion tensor}
	\mathcal F^{\text{loc}}_{s_\ell;u_1,\ldots,u_m}=\prod_\alpha\delta_{s_\ell,u_\alpha}.
\end{equation}
In this case, Eq.~\eqref{eq:fusion tensor connection relation} factorizes into a product of matrices over physical and parton virtual legs:
\begin{equation}\label{eq:factorizable fusion tensor}
	\mathcal F^{\text{fac}}_{s_\ell;\ell_1,\ldots,\ell_m} =\prod_{\alpha=1}^{m}\varphi_{\ell_\alpha}^{(\alpha)}(s_\ell),
\end{equation}
and the evaluation of the HyperDet wavefunction reduces to the simple product-of-determinants form:
\begin{equation}\label{eq:HyperDet of factorizable fusion tensor}
	\mathop{\mathsf{HyperDet}}[\mathcal F^{\text{fac}}_{\bm s}] = \prod_{\alpha=1}^{m}\det\Big[\varphi^{(\alpha)}_{j}\big(s_{\ell}\big)\Big].
\end{equation}
The on-site fusion map in Eqs.~\eqref{eq:on-site fusion tensor} and \eqref{eq:factorizable fusion tensor} corresponds exactly to the conventional local parton-splitting rules: in second quantization, when working in real-space orbitals for both physical and parton degrees of freedom, the on-site fusion map exactly reproduces the form widely used in the literature~\cite{wen1999projective,balram2018parton,balram2019parton,lu2012symmetry,gu2010grassmann}
\begin{equation}\label{eq:on-site fusion map in second quantized form}
	c_{\bm r}=\prod_{\alpha=1}^m f^{(\alpha)}_{\bm r}.
\end{equation}
A typical example is the $\nu=1/m$ Laughlin wavefunction $\Psi_{1/m}(\bm r)=\prod_{i<j}(z_{i}-z_{j})^{m}e^{-\sum_i |z_i|^2}$.  In the parton description of FQH physics~\cite{jain1989composite,jain1989incompressible,jain1990theory}, each parton species occupies its own integer-filled Landau level, forming a product of parton Slater determinants $\Psi^{\text{parton}}(\{z^{(1)}\}, \dots, \{z^{(m)}\}) = \prod_{\alpha}\det[\varphi^{(\alpha)}_j(z^{(\alpha)}_k)]$, where $\varphi^{(\alpha)}$ describes orbitals in the lowest Landau level of parton species $\alpha$. The on-site fusion map here simply supplies the coordinate-identification rule $z^{(\alpha)}_\ell=z_\ell$ for each parton species, while the zeros, phases, and correlations of the wavefunction are governed by the choice of occupied parton orbitals.

The factorization in Eq.~\eqref{eq:factorizable fusion tensor} is a special case. As demonstrated in Sec.~\ref{sec:expressiveness}, optimized fusion tensors need not admit this factorization, and the resulting HyperDet wavefunctions cannot be reduced to products of determinants~\cite{hu2024hyperdeterminants,lin2026hyperdeterminant}, reflecting a generic $\textsf{\#P-Hard}$ evaluation complexity~\cite{barvinok1995new,gurvits2005complexity,hillar2013most}. The full HyperDet therefore provides a more general ansatz, with the traditional product-of-determinants construction contained as a restricted limit. Allowing \emph{non-factorizable} fusion tensors is central to our approach: it gives the ansatz access to states that the conventional construction cannot represent.

This parton-to-physical fusion complexity can be further quantified through the singular value decomposition (SVD) of the bipartite fusion matrix, or the species-$\alpha$ unfolding of the fusion tensor:
\begin{equation}\label{eq:bipartite fusion matrix}
	\mathcal F_{\ell_\alpha\mid\ell_{\bar\alpha}}(s_\ell) := \mathcal F_{s_\ell;\ell_\alpha\mid\ell_1,\ldots,\widehat{\ell_\alpha},\ldots,\ell_m},
\end{equation}
constructed by squashing the original fusion tensor under virtual-leg bipartitions: the selected parton species $\ell_\alpha$, and the other parton species $\ell_{\bar\alpha}=\{\ell_1,\ldots,\ell_{\alpha-1}, \ell_{\alpha+1},\ldots,\ell_m\}$. For the $m=2$ two-parton case, this bipartite fusion matrix is exactly the physical-leg-restricted fusion tensor itself. For general $m$, its species-resolved SVD is
\begin{equation}\label{eq:SVD of bipartite fusion matrix}
	\mathcal F_{\ell_\alpha\mid\ell_{\bar\alpha}}(s_\ell) = \sum_{c=1}^{r^{(\alpha)}_{s_\ell}} \lambda^{(\alpha)}_c(s_\ell) [\bm U_{s_\ell}]_{\ell_\alpha,c} [\bm V_{s_\ell}]_{\ell_{\bar\alpha},c}^*.
\end{equation}
Here $r^{(\alpha)}(s_\ell)$ is the rank of the species-$\alpha$ unfolding. A fusion tensor factorizes completely if every species unfolding has rank one.

Given this understanding, we propose that the normalized, physical-leg averaged SVD spectrum of the bipartite fusion matrix
\begin{equation}\label{eq:SVD spectrum of bipartite fusion matrix}
	\{\lambda^{(\alpha)}\} = \Big\langle\mathsf{Spec}[\mathcal F_{\ell_\alpha\mid\ell_{\bar\alpha}}(s_\ell)]\Big\rangle_{s_\ell}
\end{equation}
can be used to characterize the fusion complexity of the underlying many-body state. For phase transitions where fusion tensor undergoes structural changes, we expect the redistribution of the bipartite fusion matrix SVD spectrum to provide an internal structural diagnostic for many-body phase transitions.  This will be shown in Sec.~\ref{sec:expressiveness}

We note that the SVD spectrum of the bipartite fusion matrix is invariant only under the unitary part of the HyperDet gauge group $\mathrm{GL}(N;\mathbb{C})^m\rtimes S_m$ (see below) and therefore depends on which representative is chosen on a gauge orbit unless the non-unitary part is fixed.  We remove this ambiguity with the balanced gauge-fixing procedure detailed in the Appendix~\cite{Appendix}.

\section{On-site Weighted Tucker Reconstruction and Determinant Expansion}\label{sec:determinant expansion}
Even when the fusion tensor does not factorize, the species-resolved SVDs in Eq.~\eqref{eq:SVD of bipartite fusion matrix} can be used to express the HyperDet wavefunction as a sum of products of determinants. We first construct this exact expansion and then show how it identifies reconstructed parton orbitals. These orbitals provide the basis for the parton analysis in Sec.~\ref{sec:interpretability}.

For each physical orbital $s_\ell$, we use the nonzero singular values and corresponding left singular vectors of the bipartite fusion matrix in Eq.~\eqref{eq:SVD of bipartite fusion matrix} to define the species-$\alpha$ matrix
\begin{equation}\label{eq:occupied-orbital amplitudes matrix A}
	[\bm A^{(\alpha)}]_{(s_\ell,c),\ell_\alpha} = \sqrt{\lambda^{(\alpha)}_c(s_\ell)}\,[\bm U_{s_\ell}^{(\alpha)}]_{\ell_\alpha,c}.
\end{equation}
Here $c=1,\ldots,r_{s_\ell}^{(\alpha)}$ labels the retained SVD channels. Stacking the rows $(s_\ell,c)$ over all $N_s$ physical orbitals gives an $M_\alpha\times N$ matrix, with $M_\alpha=\sum_{s_\ell}r_{s_\ell}^{(\alpha)}$. Thus $\bm A^{(\alpha)}$ is determined by the full $N_s\times N^m$ hyper-rectangular fusion tensor and is independent of the physical configuration $|\bm s\rangle$ at which the wavefunction is evaluated.

The matrices $\{\bm A^{(\alpha)}\}$ reconstruct the fusion tensor through an exact weighted Tucker decomposition~\cite{tucker1966some,kolda2009tensor}:
\begin{equation}\label{eq:simultaneous fusion reconstruction}
	\mathcal F_{s_\ell;\ell_1,\ldots,\ell_m}=\sum_{c_1,\ldots,c_m}\mathcal T_{s_\ell;c_1,\ldots,c_m}\prod_{\alpha=1}^{m} [\bm A^{(\alpha)}]_{(s_\ell,c_\alpha),\ell_\alpha},
\end{equation}
where the \emph{weighted Tucker core} is
\begin{equation}\label{eq:weighted Tucker core}
	\mathcal T_{s_\ell;c_1,\ldots,c_m}=\sum_{\ell_1,\ldots,\ell_m=1}^{N}\mathcal F_{s_\ell;\ell_1,\ldots,\ell_m}\prod_{\alpha=1}^{m}\frac{[\bm U_{s_\ell}^{(\alpha)}]^*_{\ell_\alpha,c_\alpha}}{\sqrt{\lambda_{c_\alpha}^{(\alpha)}(s_\ell)}}.
\end{equation}
The identity follows because projecting onto the support of any species unfolding leaves the fusion tensor unchanged, and projections acting on different species commute (see Appendix~\cite{Appendix} for the derivation). Retaining all nonzero singular values therefore gives an exact reconstruction.

Substituting Eq.~\eqref{eq:simultaneous fusion reconstruction} into the HyperDet definition in Eq.~\eqref{eq:hyperdeterminant wavefunction amplitude} gives the determinant expansion
\begin{align}\label{eq:HyperDet determinant expansion}
	\Psi_{\mathcal F}^{\text{phys}}(\bm s) &= \sum_{\{c_\alpha(\ell)\}}\left[\prod_{\ell=1}^{N}\mathcal T_{s_\ell;c_1(\ell),\ldots,c_m(\ell)}\right]\nonumber\\
	&\qquad\qquad\prod_{\alpha=1}^{m}\det\left[[\bm A^{(\alpha)}]_{(s_\ell,c_\alpha(\ell)),j}\right]_{\ell,j=1}^{N}.
\end{align}
The sum runs over all \emph{channel configurations}, with one independent channel index $c_\alpha(\ell)\in\{1,\ldots,r^{(\alpha)}_{s_\ell}\}$ for each physical slot $\ell$ and parton species $\alpha$:
\begin{equation}\label{eq:channel configuration sum}
	\sum_{\{c_\alpha(\ell)\}}\equiv\prod_{\alpha=1}^{m}\prod_{\ell=1}^{N}\sum_{c_\alpha(\ell)=1}^{r^{(\alpha)}_{s_\ell}}.
\end{equation}
For a fixed channel configuration, each determinant is formed by selecting the rows $(s_1,c_\alpha(1)),\ldots,(s_N,c_\alpha(N))$ of the same matrix $\bm A^{(\alpha)}$ and retaining all $N$ columns. The core tensors $\mathcal T_s$ thus weight and couple these channel choices across species.

The determinant expansion now gives $\bm A^{(\alpha)}$ a precise orbital interpretation. If its column rank were smaller than $N$, every determinant for species $\alpha$ in Eq.~\eqref{eq:HyperDet determinant expansion} would vanish, and hence so would the physical wavefunction. A nonzero wavefunction therefore requires $\mathop{\mathrm{rank}}\bm A^{(\alpha)}=N$ for every species. Its $N$ linearly independent columns can then be regarded as occupied-orbital amplitudes in the reconstructed one-body space $\mathcal H_{\text{rec}}^{(\alpha)}=\bigoplus_{s_\ell}\mathbb C^{r_{s_\ell}^{(\alpha)}}$, whose basis states are labeled by the composite coordinate $(s_\ell,c)$. After orthonormalization, these columns define a single $N$-particle Slater determinant for each species. The determinants in Eq.~\eqref{eq:HyperDet determinant expansion} are configuration amplitudes of these reconstructed Slater states, up to configuration-independent normalization factors. Their occupied subspaces are the starting point for the parton analysis in Sec.~\ref{sec:parton_chern_numbers}.

In this representation, the weighted Tucker core $\mathcal T_{s_\ell}$, or precisely the induced virtual-to-physical map $\mathcal G\equiv\bigotimes_{s_\ell} \mathcal G_{s_\ell}$ where
\begin{align}\label{eq:Tucker-core fusion}
	&\mathcal G_{s_\ell} := \ket{0}_{\text{phys}}\bra{0}_{\text{aux}}\nonumber\\
	&\qquad + \ket{1}_{\text{phys}}\sum_{c_1,\ldots,c_m}\mathcal T_{s_\ell;c_1,\ldots,c_m}\bra{(s_\ell,c_1),\ldots,(s_\ell,c_m)}_{\text{aux}},
\end{align}
fuses auxiliary channel states attached to the same physical orbital $s_\ell$. Fusion is therefore \emph{on-site in the reconstructed channel space}. This does not require the original fusion map to be on-site in the microscopic parton basis of Eq.~\eqref{eq:fusion tensor coefficients in parton local basis}. A relative-coordinate rewriting makes the distinction explicit. On a periodic lattice, write $s_\ell\equiv(\bm R_{s_\ell},a_{s_\ell})$ and introduce $\delta_\alpha\equiv(\bm d_\alpha,b_\alpha)$, where $\bm d_\alpha$ is a cell displacement and $b_\alpha$ labels an orbital in the displaced cell. The shorthand
\begin{equation*}
	s_\ell\oplus\delta_\alpha\equiv(\bm R_{s_\ell}+\bm d_\alpha,b_\alpha)=u_\alpha
\end{equation*}
defines a bijection between the microscopic parton coordinates $u_\alpha$ and the relative-coordinate labels $\delta_\alpha$ for each fixed $s_\ell$. Reindexing Eq.~\eqref{eq:fusion tensor connection relation} gives
\begin{equation}\label{eq:relative-coordinate re-writing of fusion tensor relation}
	\mathcal F_{s_\ell;\ell_1,\ldots,\ell_m} = \sum_{\delta_1,\ldots,\delta_m} \mathcal F^{\text{loc,rel}}_{s_\ell;\delta_1,\ldots,\delta_m} \prod_{\alpha=1}^{m} \varphi^{(\alpha)}_{\ell_\alpha}(s_\ell\oplus\delta_\alpha),
\end{equation}
with
\begin{equation}\label{eq:relative-coordinate resolved local fusion tensor}
	\mathcal F^{\text{loc,rel}}_{s_\ell;\delta_1,\ldots,\delta_m} \equiv \langle\phi_{s_\ell}|\hat F_\ell|\phi^{(1)}_{s_\ell\oplus\delta_1},\ldots,\phi^{(m)}_{s_\ell\oplus\delta_m}\rangle.
\end{equation}
Each $\delta_\alpha$ runs over all cell displacements and target orbitals. These labels can be treated as internal flavors attached to $s_\ell$, even when the corresponding microscopic parton coordinates are far apart. The extension to a generic graph is similar: rewrite using a distance shell and a vertex label within that shell. The Tucker channels in Eq.~\eqref{eq:weighted Tucker core} play the same role as the internal flavors under the on-site rewriting in Eq.~\eqref{eq:relative-coordinate re-writing of fusion tensor relation} and Eq.~\eqref{eq:relative-coordinate resolved local fusion tensor}, but their dimensions are fixed by the unfolding ranks $r_{s_\ell}^{(\alpha)}$, rather than by the number of relative-coordinate labels. These dimensions are minimal among exact Tucker representations of the fixed fusion tensor: any such representation must have a channel dimension at least as large as the corresponding unfolding rank.

The conventional product-of-determinants form Eq.~\eqref{eq:HyperDet of factorizable fusion tensor} is recovered when every species unfolding has rank one. In that limit, the channel sum contains a single term, and the residual scalar core can be absorbed into the corresponding rows of one species matrix. At higher ranks, the weighted Tucker core $\mathcal T_{s_\ell}$ in Eq.~\eqref{eq:HyperDet determinant expansion} couples multiple channels and retains the nonfactorizable fusion structure. Its species indices remain independent and need not collapse to a single shared channel index. The auxiliary channel dimensions may grow with particle number, and the number of channel assignments in Eq.~\eqref{eq:channel configuration sum} is generally exponential in system size. The reconstruction thus establishes an exact parton interpretation for general fusion tensors while retaining the computational complexity of the full HyperDet.

\section{Variational Monte Carlo Optimization of Fusion Tensor}\label{sec:vmc_optimization}
The core ingredient for evaluating the HyperDet wavefunction Eq.~\eqref{eq:hyperdeterminant wavefunction amplitude} is the hypercubic, configuration-restricted fusion tensor $\mathcal F_{\bm s}$ sliced from the full hyper-rectangular fusion tensor $\mathcal F$. We thus parameterize physical states by this $\mathbb C$-valued fusion tensor, and directly take the entries of $\mathcal F$ as the variational parameters in our HyperDet wavefunction ansatz, denoted as $|\Psi^{\text{phys}}_{\mathcal F}\rangle$ from now on.

Given a many-body Hamiltonian $H$ at filling $\nu$, the best fusion tensor for a fixed number of parton species $m$ can be optimized using the standard time-dependent variational principle (TDVP). For each physical configuration $|\bm s\rangle$, TDVP determines the infinitesimal fusion-tensor update $\mathcal F\to\mathcal F+\delta\mathcal F$ by aligning it with the infinitesimal imaginary-time evolution $|\Psi^{\text{phys}}_{\mathcal F+\delta\mathcal F}\rangle \simeq e^{-\delta\tau H}|\Psi^{\text{phys}}_{\mathcal F}\rangle$, with $\delta\tau\ll1$ interpreted as the learning rate. Since our HyperDet wavefunction ansatz is holomorphic with respect to the fusion tensor~\footnote{This is \emph{not} the holomorphicity condition in the lowest Landau level with respect to the physical coordinates.}, its infinitesimal update satisfies the well-known stochastic-reconfiguration (SR)~\cite{sorella1998green,sorella2001generalized} rule
\begin{equation}\label{eq:SR update rule}
	\delta\mathcal F = -\delta\tau \bm{\mathcal S}^{-1}\bm f,
\end{equation}
with the force vector $\bm f$ and Fisher information matrix $\bm{\mathcal S}$~\cite{park2020geometry}:
\begin{align}\label{eq:force vector and Fisher information matrix}
	f_{\mu} &= \langle\!\langle\hat{\mathcal O}_\mu^*H\rangle\!\rangle - \langle\!\langle\hat{\mathcal O}_\mu^*\rangle\!\rangle\langle\!\langle H\rangle\!\rangle\nonumber\\[0.375em]
	\mathcal S_{\mu\nu} &= \langle\!\langle\hat{\mathcal O}_\mu^*\hat{\mathcal O}_\nu\rangle\!\rangle - \langle\!\langle\hat{\mathcal O}_\mu^*\rangle\!\rangle\langle\!\langle\hat{\mathcal O}_\nu\rangle\!\rangle
\end{align}
Here $\langle\!\langle\hat A\rangle\!\rangle\equiv\langle\Psi^{\text{phys}}_{\mathcal F}|\hat A|\Psi^{\text{phys}}_{\mathcal F}\rangle/\langle\Psi^{\text{phys}}_{\mathcal F}|\Psi^{\text{phys}}_{\mathcal F}\rangle$ denotes the normalized expectation value. The subscripts $\mu,\nu=1,2,\ldots,N_s\times N^m$ label the flattened indices of the fusion tensor elements, and the logarithmic derivative $\hat{\mathcal O}_\mu=\sum_{\bm s}|\bm s\rangle\Big[\partial_{\mathcal F_\mu}\ln\big(\Psi^{\text{phys}}_{\mathcal F}(\bm s)\big)\Big]\langle\bm s|$. There is one additional subtlety: the gauge redundancy inherited by our HyperDet ansatz (see below) may affect the SR updates. We address this by projecting the raw SR updates onto the Frobenius-horizontal complement of the gauge tangent space, as detailed in the Appendix~\cite{Appendix}.

Resolving the symmetry sector is helpful for a topologically degenerate ground-state manifold. For a system with an abelian symmetry group $G$ whose target sector satisfies $U_g|\Psi^{\text{phys}}\rangle=\chi(g)|\Psi^{\text{phys}}\rangle$, the VMC sampling can be modified by adding a positive semi-definite penalty $Q_\chi=(U_g-\chi(g)I)^\dagger(U_g-\chi(g)I)$,
\begin{equation}\label{eq:sector-pinning penalty}
	H\mapsto H_\kappa \equiv H + \kappa Q_{\chi},\quad \kappa\geq0.
\end{equation}
Alternatively, one can select this sector by applying the irrep projector
\begin{equation}\label{eq:irrep-projector}
	|\Psi^{\text{phys}}_{\mathcal F;\chi}\rangle := P_\chi |\Psi^{\text{phys}}_{\mathcal F}\rangle,\quad P_\chi = \dfrac{1}{|G|}\sum_{g\in G}\chi(g)^* U_g.
\end{equation}
Sampling the sector-projected state $|\Psi^{\text{phys}}_{\mathcal F;\chi}\rangle$ increases the cost of each VMC update by a factor of $|G|$, but the improvement in VMC energy convergence is substantial. It also prevents symmetry-sector mixing, making the comparison with an individual ED eigenstate meaningful when the ground-state manifold is nearly degenerate.

In practice, for both translation-invariant lattice models below, we implemented irrep-projected sampling to examine the expressiveness of our HyperDet wavefunction ansatz. Optimizer SPRING~\cite{goldshlager2024kaczmarz} is chosen for VMC optimization with hyperparameter settings given in the Appendix~\cite{Appendix}. We performed extensive numerical experiments with different initializations, hyperparameter settings and optimizer choices, and found that the performance of our HyperDet wavefunction ansatz is robust and insensitive to these training details.

\section{Expressiveness of the HyperDet Wavefunction Ansatz}\label{sec:expressiveness}
\subsection{Phase Transitions near the $\nu=1/2$ Bosonic FCI}
To showcase the expressiveness of our HyperDet wavefunction ansatz, we first consider a hard-core Bose-Hubbard model on the Haldane honeycomb lattice with extended hoppings~\cite{wang2011fractional,lu2025continuous}
\begin{equation}\label{eq:Bose-Hubbard Hamiltonian}
	H = -\sum_{\langle i,j\rangle}t b_{i}^{\dagger}b_{j}- \sum_{\langle\langle i,j\rangle\rangle}t' e^{i\phi_{i,j}}b_{i}^{\dagger}b_{j}- \sum_{\langle\langle\langle i,j\rangle\rangle\rangle}t'' b_{i}^{\dagger}b_{j}+\text{h.c.}
\end{equation}
where $b_i^\dagger/b_i$ creates/annihilates a hard-core boson. At half band filling $\nu=1/2$, the flatband parameters $t=1$ (as the energy unit), $t'=-0.6$, $\phi_{ij}=\pm0.4\pi$, and $t''=-0.58$ stabilize a bosonic FCI~\cite{wang2011fractional} even without explicit interaction terms. This interaction-driven physics occurs because the hard-core constraint is imposed in Eq.~\eqref{eq:Bose-Hubbard Hamiltonian}. Ref.~\cite{lu2025continuous} further focused on the regime obtained by sweeping the extended-hopping strength $t''\in[-1,0]$, where a DMRG study locates a first-order SF@M--FCI transition near $t''\approx-0.7$ and a continuous FCI--SF@$\Gamma$ transition near $t''\approx-0.495$.
\begin{figure}[btp!] 
	\centering
	\includegraphics[width=1.0\linewidth]{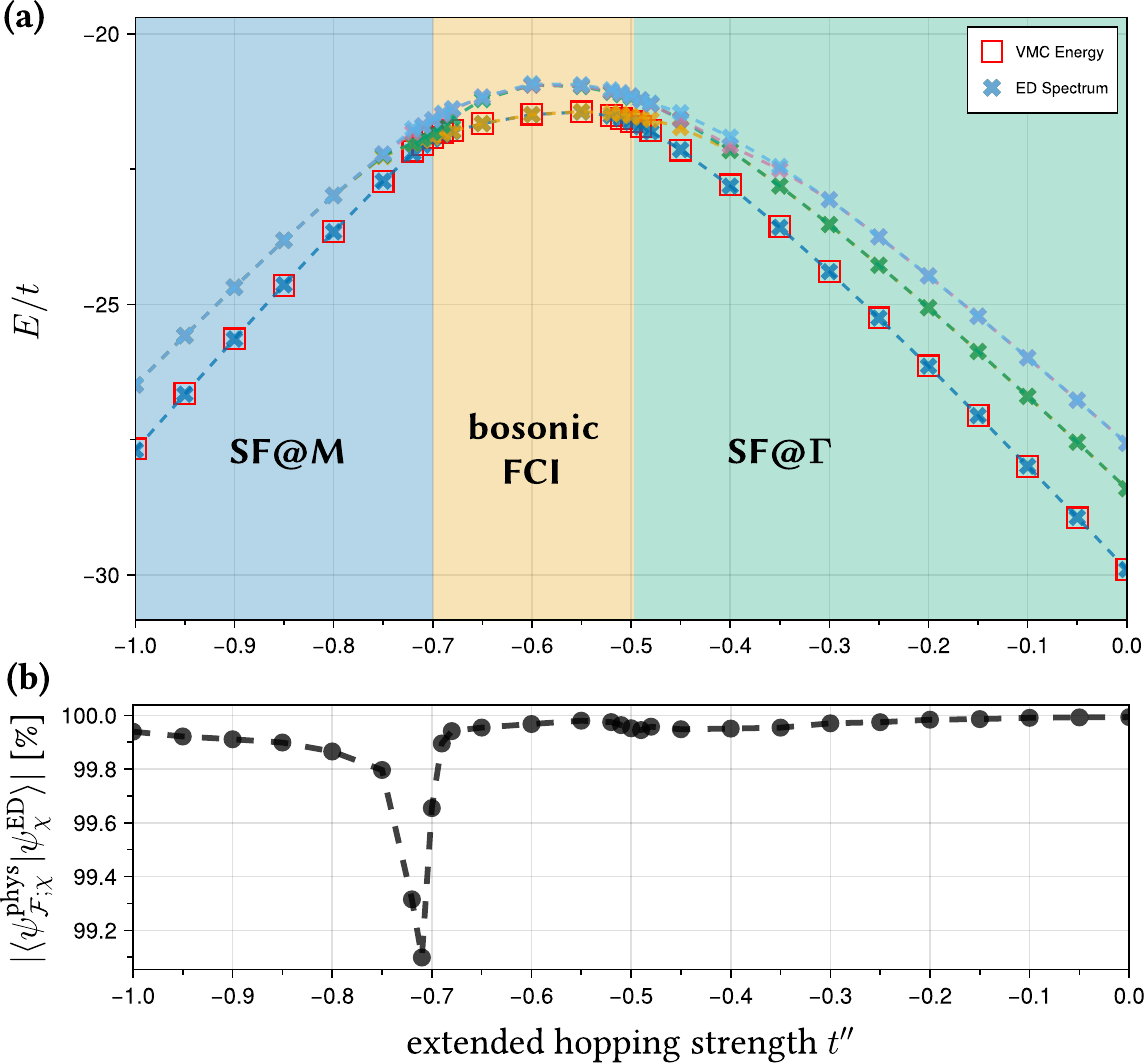}
	\caption{Benchmark of the irrep-projected HyperDet wavefunction ansatz against ED results on a $3\times6\times2$ sample ($2$ is the sublattice) for model Eq.~\eqref{eq:Bose-Hubbard Hamiltonian} at half-filling $\nu=1/2$ as the extended-hopping strength $t''\in[-1,0]$ is varied. (a) Comparison of the lowest twenty ED energies across all sectors and the irrep-projected variational energy in the global ground-state sector. The three phases SF@M, bosonic FCI, and SF@$\Gamma$ are shaded with their phase boundaries adapted from the DMRG calculations in Ref.~\cite{lu2025continuous}. (b) Wavefunction overlap with the ED ground states. Note: the prominent drop in the overlap at $t''\approx-0.71$ is due to the closing of the direct gap in the chosen sector $(0,3)$. Such a tiny-gap region is notoriously difficult for VMC to resolve~\cite{zhang2026projected}. The resulting overlap drop is thereby an artifact of VMC; it does not compromise the expressiveness of our HyperDet wavefunction ansatz.}
	\label{fig:bosonic_combined_phase_diagram}
\end{figure}

We perform VMC optimization of the fusion tensor with two parton species, $m=2$, for this bosonic model by sweeping $t''\in[-1,0]$ following Ref.~\cite{lu2025continuous}. This single sweep is a stringent phase-agnostic test. The middle bosonic FCI phase has intrinsic semion topological order, whereas both neighboring phases break the global $U(1)$ symmetry, and SF@M further breaks lattice translation and rotation symmetry in the thermodynamic limit. Moreover, as established in Ref.~\cite{lu2025continuous}, the two boundaries are qualitatively different: the SF@M--FCI transition is first order, while the FCI--SF@$\Gamma$ transition is a continuous transition widely believed to be described by a change of parton band topology and critical gauge fluctuations~\cite{barkeshli2014continuous,barkeshli2015continuous,song2024phase}. A variational architecture that succeeds throughout the entire sweep thus cannot rely on any bias towards gapped or gapless states, or towards states with or without topological or symmetry-breaking orders.

Specifically, we directly compare the irrep-projected variational energy and variational states $|\Psi^{\text{phys}}_{\mathcal F;\chi}\rangle$ with ED on a $3\times6\times2$ sample ($2$ counts the sublattices), whose full Hilbert space dimension is over $9.41\times 10^7$. At this geometry, the SF@M ground state lies in the $(0,3)$ sector, the two almost degenerate bosonic FCI ground states lie in the $(0,0)$ and $(0,3)$ sectors, which undergo a sector rearrangement within the ground state manifold around $t''\approx-0.6$ without a phase transition, and the SF@$\Gamma$ ground state lies in the $(0,0)$ sector. We track the global ground state during the $t''$-sweep by focusing on the $(0,3)$ sector when $t''\in[-1,-0.6)$, and on $(0,0)$ sector for other $t''$ values. As shown in~\cref{fig:bosonic_combined_phase_diagram}, the variational energy error remains below $10^{-3}$ per particle throughout the sweep, while the wavefunction maintains a consistently excellent overlap over $99.9\%$, \emph{in all of SF@M, FCI, and SF@$\Gamma$ phases}, within two thousand VMC updates.

\begin{figure}[btp!] 
	\centering
	\includegraphics[width=1.0\linewidth]{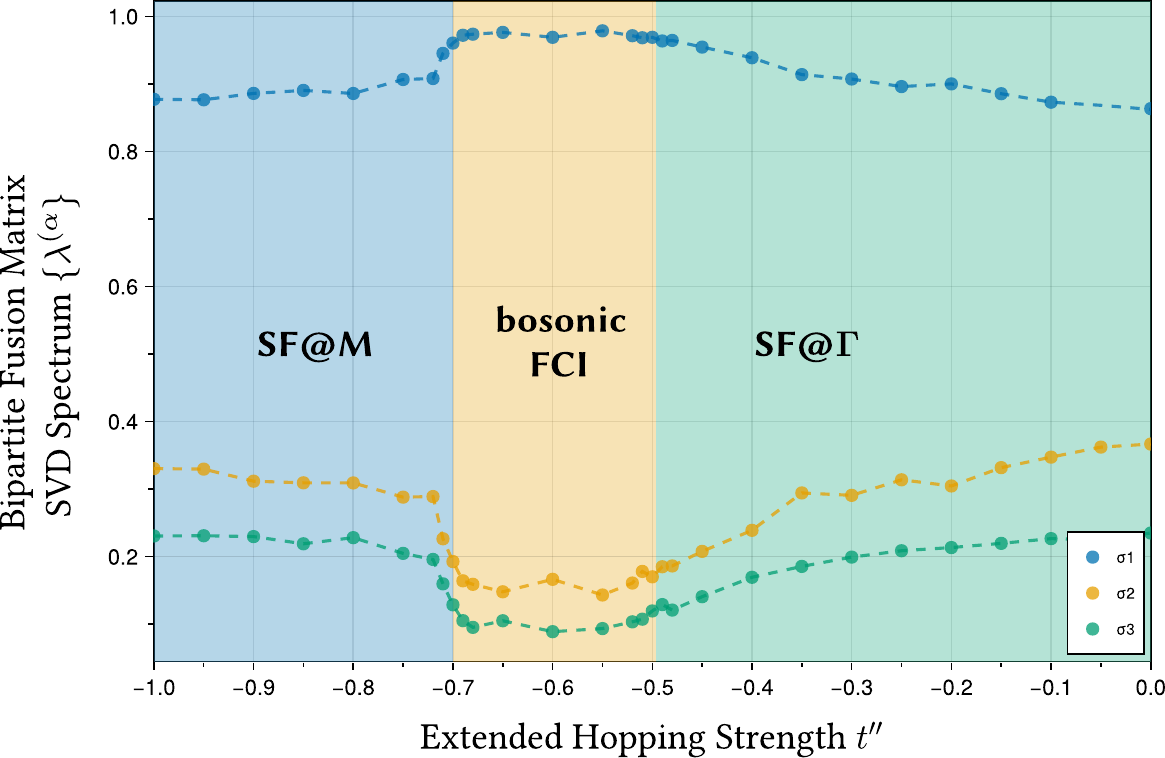}
	\caption{Bipartite Fusion Matrix SVD Spectrum of Bosonic Two-parton Case. Top three largest normalized singular values of the bipartite fusion matrix Eq.~\eqref{eq:SVD spectrum of bipartite fusion matrix} along the extended-hopping sweep $t''\in[-1,0]$ for the bosonic model Eq.~\eqref{eq:Bose-Hubbard Hamiltonian} on a $3\times6\times2$ sample at $\nu=1/2$. The abrupt jump in the singular value strengths around $t''\approx-0.7$, and the smooth change of the slope around $t''\approx-0.5$, are consistent with the first-order and continuous phase boundaries revealed in Ref.~\cite{lu2025continuous}, respectively. The shaded phase regions, taken from Ref.~\cite{lu2025continuous}, are also included for comparison.}
	\label{fig:bosonic bipartite fusion matrix SVD spectrum}
\end{figure}
Traditional parton field theories describe the FCI--SF@$\Gamma$ transition by splitting each boson into two fermionic partons $b_{\bm r}=f^{(1)}_{\bm r}f^{(2)}_{\bm r}$ where each parton occupies a band of Chern number $\mathcal C^{(\alpha)}=1$~\cite{barkeshli2014continuous,barkeshli2015continuous}. From a wavefunction perspective, this kind of local enlargement, or conversely the on-site fusion map Eq.~\eqref{eq:on-site fusion tensor} from the occupied parton orbitals to physical orbitals, in our language, \emph{assumes} a coordinate-identification rule similar to that in Laughlin wavefunctions. This construction forces the fusion tensor into the factorized form as Eq.~\eqref{eq:factorizable fusion tensor}, and therefore simplifies the HyperDet wavefunction to be products of Slater determinants Eq.~\eqref{eq:HyperDet of factorizable fusion tensor}. In the FCI phase, our optimized fusion tensors approach this factorization limit. Specifically, we construct the bipartite fusion matrix Eq.~\eqref{eq:bipartite fusion matrix} from each energy-optimized fusion tensor in the $t''$-sweep and perform the SVD analysis. As shown in~\cref{fig:bosonic bipartite fusion matrix SVD spectrum}, in the FCI phase where $-0.7\leq t''<-0.5$, there is a clear single dominant singular value component that in this $m=2$ example indicates a factorizable fusion tensor and thus a simplification of physical wavefunction to the product of Slater determinants.

However, as the system enters the SF@$\Gamma$ phase, additional singular values gain weight, indicating a departure from the nearly factorized fusion structure found in the FCI.  Indeed, a projected paired-parton Pfaffian ansatz was recently proposed to describe the bosonic FCI--SF@$\Gamma$ transition in Ref.~\cite{lotrivc2026paired} because the conventional product of parton Slater determinants, successful in the FCI, fails to describe the SF@$\Gamma$ phase. Our framework offers an alternative understanding to this transition: we retain a Slater determinant for each parton species throughout the sweep while treating the virtual-to-physical fusion map as a fully variational degree of freedom.  The excellent overlaps with ED ground states across the FCI-SF@$\Gamma$ shows that non-factorizable fusion can accurately describe the SF@$\Gamma$ phase without introducing explicit pairing into the auxiliary parton states.


Beyond this change in fusion structure, the redistribution of the bipartite fusion spectrum also tracks the locations and distinct characters of both transitions identified by DMRG~\cite{lu2025continuous}. We observe a sharp jump near the first-order SF@M--FCI transition at $t''\approx-0.7$, and a smooth evolution with a change of slope near the continuous FCI--SF@$\Gamma$ transition at $t''\approx-0.5$. These observations connect changes within the optimized fusion tensor to the phase boundaries and their transition behavior. The fusion spectrum thus provides a structural diagnostic across phases with different symmetry-breaking and topological orders, without evaluating phase-specific physical observables, and offers a way to probe transitions beyond the Ginzburg--Landau paradigm~\cite{barkeshli2014continuous,barkeshli2015continuous}.


\subsection{Phase Transitions near the $\nu=1/3$ Fermionic FCI}
To further showcase the expressiveness of our HyperDet wavefunction ansatz, we next consider an extended Hubbard model similar to Eq.~\eqref{eq:Bose-Hubbard Hamiltonian}, but for fermions on the checkerboard lattice:
\begin{align}
	&H = -\sum_{\langle i,j\rangle}t e^{i\phi_{i,j}} c_{i}^{\dagger}c_{j}- \sum_{\langle\langle i,j\rangle\rangle}t'_{i,j} c_{i}^{\dagger}c_{j}- \sum_{\langle\langle\langle i,j\rangle\rangle\rangle}t'' c_{i}^{\dagger}c_{j}+ \text{h.c.}\notag \\
	& \quad + V_{1}\sum_{\langle i,j\rangle}n_{i}n_{j} + V_{2}\sum_{\langle\langle i,j\rangle\rangle}n_{i}n_{j} + V_{3}\sum_{\langle\langle\langle i,j\rangle\rangle\rangle}n_{i}n_{j},\label{eq:Fermi-Hubbard Hamiltonian}
\end{align}
where $c_i^\dagger/c_i$ creates/annihilates a fermion and $n_i=c_i^\dagger c_i$ is the fermion number operator. It has been shown earlier in Ref.~\cite{sheng2011fractional} that both FCIs at $\nu=1/3$ and $\nu=1/5$ can be stabilized in the presence of the nearest-neighbor $V_1$ and next-nearest-neighbor $V_2$ interaction at the flatband settings~\cite{sun2011nearly}. Here, we instead push the model to an unexplored regime away from the flatband limit~\cite{sun2011nearly} under the parameter settings $t=1$ (as the energy unit), $t'_{1,2}=\pm1.4/(2+\sqrt2)$, $t''=-1/(2+2\sqrt2)\approx-0.21$, and complex phases $\phi_{i,j}=\pm\pi/4$, and investigate the strong-coupling limit of the $\nu=1/3$ state with all three interaction terms present $\bm V\equiv(V_1,V_2,V_3)=(5.6,1.26,0.56)$ to examine the expressiveness of our HyperDet wavefunction. We emphasize that such large interaction strengths induce strong band-mixing effects, positioning the model in a regime away from the usual LLL limit, as the interaction strength $||\bm V||\approx5.8$, band gap $\Delta_g\approx2.8$, and bandwidth $\Delta\approx0.5$ satisfy $||\bm V||>|\Delta_g|>|\Delta|$. ED diagnostics on the accessible clusters identify a robust Laughlin $U(1)_3$ FCI with significant gap enhancement for $-0.41\lesssim t''\lesssim0.15$, a candidate AHC at more negative $t''$, and a candidate CDW at more positive $t''$. The detailed numerical evidence is provided in the Appendix~\cite{Appendix}.

We perform VMC optimization by sweeping the extended hopping strength $t''\in[-0.62,0.41]$ for the fusion tensor with three parton species, $m=3$. We still benchmark the irrep-projected variational energy and variational states $|\Psi^{\text{phys}}_{\mathcal F;\chi}\rangle$ against ED results, but this time on a $3\times5\times2$ sample ($2$ is the sublattice)~\footnote{Our ansatz works equally well for larger systems, such as $3\times6\times2$ geometry. The $3\times5\times2$ geometry is chosen here because the three almost degenerate ground states of the FCI collapse to the same momentum sector on the $3\times6\times2$ geometry, and VMC is known to have difficulty distinguishing such a tiny-gap region~\cite{zhang2026projected} and can therefore be ``trapped'' in the middle of the ground state manifold of the FCI as an artifact, which can be removed by computing the ground state manifold overlap rather than a single-state overlap, see Appendix~\cite{Appendix} for details.}. Within FCI's threefold ground state manifold, the global ground state switches from the two exact degenerate sectors $(1,0)$ and $(2,0)$ to a single sector $(0,0)$ on this geometry at $t''\approx-0.28$ during the sweep. Again this sector reorganization occurs within the FCI phase without a phase transition. We still track the global ground state during the parameter sweep. We target sector $(1,0)$ for $t''\in[-0.62, -0.41)$ within the candidate AHC, and continue to target it for $t''\in[-0.41,-0.28)$ upon entering the FCI. We then switch the focus to the single ground-state sector $(0,0)$ for the remaining part of the FCI, $t''\in[-0.28,0.15)$, and retain it for $t''\in[0.15,0.41]$ within the candidate CDW.
\begin{figure}[btp!] 
	\centering
	\includegraphics[width=1.0\linewidth]{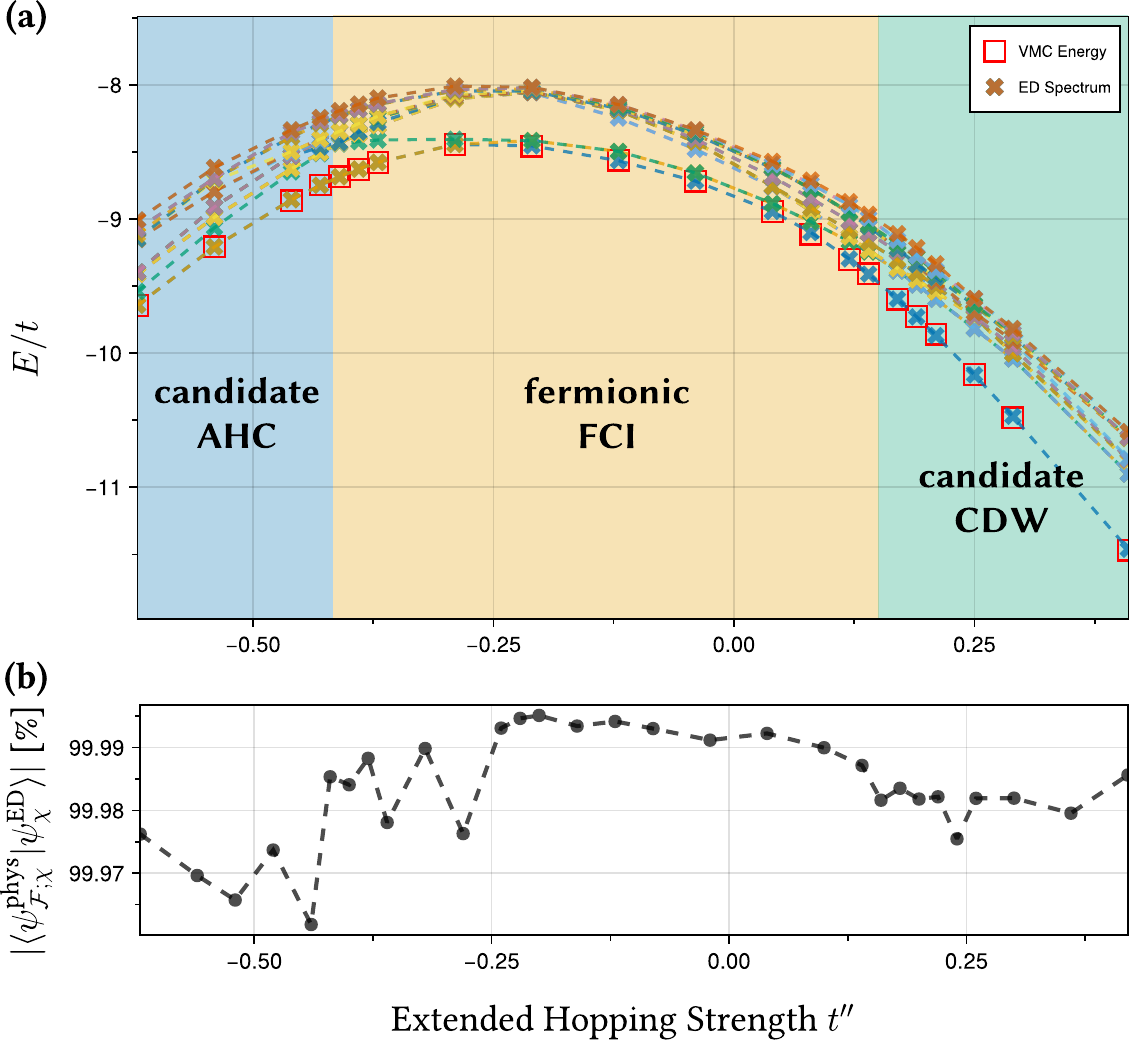}
	\caption{Benchmark of the irrep-projected HyperDet wavefunction ansatz against ED results on a $3\times5\times2$ sample ($2$ is the sublattice) for model Eq.~\eqref{eq:Fermi-Hubbard Hamiltonian} at band filling $\nu=1/3$ as the extended-hopping strength $t''\in[-0.62,0.41]$ is varied. (a) Comparison of the lowest twenty ED energies across all sectors and the chosen global ground state of the irrep-projected variational energy. The phase boundaries between candidate AHC, fermionic FCI, and candidate CDW here are determined from the ED spectrum sweep detailed in Appendix~\cite{Appendix}. Although this is not a thermodynamic-limit phase diagram, it suffices to benchmark the expressiveness of the ansatz on this finite cluster. (b) Wavefunction overlap with the ED ground states.}
	\label{fig:fermionic_combined_phase_diagram}
\end{figure}
The optimization results are summarized in~\cref{fig:fermionic_combined_phase_diagram}. The performance of the HyperDet wavefunction in the fermionic case remains excellent: the variational energy error remains less than $1.0\times10^{-3}$ per particle, with excellent wavefunction overlaps over $99.96\%$ across the entire sweep of AHC, FCI, and CDW phases, and in particular over $99.98\%$ throughout the FCI phase, within two thousand VMC updates.

We also perform the same SVD analysis of the bipartite fusion matrix Eq.~\eqref{eq:bipartite fusion matrix} for this $m=3$ fermionic case. Now there are three bipartitions $1|23$, $2|31$ and $3|12$, and we compute its species-averaged SVD spectrum $\bar\lambda\equiv\big\langle\frac{1}{3}\sum_{\alpha=1}^3 \lambda^{(\alpha)}(s_\ell)\big\rangle_{s_\ell}$ to quotient out the discrete species-label gauge redundancy. We also observe a coincidence of all three species-resolved spectra across the $t''$-sweep that further reveals a permutation-isotropic structure, see the Appendix~\cite{Appendix} for details.
\begin{figure}[btp!] 
	\centering
	\includegraphics[width=1.0\linewidth]{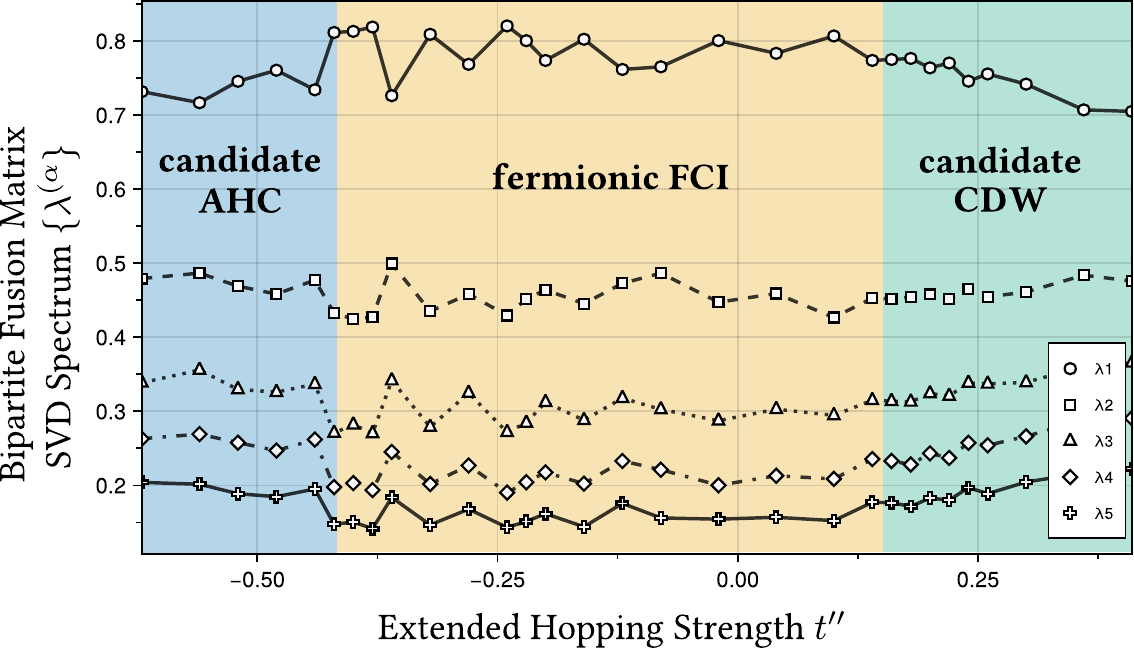}
	\caption{Species-Averaged Bipartite Fusion Matrix SVD Spectrum of Fermionic Three-parton Case. Top five largest normalized, species-averaged singular values of the bipartite fusion matrix Eq.~\eqref{eq:bipartite fusion matrix} obtained from the one-versus-two parton bipartitions of the optimized fusion tensor. The calculation is done for the fermionic model Eq.~\eqref{eq:Fermi-Hubbard Hamiltonian} on a $3\times5\times2$ sample at $\nu=1/3$. The observed singular value strengths jump around $t''\approx-0.41$ and a smooth change of slope around $t''\approx0.15$ is consistent with the ED-resolved finite-size AHC--FCI and FCI--CDW phase boundaries.}
	\label{fig:fermionic_bipartite_SVD_redistribution}
\end{figure}
As shown in~\cref{fig:fermionic_bipartite_SVD_redistribution}, this species-averaged SVD spectrum clearly undergoes structural changes around $t''\approx-0.41$ and $t''\approx0.15$, consistent with our ED-resolved AHC--FCI and FCI--CDW phase boundaries. Although not as sharp as the direct visualization of the fusion tensor factorizability for the $m=2$ bosonic case, this species-averaged SVD spectrum at $m=3$ still provides a useful proxy for visualization of structural changes in the fusion tensor. Due to the increased dimension, the fusion tensor factorizability can no longer be deduced from the singular value strengths of the $m=3$ fermionic spectrum, in contrast to the $m=2$ bosonic case. However, this showcase still supports our conjecture that the fusion-structure redistribution, here inferred by the SVD spectrum of the bipartite fusion matrix Eq.~\eqref{eq:SVD spectrum of bipartite fusion matrix}, provides a diagnostic tool for many-body phase transitions without extra information from phase-specific physical observables.

\section{Interpretability of the HyperDet Wavefunction Ansatz}\label{sec:interpretability}
\subsection{Parton Chern Numbers from the Optimized Fusion Tensor}\label{sec:parton_chern_numbers}
A central question for the interpretability of the optimized fusion tensor is whether its reconstructed occupied subspaces encodes the parton topology information, especially the parton Chern numbers $\mathcal C^{(\alpha)}$. We now show that it does, under the explicit hypothesis stated below. The extraction only relies on the occupied-orbital projectors $\bm P^{(\alpha)}$ constructed from the bipartite fusion matrices in Eq.~\eqref{eq:bipartite fusion matrix}, and assumes no input parton band structure or symmetry-fractionalization data.

Specifically, for the $N$ linearly independent columns of the occupied-orbital amplitude $\bm A^{(\alpha)}$ defined in Eq.~\eqref{eq:occupied-orbital amplitudes matrix A}, we can construct its orthonormal polar frame $\bm Q^{(\alpha)} = \bm A^{(\alpha)}(\bm A^{(\alpha)\dagger}\bm A^{(\alpha)})^{-1/2}$ for this column space, and define the orthogonal projector
\begin{equation}\label{eq:occupied-orbital projector}
	\bm P^{(\alpha)}=\bm A^{(\alpha)}\bigl(\bm A^{(\alpha)\dagger}\bm A^{(\alpha)}\bigr)^{-1}\bm A^{(\alpha)\dagger}=\bm Q^{(\alpha)}\bm Q^{(\alpha)\dagger}
\end{equation}
as the \emph{occupied-orbital projector}. Its physical meaning can be seen clearly as following: let $f_{s,c}^{(\alpha)\dagger}$ creates the auxiliary channel state $|s_\ell,c\rangle$, then for the reconstructed $N$-occupied parton Slater determinant state
\begin{equation}\label{eq:reconstructed Slater determinant state}
	|\Psi_{\text{rec}}^{(\alpha)}\rangle = \prod_{\ell_\alpha=1}^N d_{\ell_\alpha}^{(\alpha)\dagger} |0\rangle,\quad d_{\ell_\alpha}^{(\alpha),\dagger} := \sum_{(s,c)}\bm Q_{(s,c),\ell_\alpha}^{(\alpha)} f_{s,c}^{(\alpha),\dagger},
\end{equation}
the one-body correlation projector is exactly the constructed projector in Eq.~\eqref{eq:occupied-orbital projector}
\begin{equation}\label{eq:one-body correlation matrix meaning of occupied-orbital projector}
	\langle f_{s',c'}^{(\alpha)\dagger}f_{s,c}^{(\alpha)}\rangle_{\Psi_{\mathrm{rec}}^{(\alpha)}} = [\bm P^{(\alpha)}]_{(s,c),(s',c')}.
\end{equation}

We compute the real-space Bott index $\mathsf{Bott}(\bm P^{(\alpha)})$ as a finite-system diagnostic of the reconstructed occupied subspace. Its identification with the bulk parton Chern number $\mathcal C^{(\alpha)}$ relies on the following assumption:~\cite{hastings2010almost,hastings2011topological}
\begin{assumption}[Local Parton Insulator Realization]\label{ass:parton insulator}
	The constructed projector admits a local (or sufficiently quasilocal) gapped parent Hamiltonian whose occupied subspace varies smoothly and stays gapped under boundary twists.
\end{assumption}
\noindent The finite-system branch gap quantifies the stability of the Bott assignment, and the realization of parton insulator can also be achieved when translation fractionalization is present and parton band-filling becomes commensurate (see below). The construction, gauge convention, and numerical diagnostics are detailed in the Appendix~\cite{Appendix}.

We perform a VMC optimization of the HyperDet wavefunction for the bosonic Haldane model on a $4\times4\times2$ geometry at $t''=-0.6$ within the bosonic FCI phase. For the optimized fusion tensor in the $m=2$ case there, we find, in the canonical compact local-SVD gauge, a robust pair of parton Chern numbers $\mathcal C^{(1)}=\mathcal C^{(2)}=1$ for the two parton species, with a large branch gap $2.29$. We also perform a VMC optimization for the fermionic checkerboard model on a $3\times6\times2$ geometry at $t''=-0.12$ within the fermionic FCI phase. For the optimized fusion tensor in the $m=3$ case there, we obtain a robust tuple of parton Chern numbers $\mathcal C^{(1)}=\mathcal C^{(2)}=\mathcal C^{(3)}=1$ for all three parton species, with branch gaps $(0.661,0.575,0.797)$. The construction details and the robustness tests of these parton Chern numbers are provided in the Appendix~\cite{Appendix}. These extracted parton Chern numbers are all consistent with the topological contents expected for the $\nu=1/2$ bosonic FCI and $\nu=1/3$ fermionic FCI states. In the traditional field-theory descriptions of these states, the parton Chern numbers are prescribed as theoretical inputs. Here, by contrast, these parton-band topologies are determined directly from energy optimization. They are obtained without explicit construction of the parton mean-field Hamiltonian.

We would like to emphasize that relating the extracted parton Chern numbers to the physical Hall conductivity requires additional dynamical input beyond the static parton band topology. The Tucker reconstruction Eq.~\eqref{eq:simultaneous fusion reconstruction} makes this explicit: the physical wavefunction is determined by \emph{both} the occupied-orbital amplitudes $\bm A^{(\alpha)}$ and the Tucker core $\mathcal T_s$, whereas the parton Chern numbers are solely extracted from the $\bm A^{(\alpha)}$ matrices. A boundary twist that probes the physical topological responses, by contrast, probes both ingredients. If the twist is carried entirely by the channel amplitudes while the Tucker core is held periodic, the determinant expansion Eq.~\eqref{eq:HyperDet determinant expansion} fixes the physical twisted phase to be the sum of the species twisted phases $\varphi^{\text{phys}}_\mu=\sum_{\alpha}\varphi^{(\alpha)}_\mu\pmod{2\pi}$ for $\mu=1,2$. How the Tucker core transforms under a general twist, and whether it contributes to the Hall response, is an interesting open question that we leave for future work.

\subsection{Parameterization Gauge Redundancy and Fusion Tensor Stabilizer}
Next we turn to the intrinsic gauge structure that governs the symmetry-fractionalization content carried by the fusion tensor. The HyperDet wavefunction comes with an intrinsic gauge structure by its definition. As shown in Appendix~\cite{Appendix}, the physical wavefunction $\Psi^{\text{phys}}(\bm s)=\mathop{\mathsf{HyperDet}}[\mathcal F_{\bm s}]$ remains invariant up to an unimportant normalization factor
\begin{equation}\label{eq:HyperDet gauge transformation}
	\mathsf{HyperDet}[\mathcal F_{\bm s}']=\Big[\prod_{\alpha=1}^m\det\bm G^{(\alpha)}\Big] \mathsf{HyperDet}[\mathcal F_{\bm s}]
\end{equation}
under invertible linear transformations $\bm G^{(\alpha)}\in\mathrm{GL}(N;\mathbb C)$ and discrete parton-species permutation $\sigma\in S_m$ of the parton virtual legs of the fusion tensor:
\begin{equation}\label{eq:lifted gauge action}
	\mathcal F_{\bm s}\to \mathcal F_{\bm s}'=\sigma\Big[\bm G^{(1)}\times_1\cdots\bm G^{(m)}\times_m\Big]\mathcal F_{\bm s}.
\end{equation}
Here $\times_\alpha$ is the mode-$\alpha$ product, and $\sigma$ permutes $m$ virtual parton legs. These transformations constitute a gauge group of the form $\mathsf{GG}=\mathrm{GL}(N;\mathbb C)^m \rtimes S_m$. We emphasize that this gauge group is the \emph{minimal} gauge group present for every fusion tensor. Additional emergent gauge structures can arise for a specific optimized fusion tensor.

Among all gauge transformations, there exists a subset that leaves not only the physical wavefunction invariant but also the fusion tensor itself unchanged. These transformations form the stabilizer of the fusion tensor, denoted as $\mathsf{Stab}(\mathcal F)\subset\mathsf{GG}$, which is the HyperDet analogue of Wen's invariant gauge group in projective constructions of spin liquids~\cite{wen2002quantum}. An obvious part of the stabilizer consists of gauge transformations that multiply each parton species by an independent phase $e^{i\theta_\alpha}$ subject to $\prod_{\alpha=1}^m e^{i\theta_\alpha}=1$, which form $U(1)^{m-1}$. This stabilizer is built-in directly by the fusion map Eqs.~\eqref{eq:particle-slot-major parton mean-field state} and \eqref{eq:hyperdeterminant wavefunction amplitude} in HyperDet's construction, since inter-species parton hybridization is forbidden there.

\subsection{Symmetry Virtual Lifts and Stabilizer Ambiguities}
For a system with the symmetry group $\mathsf{SG}$, each $g\in\mathsf{SG}$ acts as $U_g\Psi^{\text{phys}}(\bm s)\simeq\Psi^{\text{phys}}(g\cdot\bm s)$ up to a phase. The HyperDet gauge structure Eqs.~\eqref{eq:HyperDet gauge transformation} and \eqref{eq:lifted gauge action} motivates a natural question: can this physical symmetry transformation be lifted to the virtual parton legs? Concretely, we seek a $U_g$-induced gauge transformation $\mathcal L_g=(\bm G_g^{(1)},\ldots,\bm G_g^{(m)}; \sigma_g)\in\mathsf{GG}$ that acts on the fusion tensor following the commutative diagram:
\begin{equation*}
	\xymatrix@R=2em@C=5em{\mathcal F_{\bm s} \ar@{:>}[r]^{\mathcal L_g} \ar[d]_{\mathsf{HyperDet}} & \mathcal F_{g\cdot\bm s} \ar[d]^{\mathsf{HyperDet}} \\
	\Psi^{\text{phys}}(\bm s) \ar[r]^{U_g} & \Psi^{\text{phys}}(g\cdot\bm s)}
\end{equation*}
In practice, we search for $\mathcal L_g\in\mathsf{GG}$ by minimizing the normalized residual with respect to the Frobenius norm $\|\cdot\|_F$:
\begin{equation*}
	\epsilon(\mathcal L_g)^2=\dfrac{\sum_{s_\ell}\|\mathcal F(g\cdot s_\ell)-\mathcal L_g\cdot\mathcal F(s_\ell)\|_F^2}{\sum_{s_\ell}\|\mathcal F(g\cdot s_\ell)\|_F^2}
\end{equation*}
As a standard least-squares problem, the minimization can be performed iteratively by optimizing matrices $\bm G_g^{(\alpha)}\in\mathrm{GL}(N;\mathbb C)$ and enumerating every parton-species permutation $\sigma_g\in S_m$ (see Appendix~\cite{Appendix} for details).

Even when a virtual lift exists with nearly vanishing residual $\epsilon(\mathcal L_g)\approx0$, the physical compatibility conditions $\omega(\{g\})=\bm e$, such as the translation commutation relation $\omega(\{T_1,T_2\})=T_1 T_2 T_1^{-1} T_2^{-1}=\bm e$, can still retain $\mathsf{Stab}(\mathcal F)$-valued ambiguities after virtual lifts, quantified by $\omega(\{\mathcal L_g\})\in\mathsf{Stab}(\mathcal F)$. Similar virtual-lift ambiguities are also present in other variational ansatzes, notably the matrix product states (MPS), where group cohomology obstructions $H^{d+1}(\mathsf{SG},U(1))$ are recognized to classify symmetry-protected topological phases~\cite{chen2011classification,chen2010local,chen2013symmetry,perez2008string,williamson2014matrix}.

A nontrivial ambiguity $\omega(\{\mathcal L_g\})$ signals a \emph{projective} representation of the physical symmetry $g\in\mathsf{SG}$ at the parton level --- a property known as \emph{parton symmetry fractionalization}~\cite{wen2002quantum,essin2013classifying,barkeshli2019symmetry}. Since we never impose fractionalization data a priori, the observation of a nontrivial $\omega(\{\mathcal L_g\})$ directly reveals the \emph{emergent} parton symmetry fractionalization. This information is extracted purely from fusion-tensor optimization without theoretical input.

\subsection{Translation Fractionalization}
For translation-symmetric systems, we first examine whether the symmetry can be realized at the parton level. By minimizing the residual $\epsilon(\mathcal L_{T_{1,2}})$ for the optimized fusion tensor obtained from the bosonic FCI phase on a $4\times4\times2$ geometry at $t''=-0.6$, we find that translations in both directions can be virtually lifted with small residuals $\epsilon(\mathcal L_{T_{1,2}})<0.022$, as long as the lifting is \emph{parton-identical} with the trivial permutation $\sigma=\mathrm{id}$. By contrast, large residuals $\epsilon(\mathcal L_{T_{1,2}})\approx 0.30$ are obtained for parton-exchange virtual lifts where $\sigma=(12)$.

The translation commutation relation enforces a compatible identity $\omega(\{T_1,T_2\})=T_1 T_2 T_1^{-1} T_2^{-1}=\bm e$. At the parton level, however, this relation closes only up to an $\mathsf{Stab}(\mathcal F)$-valued ambiguity for the lifted translation operators $\mathcal L_{T_{1,2}}=(\bm G^{(1)}_{T_{1,2}},\bm G^{(2)}_{T_{1,2}};\sigma=\mathrm{id})$, namely
\begin{equation}\label{eq:translation lift ambiguity}
	\omega(\{\mathcal L_{T_1},\mathcal L_{T_2}\})= \mathcal L_{T_1}\mathcal L_{T_2}\mathcal L_{T_1}^{-1}\mathcal L_{T_2}^{-1}\in\mathsf{Stab}(\mathcal F),
\end{equation}

For the simplest compact parton-identical stabilizer $\mathsf{Stab}(\mathcal F)=U(1)^{m-1}$, the lifted commutator Eq.~\eqref{eq:translation lift ambiguity} takes the diagonal form $e^{i\theta_\alpha}\bm I_N$ for each parton species under the constraint $\prod_{\alpha=1}^m e^{i\theta_\alpha}=1$. The periodic boundary conditions of the 2D bosonic model further restrict the fractionalization phases to the four solutions:
\begin{equation*}
	(e^{i\theta_1}, e^{i\theta_2}) = \Big\{(1,1),(-1,-1),(i,-i),(-i,i)\Big\}.
\end{equation*}
We find a robust parton translation fractionalization relation with $(e^{i\theta_1}, e^{i\theta_2})=(-1,-1)$ for both parton species. Specifically, the lifted matrices in Eq.~\eqref{eq:translation lift ambiguity} satisfy (here $\nu=1/2$)
\begin{equation}\label{eq:translation fractionalization}
	\bm G^{(\alpha)}_{T_1}\bm G^{(\alpha)}_{T_2}(\bm G^{(\alpha)}_{T_1})^{-1}(\bm G^{(\alpha)}_{T_2})^{-1}\approx e^{i\nu 2\pi}\bm I_N,
\end{equation}
for both $\alpha=1,2$ with a small centrality error $<1.2\times10^{-4}$.

This non-commutative relation of the lifted translation operators Eq.~\eqref{eq:translation fractionalization} is a direct signature of the \emph{emergent} $\pi$-flux phase felt by partons in the $\nu=1/2$ bosonic FCI. To accommodate an integer number of fluxes per unit cell, the parton unit cell must be doubled, and the parton Brillouin zone is reduced by half. The parton band structure is consequently folded into the reduced Brillouin zone to allow each species to satisfy the integer-filling condition $\nu^{(\alpha)}=1$ for $\alpha=1,2$. This is a necessary condition for each parton mean-field state to be a gapped insulator well-described by a single parton Slater determinant in Eq.~\eqref{eq:parton single Slater determinant}, consistent with the expected parton description of FQH physics~\cite{jain1989incompressible,jain1990theory}. In traditional parton constructions, this is a theoretical input~\cite{barkeshli2014continuous,barkeshli2015continuous}, but here it emerges as a feature of the optimized fusion tensor without any prior assumptions. This highlights the interpretability of the HyperDet wavefunction ansatz as an unbiased detector of parton symmetry fractionalization.

We also extract the translation virtual lifts for optimized fusion tensor obtained from the fermionic FCI model on a commensurate $3\times6\times2$ geometry at $t''=-0.12$. The fitting quality is, however, much worse than the bosonic case, with large fitting residuals $\epsilon(\mathcal L_{T_1})\approx0.76$ and $\epsilon(\mathcal L_{T_2})\approx0.81$ that cannot be improved by a different parton permutation (see Appendix~\cite{Appendix} for details). This failure suggests an unreliable readout of the diagonal translation fractionalization phases $e^{i\theta_\alpha}\approx e^{-i2\pi/3}$ for $\alpha=1,2,3$ in Eq.~\eqref{eq:translation lift ambiguity}. Identifying the origin of this failure requires further projective symmetry group (PSG) analysis and parton band-structure calculations~\cite{wen2002quantum,wang2006spin,wang2010schwinger,lu2012symmetry,essin2013classifying}, which we leave for future work~\cite{lin2026superfluid-FCI}.

The optimized fusion tensor thus robustly supplies reconstructed parton topological data and, when faithful virtual lifts exist, the emergent parton symmetry-fractionalization data. The exact reconstruction establishes the occupied-subspace interpretation for arbitrary fusion rank, while the parton Chern numbers hint on the underlying topological order. Taken together, they complete the microscopic data needed for the low-energy field theories of the underlying states.

\section{Discussion}\label{sec:discussion}
In this work, we generalize the conventional parton construction by promoting the virtual-to-physical fusion map to a fully variational degree of freedom, and demonstrate that the resulting hyperdeterminant (HyperDet) wavefunction $\langle\bm s|\Psi^{\text{phys}}\rangle=\mathop{\mathsf{HyperDet}}[\mathcal F_{\bm s}]$ can serve as a phase-agnostic ansatz for bosonic and fermionic strongly correlated systems.  For both lattice models studied in this work, the same HyperDet architecture achieves overlaps exceeding $99.9\%$ with ED ground states for nearly all sampled parameters across the FCI and its neighboring competing phases. These benchmarks show that the learned virtual-to-physical fusion map can adapt to states with different patterns of symmetry breaking, fractionalization, and topological order.

As the core variational object of the HyperDet wavefunction, the fusion tensor $\mathcal F$ is obtained by evaluating the virtual-to-physical fusion map on the physical many-body configuration and the occupied parton orbitals. As Eq.~\eqref{eq:fusion tensor connection relation} makes explicit, the fusion tensor draws on \emph{two} independent sources of variational freedom: the occupied parton orbitals $\varphi^{(\alpha)}(u_\alpha)$ expressed in a local parton basis, and the long-overlooked local fusion coefficients $\mathcal F^{\text{loc}}_{s_\ell;u_1,\ldots,u_m}$ that record the parton-basis to physical-basis fusion transformation. In conventional wavefunction ansatzes, such as Jastrow-Slater states~\cite{foulkes2001quantum} and Gutzwiller-projected parton states~\cite{gros1989physics,edegger2007gutzwiller}, the parton orbitals are the only variational degrees of freedom.  The parton-basis to physical-basis fusion transformation, by contrast, has not been recognized as an independent variational degree of freedom in these constructions: the corresponding fusion processes are fixed by prescription when the ansatz is written down. Such an implicit prescription restricts the variational space and thereby limits the wavefunction expressiveness.  This is demonstrated in the FCI--SF@$\Gamma$ transition the bosonic model: by treating the fusion map as a variational degree of freedom, our HyperDet wave function is able to accurately describe this transition.

In addition to its accuracy, the HyperDet wavefunction also gives direct access to the occupied parton subspaces encoded in the energy-optimized fusion tensor. We find that the gauge-fixed bipartite fusion spectra track structural changes across the phase boundaries, without having to evaluate phase-specific physical observables. The exact weighted Tucker reconstruction in Eq.~\eqref{eq:simultaneous fusion reconstruction} and determinant expansion in Eq.~\eqref{eq:HyperDet determinant expansion} identify a configuration-independent occupied-orbital matrix $\bm A^{(\alpha)}$ and a reconstructed Slater determinant for each parton species; the corresponding parton projectors $\bm P^{(\alpha)}$ yield Bott index tuples that reproduce the expected parton Chern numbers of the bosonic and fermionic FCI states, without prescribing a parton Hamiltonian. The transparent structure of the virtual-to-physical fusion map further allows us to ask whether the physical symmetry can be lifted to the parton level, from which we recover the translation-fractionalization pattern of the bosonic state without imposing a projective symmetry class assignment. The optimized fusion tensor therefore serves both as a variational representation of the HyperDet wavefunction and as an interpretable source of parton microscopics that connects the variational state to its field-theory description.

Both models studied in this work are described by lattice Hamiltonians, but the HyperDet formulation is not restricted to real-space physical orbitals: the physical basis can be either orthonormal or overcomplete (such as a coherent-state basis), and can live in real space or in momentum space. Its momentum-space formulation offers a natural route toward multiband calculations for 2D materials including twisted transition metal dichalcogenides~\cite{cai2023signatures,park2023observation,xu2023observation,zeng2023thermodynamic} and rhombohedral multilayer graphene~\cite{lu2024fractional,lu2025extended,han2025signatures,han2026evidence} for the following reasons.  First, the single-particle Hamiltonians of these moire systems are usually given by a continuum Hamiltonian and the physics is most conveniently described in the band basis.  Second, it has been shown that band mixing has a significant effect on FCI stability~\cite{abouelkomsan2024band,hou2025stabilizing,dong2024anomalous,yu2025moireIV,li2025multiband,xu2024maximally,he2025fractional,gonccalves2026spinless}.  However, many-band ED simulation rapidly become intractable as more bands are added.  In contrast, the HyperDet wavefunction is well suited for this purpose: multiband effects can be captured straightforwardly by incorporating all relevant Bloch orbitals into the momentum-space graph on which the full hyper-rectangular fusion tensor lives.  Because the configuration-restricted fusion tensor $\mathcal F_{\bm s}$ entering the hyperdeterminant definition in Eq.~\eqref{eq:hyperdeterminant wavefunction amplitude} retains its size $N^{m+1}$ at fixed particle number $N$, the evaluation of the HyperDet does \emph{not} scale with the single-particle Hilbert space dimension while the number of variational parameters grows linearly.  We leave this promising application of the HyperDet wavefunction for future study.

Another distinctive feature of the HyperDet wavefunction is that no microscopic parton basis is prescribed as input. The standard parton construction $c_{\bm r}=\prod_{\alpha}f^{(\alpha)}_{\bm r}$ aligns physical and parton operators in a common real-space basis. The HyperDet wavefunction, in contrast, absorbs species-dependent basis transformations into the fusion coefficients through Eq.~\eqref{eq:fusion tensor connection relation}, without requiring alignment with the physical basis or among species. Each species can therefore be represented in a suitable basis, opening a route toward efficient descriptions of exotic physical states with parton Fermi seas coexisting with gapped parton Chern insulators. An established example is the $\nu=1/2$ fermionic composite Fermi liquid (CFL)~\cite{rezayi1994fermi,halperin1993theory}, whose lattice trial wavefunction combines two real-space Chern-band parton Slater determinants with a momentum-space parton Fermi sea Slater determinant~\cite{mishmash2016entanglement} already in the factorizable limit as Eq.~\eqref{eq:HyperDet of factorizable fusion tensor}. The numerically observed entanglement enhencement $S_2\sim L\ln L$ for these trial states~\cite{shao2015entanglement,mishmash2016entanglement} inherited from the gapless parton Fermi surface~\cite{wolf2006violation,gioev2006entanglement,swingle2010entanglement}, thus provides a concrete \emph{lower bound} on the entanglement capacity of the HyperDet ansatz beyond a strict area law scaling~\cite{eisert2010colloquium,hastings2007area,hastings2006spectral}. Extending fusion-tensor optimization to CFL-like states and quantifying their fusion complexity are therefore natural directions for future study.

The unsuccessful translation-lift fits in the fermionic example raise a separate question about the relation between physical symmetry and its realization on the optimized fusion tensor. The large residuals prevent a reliable fractionalization readout within the tested gauge structure $\mathsf{GG}=\mathrm{GL}(N;\mathbb C)^m\rtimes S_m$, but do not establish an obstruction at the level of the physical state. There are several possible reasons. For example, the failure may reflect poor convergence toward PSG-compatible states, since our optimization is purely energetic. Another possibility is that the true energy ground state may not support the full lift of the physical crystalline symmetries. Distinguishing these possibilities requires an explicit construction of PSG-compatible parton orbitals and fusion tensors, which is addressed elsewhere~\cite{lin2026superfluid-FCI}.

Scalability is currently the main limitation when applying the \emph{exact} HyperDet wavefunction ansatz to larger system sizes due to the inherent combinatorial complexity of evaluating hyperdeterminants. The exact determinant expansion in Eq.~\eqref{eq:HyperDet determinant expansion} makes this limitation explicit through the generally combinatorial number of channel configurations. This obstacle is not unique to the HyperDet wavefunction. Indeed, exact contraction of generic two-dimensional projected entangled pair states (PEPS) is also \textsf{\#P-complete}~\cite{schuch2007computational}, while practical PEPS algorithms rely on systematically improvable approximate contractions with polynomial cost in system size. This parallel points to an important direction: developing tractable and controllable approximations to the HyperDet wavefunction. One approach is to introduce a \emph{locality structure} into the fusion tensor $\mathcal F^{\text{loc}}_{s_\ell;u_1,\ldots,u_m}$ in Eq.~\eqref{eq:fusion tensor coefficients in parton local basis} by restricting it to a bounded domain for each physical site. This reduces the hyperdeterminant evaluation over all $N$ occupied parton states to a small-cluster contraction while leaving the remainder in product-of-determinants form, conceptually similar to approximate tensor-network contraction and tensor-renormalization schemes~\cite{levin2007tensor,orus2014practical}. This approach, termed \emph{projective expansion}, has been developed in detail in our recent work~\cite{lin2026hyperdeterminant}. We leave its application for future work.

\vspace*{1em}

\section*{Acknowledgements}
We thank Andrei Bernevig, Zhen Bi, Daniel Parker, and Nicolas Regnault for helpful discussions. X. Hu thanks Jihang Zhu for suggesting SVD plots for bipartite fusion matrix. This work is supported by DOE Award No. DE-SC0012509. This research used resources of the National Energy Research Scientific Computing Center, a DOE Office of Science User Facility supported by the Office of Science of the U.S. Department of Energy under Contract No. DE-AC02-05CH11231 using NERSC awards BES-ERCAP0037104 and BES-ERCAP0037097. This work was facilitated through the use of advanced computational, storage, and networking infrastructure provided by the AI-core as well as the Hyak supercomputer system funded by the University of Washington Molecular Engineering Materials Center at the University of Washington (DMR-2308979).

\section*{Code Availability}
The dynamic-programming algorithm for evaluation of HyperDet wavefunction is wrapped as a Julia module \textsf{HyperDet\_DP} and is publicly available on GitHub \footnote{See \href{https://github.com/xiaodong-hu/HyperDet_DP}{\textsf{https://github.com/xiaodong-hu/HyperDet\_DP}}}.

\section*{Data Availability}
The large ED and VMC data are available from the authors upon reasonable request.

\bibliography{ref}

@article{tang2011high,
  author    = {Tang, Evelyn and Mei, Jia-Wei and Wen, Xiao-Gang},
  journal   = {Physical review letters},
  number    = {23},
  pages     = {236802},
  publisher = {APS},
  title     = {High-temperature fractional quantum {Hall} states},
  volume    = {106},
  year      = {2011}
}

@article{neupert2011fractional,
  author    = {Neupert, Titus and Santos, Luiz and Chamon, Claudio and Mudry, Christopher},
  journal   = {Physical review letters},
  number    = {23},
  pages     = {236804},
  publisher = {APS},
  title     = {Fractional quantum {Hall} states at zero magnetic field},
  volume    = {106},
  year      = {2011}
}

@article{sheng2011fractional,
  author    = {Sheng, DN and Gu, Zheng-Cheng and Sun, Kai and Sheng, L},
  journal   = {Nature communications},
  number    = {1},
  pages     = {389},
  publisher = {Nature Publishing Group UK London},
  title     = {Fractional quantum {Hall} effect in the absence of {Landau} levels},
  volume    = {2},
  year      = {2011}
}

@article{regnault2011fractional,
  author    = {Regnault, Nicolas and Bernevig, B Andrei},
  journal   = {Physical Review X},
  number    = {2},
  pages     = {021014},
  publisher = {APS},
  title     = {Fractional {Chern} insulator},
  volume    = {1},
  year      = {2011}
}

@article{cai2023signatures,
  author    = {Cai, Jiaqi and Anderson, Eric and Wang, Chong and Zhang, Xiaowei and Liu, Xiaoyu and Holtzmann, William and Zhang, Yinong and Fan, Fengren and Taniguchi, Takashi and Watanabe, Kenji and others},
  journal   = {Nature},
  number    = {7981},
  pages     = {63--68},
  publisher = {Nature Publishing Group UK London},
  title     = {Signatures of fractional quantum anomalous {Hall} states in twisted {MoTe$_2$}},
  volume    = {622},
  year      = {2023}
}

@article{park2023observation,
  author    = {Park, Heonjoon and Cai, Jiaqi and Anderson, Eric and Zhang, Yinong and Zhu, Jiayi and Liu, Xiaoyu and Wang, Chong and Holtzmann, William and Hu, Chaowei and Liu, Zhaoyu and others},
  journal   = {Nature},
  number    = {7981},
  pages     = {74--79},
  publisher = {Nature Publishing Group UK London},
  title     = {Observation of fractionally quantized anomalous {Hall} effect},
  volume    = {622},
  year      = {2023}
}

@article{xu2023observation,
  author    = {Xu, Fan and Sun, Zheng and Jia, Tongtong and Liu, Chang and Xu, Cheng and Li, Chushan and Gu, Yu and Watanabe, Kenji and Taniguchi, Takashi and Tong, Bingbing and others},
  journal   = {Physical Review X},
  number    = {3},
  pages     = {031037},
  publisher = {APS},
  title     = {Observation of integer and fractional quantum anomalous {Hall} effects in twisted bilayer {MoTe$_2$}},
  volume    = {13},
  year      = {2023}
}

@article{zeng2023thermodynamic,
  author    = {Zeng, Yihang and Xia, Zhengchao and Kang, Kaifei and Zhu, Jiacheng and Kn{\"u}ppel, Patrick and Vaswani, Chirag and Watanabe, Kenji and Taniguchi, Takashi and Mak, Kin Fai and Shan, Jie},
  journal   = {Nature},
  number    = {7981},
  pages     = {69--73},
  publisher = {Nature Publishing Group UK London},
  title     = {Thermodynamic evidence of fractional {Chern} insulator in moir{\'e} {MoTe$_2$}},
  volume    = {622},
  year      = {2023}
}

@article{lu2024fractional,
  author    = {Lu, Zhengguang and Han, Tonghang and Yao, Yuxuan and Reddy, Aidan P and Yang, Jixiang and Seo, Junseok and Watanabe, Kenji and Taniguchi, Takashi and Fu, Liang and Ju, Long},
  journal   = {Nature},
  number    = {8000},
  pages     = {759--764},
  publisher = {Nature Publishing Group UK London},
  title     = {Fractional quantum anomalous {Hall} effect in multilayer graphene},
  volume    = {626},
  year      = {2024}
}

@article{lu2025extended,
  author    = {Lu, Zhengguang and Han, Tonghang and Yao, Yuxuan and Hadjri, Zach and Yang, Jixiang and Seo, Junseok and Shi, Lihan and Ye, Shenyong and Watanabe, Kenji and Taniguchi, Takashi and others},
  journal   = {Nature},
  number    = {8048},
  pages     = {1090--1095},
  publisher = {Nature Publishing Group UK London},
  title     = {Extended quantum anomalous {Hall} states in graphene/hBN moir{\'e} superlattices},
  volume    = {637},
  year      = {2025}
}

@article{goldshlager2024kaczmarz,
  author    = {Goldshlager, Gil and Abrahamsen, Nilin and Lin, Lin},
  journal   = {Journal of Computational Physics},
  pages     = {113351},
  publisher = {Elsevier},
  title     = {A Kaczmarz-inspired approach to accelerate the optimization of neural network wavefunctions},
  volume    = {516},
  year      = {2024}
}

@article{gu2010grassmann,
  author  = {Gu, Zheng-Cheng and Verstraete, Frank and Wen, Xiao-Gang},
  journal = {arXiv preprint arXiv:1004.2563},
  title   = {Grassmann tensor network states and its renormalization for strongly correlated fermionic and bosonic states},
  year    = {2010}
}

@article{hu2024hyperdeterminants,
  author    = {Hu, Xiaodong and Xiao, Di and Ran, Ying},
  journal   = {Physical Review B},
  number    = {24},
  pages     = {245125},
  publisher = {APS},
  title     = {Hyperdeterminants and composite fermion states in fractional {Chern} insulators},
  volume    = {109},
  year      = {2024}
}

@article{hu2026composite,
  author    = {Hu, Xiaodong and Ran, Ying and Xiao, Di},
  journal   = {Physical Review Letters},
  number    = {6},
  pages     = {066504},
  publisher = {APS},
  title     = {Composite Fermion Theory of Fractional {Chern} Insulator Stability},
  volume    = {136},
  year      = {2026}
}

@article{jain1989incompressible,
  author    = {Jain, Jainendra K},
  journal   = {Physical Review B},
  number    = {11},
  pages     = {8079},
  publisher = {APS},
  title     = {Incompressible quantum {Hall} states},
  volume    = {40},
  year      = {1989}
}

@article{jain1990theory,
  author    = {Jain, Jainendra K},
  journal   = {Physical Review B},
  number    = {11},
  pages     = {7653},
  publisher = {APS},
  title     = {Theory of the fractional quantum {Hall} effect},
  volume    = {41},
  year      = {1990}
}

@article{lu2012symmetry,
  author    = {Lu, Yuan-Ming and Ran, Ying},
  journal   = {Physical Review B—Condensed Matter and Materials Physics},
  number    = {16},
  pages     = {165134},
  publisher = {APS},
  title     = {Symmetry-protected fractional {Chern} insulators and fractional topological insulators},
  volume    = {85},
  year      = {2012}
}

@article{lu2025continuous,
  author    = {Lu, Hongyu and Wu, Han-Qing and Chen, Bin-Bin and Meng, Zi Yang},
  journal   = {Physical Review Letters},
  number    = {7},
  pages     = {076601},
  publisher = {APS},
  title     = {Continuous Transition between Bosonic Fractional {Chern} Insulator and Superfluid},
  volume    = {134},
  year      = {2025}
}

@article{barkeshli2014continuous,
  author    = {Barkeshli, Maissam and McGreevy, John},
  journal   = {Physical Review B},
  number    = {23},
  pages     = {235116},
  publisher = {APS},
  title     = {Continuous transition between fractional quantum {Hall} and superfluid states},
  volume    = {89},
  year      = {2014}
}

@article{barkeshli2015continuous,
  author    = {Barkeshli, M and Yao, NY and Laumann, CR},
  journal   = {Physical review letters},
  number    = {2},
  pages     = {026802},
  publisher = {APS},
  title     = {Continuous preparation of a fractional {Chern} insulator},
  volume    = {115},
  year      = {2015}
}

@article{song2024phase,
  author    = {Song, Xue-Yang and Zhang, Ya-Hui and Senthil, T},
  journal   = {Physical Review B},
  number    = {8},
  pages     = {085143},
  publisher = {APS},
  title     = {Phase transitions out of quantum {Hall} states in moir{\'e} materials},
  volume    = {109},
  year      = {2024}
}

@article{senthil2004deconfined,
  author    = {Senthil, Todadri and Vishwanath, Ashvin and Balents, Leon and Sachdev, Subir and Fisher, Matthew PA},
  journal   = {Science},
  number    = {5663},
  pages     = {1490--1494},
  publisher = {American Association for the Advancement of Science},
  title     = {Deconfined quantum critical points},
  volume    = {303},
  year      = {2004}
}

@article{senthil2004quantum,
  author    = {Senthil, T and Balents, Leon and Sachdev, Subir and Vishwanath, Ashvin and Fisher, Matthew PA},
  journal   = {Physical Review B—Condensed Matter and Materials Physics},
  number    = {14},
  pages     = {144407},
  publisher = {APS},
  title     = {Quantum criticality beyond the Landau-Ginzburg-Wilson paradigm},
  volume    = {70},
  year      = {2004}
}

@article{wang2011fractional,
  author    = {Wang, Yi-Fei and Gu, Zheng-Cheng and Gong, Chang-De and Sheng, DN},
  journal   = {Physical review letters},
  number    = {14},
  pages     = {146803},
  publisher = {APS},
  title     = {Fractional quantum {Hall} effect of hard-core bosons in topological flat bands},
  volume    = {107},
  year      = {2011}
}

@article{wang2017deconfined,
  author    = {Wang, Chong and Nahum, Adam and Metlitski, Max A and Xu, Cenke and Senthil, T},
  journal   = {Physical Review X},
  number    = {3},
  pages     = {031051},
  publisher = {APS},
  title     = {Deconfined quantum critical points: symmetries and dualities},
  volume    = {7},
  year      = {2017}
}

@article{wen1999projective,
  author    = {Wen, Xiao-Gang},
  journal   = {Physical Review B},
  number    = {12},
  pages     = {8827},
  publisher = {APS},
  title     = {Projective construction of non-Abelian quantum {Hall} liquids},
  volume    = {60},
  year      = {1999}
}

@article{zhang2020pseudogap,
  author    = {Zhang, Ya-Hui and Sachdev, Subir},
  journal   = {Physical Review Research},
  number    = {2},
  pages     = {023172},
  publisher = {APS},
  title     = {From the pseudogap metal to the Fermi liquid using ancilla qubits},
  volume    = {2},
  year      = {2020}
}

@article{zhou2025variational,
  author    = {Zhou, Boran and Jin, Hui-Ke and Zhang, Ya-Hui},
  journal   = {Physical Review B},
  number    = {11},
  pages     = {115159},
  publisher = {APS},
  title     = {Variational wavefunction for a {Mott} insulator at finite {$U$} using ancilla qubits},
  volume    = {112},
  year      = {2025}
}

@article{laughlin1983anomalous,
  author    = {Laughlin, Robert B},
  journal   = {Physical Review Letters},
  number    = {18},
  pages     = {1395},
  publisher = {APS},
  title     = {Anomalous quantum {Hall} effect: An incompressible quantum fluid with fractionally charged excitations},
  volume    = {50},
  year      = {1983}
}

@article{lotrivc2026paired,
  author    = {Lotri{\v{c}}, Tev{\v{z}} and Simon, Steven H},
  journal   = {Physical Review Letters},
  number    = {9},
  pages     = {096601},
  publisher = {APS},
  title     = {Paired parton trial states for the superfluid-fractional chern insulator transition},
  volume    = {136},
  year      = {2026}
}

@article{sun2011nearly,
  author    = {Sun, Kai and Gu, Zhengcheng and Katsura, Hosho and Das Sarma, S},
  journal   = {Physical review letters},
  number    = {23},
  pages     = {236803},
  publisher = {APS},
  title     = {Nearly flatbands with nontrivial topology},
  volume    = {106},
  year      = {2011}
}

@article{haldane1985finite,
  author    = {Haldane, F Duncan M and Rezayi, Edward H},
  journal   = {Physical review letters},
  number    = {3},
  pages     = {237},
  publisher = {APS},
  title     = {Finite-size studies of the incompressible state of the fractionally quantized {Hall} effect and its excitations},
  volume    = {54},
  year      = {1985}
}

@article{yoshioka1983ground,
  author    = {Yoshioka, D and Halperin, Bertrand I and Lee, PA},
  journal   = {Physical review letters},
  number    = {16},
  pages     = {1219},
  publisher = {APS},
  title     = {Ground state of two-dimensional electrons in strong magnetic fields and 1 3 quantized {Hall} effect},
  volume    = {50},
  year      = {1983}
}

@article{jain1989composite,
  author    = {Jain, Jainendra K},
  journal   = {Physical review letters},
  number    = {2},
  pages     = {199},
  publisher = {APS},
  title     = {Composite-fermion approach for the fractional quantum {Hall} effect},
  volume    = {63},
  year      = {1989}
}

@article{kamilla1996excitons,
  author    = {Kamilla, RK and Wu, XG and Jain, JK},
  journal   = {Physical Review B},
  number    = {7},
  pages     = {4873},
  publisher = {APS},
  title     = {Excitons of composite fermions},
  volume    = {54},
  year      = {1996}
}

@article{jain1997composite,
  author    = {Jain, JK and Kamilla, RK},
  journal   = {International Journal of Modern Physics B},
  number    = {22},
  pages     = {2621--2660},
  publisher = {World Scientific},
  title     = {Composite fermions in the Hilbert space of the lowest electronic Landau level},
  volume    = {11},
  year      = {1997}
}

@article{ledwith2023vortexability,
  author    = {Ledwith, Patrick J and Vishwanath, Ashvin and Parker, Daniel E},
  journal   = {Physical Review B},
  number    = {20},
  pages     = {205144},
  publisher = {APS},
  title     = {Vortexability: A unifying criterion for ideal fractional {Chern} insulators},
  volume    = {108},
  year      = {2023}
}

@article{roy2014band,
  author    = {Roy, Rahul},
  journal   = {Physical Review B},
  number    = {16},
  pages     = {165139},
  publisher = {APS},
  title     = {Band geometry of fractional topological insulators},
  volume    = {90},
  year      = {2014}
}

@article{wang2021exact,
  author    = {Wang, Jie and Cano, Jennifer and Millis, Andrew J and Liu, Zhao and Yang, Bo},
  journal   = {Physical review letters},
  number    = {24},
  pages     = {246403},
  publisher = {APS},
  title     = {Exact Landau level description of geometry and interaction in a flatband},
  volume    = {127},
  year      = {2021}
}

@article{ledwith2020fractional,
  author    = {Ledwith, Patrick J and Tarnopolsky, Grigory and Khalaf, Eslam and Vishwanath, Ashvin},
  journal   = {Physical Review Research},
  number    = {2},
  pages     = {023237},
  publisher = {APS},
  title     = {Fractional {Chern} insulator states in twisted bilayer graphene: An analytical approach},
  volume    = {2},
  year      = {2020}
}

@article{parameswaran2012fractional,
  author    = {Parameswaran, SA and Roy, R and Sondhi, Shivaji L},
  journal   = {Physical Review B—Condensed Matter and Materials Physics},
  number    = {24},
  pages     = {241308},
  publisher = {APS},
  title     = {Fractional {Chern} insulators and the {$W_\infty$} algebra},
  volume    = {85},
  year      = {2012}
}

@article{qi2011generic,
  author    = {Qi, Xiao-Liang},
  journal   = {Physical review letters},
  number    = {12},
  pages     = {126803},
  publisher = {APS},
  title     = {Generic Wave-Function Description of Fractional Quantum Anomalous {Hall} States and Fractional Topological Insulators},
  volume    = {107},
  year      = {2011}
}

@article{murthy2003hamiltonian,
  author    = {Murthy, Ganpathy and Shankar, R},
  journal   = {Reviews of Modern Physics},
  number    = {4},
  pages     = {1101},
  publisher = {APS},
  title     = {Hamiltonian theories of the fractional quantum {Hall} effect},
  volume    = {75},
  year      = {2003}
}

@article{shankar1999hamiltonian,
  author    = {Shankar, R},
  journal   = {Physical review letters},
  number    = {12},
  pages     = {2382},
  publisher = {APS},
  title     = {Hamiltonian description of composite fermions: aftermath},
  volume    = {83},
  year      = {1999}
}

@article{murthy2012hamiltonian,
  author    = {Murthy, Ganpathy and Shankar, R},
  journal   = {Physical Review B—Condensed Matter and Materials Physics},
  number    = {19},
  pages     = {195146},
  publisher = {APS},
  title     = {Hamiltonian theory of fractionally filled {Chern} bands},
  volume    = {86},
  year      = {2012}
}

@article{chen2010local,
  author    = {Chen, Xie and Gu, Zheng-Cheng and Wen, Xiao-Gang},
  journal   = {Physical Review B—Condensed Matter and Materials Physics},
  number    = {15},
  pages     = {155138},
  publisher = {APS},
  title     = {Local unitary transformation, long-range quantum entanglement, wave function renormalization, and topological order},
  volume    = {82},
  year      = {2010}
}

@article{arovas1984fractional,
  author    = {Arovas, Daniel and Schrieffer, John R and Wilczek, Frank},
  journal   = {Physical review letters},
  number    = {7},
  pages     = {722},
  publisher = {APS},
  title     = {Fractional statistics and the quantum {Hall} effect},
  volume    = {53},
  year      = {1984}
}

@article{wen2017colloquium,
  author    = {Wen, Xiao-Gang},
  journal   = {Reviews of Modern Physics},
  number    = {4},
  pages     = {041004},
  publisher = {APS},
  title     = {Colloquium: Zoo of quantum-topological phases of matter},
  volume    = {89},
  year      = {2017}
}

@article{kitaev2006topological,
  author    = {Kitaev, Alexei and Preskill, John},
  journal   = {Physical review letters},
  number    = {11},
  pages     = {110404},
  publisher = {APS},
  title     = {Topological entanglement entropy},
  volume    = {96},
  year      = {2006}
}

@article{levin2005string,
  author    = {Levin, Michael A and Wen, Xiao-Gang},
  journal   = {Physical Review B—Condensed Matter and Materials Physics},
  number    = {4},
  pages     = {045110},
  publisher = {APS},
  title     = {String-net condensation: A physical mechanism for topological phases},
  volume    = {71},
  year      = {2005}
}

@article{liu2025fractional,
  author  = {Liu, Feng and Xu, Fan and Xu, Cheng and Li, Jiayi and Sun, Zheng and Xiao, Jiayong and Mao, Ning and Chang, Xumin and Tao, Xinglin and Watanabe, Kenji and others},
  journal = {arXiv preprint arXiv:2512.03622},
  title   = {From fractional {Chern} insulators to topological electronic crystals in moir\'e {MoTe$_2$}: quantum geometry tuning via remote layer},
  year    = {2025}
}

@article{chen2026fractional,
  author    = {Chen, Jialin and Li, Qiaoyi and Wang, Xiaoyu and Li, Wei},
  journal   = {Science Bulletin},
  publisher = {Elsevier},
  title     = {Fractional {Chern} insulator and quantum anomalous {Hall} crystal in twisted {MoTe$_2$}},
  year      = {2026}
}

@article{xu2025signatures,
  author  = {Xu, Fan and Sun, Zheng and Li, Jiayi and Zheng, Ce and Xu, Cheng and Gao, Jingjing and Jia, Tongtong and Su, Yanfei and Watanabe, Kenji and Taniguchi, Takashi and others},
  journal = {arXiv preprint arXiv:2504.06972},
  title   = {Signatures of unconventional superconductivity near reentrant and fractional quantum anomalous {Hall} insulators},
  year    = {2025}
}

@article{xie2021fractional,
  author    = {Xie, Yonglong and Pierce, Andrew T and Park, Jeong Min and Parker, Daniel E and Khalaf, Eslam and Ledwith, Patrick and Cao, Yuan and Lee, Seung Hwan and Chen, Shaowen and Forrester, Patrick R and others},
  journal   = {Nature},
  number    = {7889},
  pages     = {439--443},
  publisher = {Nature Publishing Group UK London},
  title     = {Fractional {Chern} insulators in magic-angle twisted bilayer graphene},
  volume    = {600},
  year      = {2021}
}

@article{aronson2025displacement,
  author    = {Aronson, Samuel H and Han, Tonghang and Lu, Zhengguang and Yao, Yuxuan and Butler, Jackson P and Watanabe, Kenji and Taniguchi, Takashi and Ju, Long and Ashoori, Raymond C},
  journal   = {Physical Review X},
  number    = {3},
  pages     = {031026},
  publisher = {APS},
  title     = {Displacement field-controlled fractional {Chern} insulators and charge density waves in a graphene/hBN moir{\'e} superlattice},
  volume    = {15},
  year      = {2025}
}

@article{sun2026twist,
  author  = {Sun, Zheng and Xu, Fan and Li, Jiayi and Jiang, Yifan and Gao, Jingjing and Xu, Cheng and Jia, Tongtong and Cheng, Kehao and Zhang, Jinyang and Tian, Wanghao and others},
  journal = {arXiv preprint arXiv:2603.16412},
  title   = {Twist-angle evolution from valley-polarized fractional topological phases to valley-degenerate superconductivity in twisted bilayer {MoTe$_2$}},
  year    = {2026}
}

@article{han2026evidence,
  author  = {Han, Tonghang and Butler, Jackson P and Ye, Shenyong and Hua, Zhenqi and Dutta, Surajit and Hadjri, Zach and Wu, Zhenghan and Yang, Jixiang and Seo, Junseok and Pattanakanvijit, Phatthanon and others},
  journal = {arXiv preprint arXiv:2604.00113},
  title   = {Evidence of metallic {Wigner} crystal in rhombohedral graphene},
  year    = {2026}
}

@article{anderson2024trion,
  author    = {Anderson, Eric and Cai, Jiaqi and Reddy, Aidan P and Park, Heonjoon and Holtzmann, William and Davis, Kai and Taniguchi, Takashi and Watanabe, Kenji and Smolenski, Tomasz and Imamo{\u{g}}lu, Ata{\c{c}} and others},
  journal   = {Nature},
  number    = {8039},
  pages     = {590--595},
  publisher = {Nature Publishing Group UK London},
  title     = {Trion sensing of a zero-field composite {Fermi} liquid},
  volume    = {635},
  year      = {2024}
}

@article{wang2024fractional,
  author    = {Wang, Chong and Zhang, Xiao-Wei and Liu, Xiaoyu and He, Yuchi and Xu, Xiaodong and Ran, Ying and Cao, Ting and Xiao, Di},
  journal   = {Physical Review Letters},
  number    = {3},
  pages     = {036501},
  publisher = {APS},
  title     = {Fractional {Chern} insulator in twisted bilayer {MoTe$_2$}},
  volume    = {132},
  year      = {2024}
}

@article{reddy2023fractional,
  author    = {Reddy, Aidan P and Alsallom, Faisal and Zhang, Yang and Devakul, Trithep and Fu, Liang},
  journal   = {Physical Review B},
  number    = {8},
  pages     = {085117},
  publisher = {APS},
  title     = {Fractional quantum anomalous {Hall} states in twisted bilayer {MoTe$_2$} and {WSe$_2$}},
  volume    = {108},
  year      = {2023}
}

@article{wilhelm2021interplay,
  author    = {Wilhelm, Patrick and Lang, Thomas C and L{\"a}uchli, Andreas M},
  journal   = {Physical Review B},
  number    = {12},
  pages     = {125406},
  publisher = {APS},
  title     = {Interplay of fractional {Chern} insulator and charge density wave phases in twisted bilayer graphene},
  volume    = {103},
  year      = {2021}
}

@article{dong2023composite,
  author    = {Dong, Junkai and Wang, Jie and Ledwith, Patrick J and Vishwanath, Ashvin and Parker, Daniel E},
  journal   = {Physical Review Letters},
  number    = {13},
  pages     = {136502},
  publisher = {APS},
  title     = {Composite {Fermi} liquid at zero magnetic field in twisted {MoTe$_2$}},
  volume    = {131},
  year      = {2023}
}

@article{pasquier1998dipole,
  author    = {Pasquier, V and Haldane, FDM},
  journal   = {Nuclear Physics B},
  number    = {3},
  pages     = {719--726},
  publisher = {Elsevier},
  title     = {A dipole interpretation of the $\nu=1/2$ state},
  volume    = {516},
  year      = {1998}
}

@article{read1994theory,
  author  = {Read, Nicholas},
  journal = {Semiconductor Science and Technology},
  number  = {11S},
  pages   = {1859--1864},
  title   = {Theory of the half-filled {Landau} level},
  volume  = {9},
  year    = {1994}
}

@article{wen2002quantum,
  author    = {Wen, Xiao-Gang},
  journal   = {Physical Review B},
  number    = {16},
  pages     = {165113},
  publisher = {APS},
  title     = {Quantum orders and symmetric spin liquids},
  volume    = {65},
  year      = {2002}
}

@article{ran2007projected,
  author    = {Ran, Ying and Hermele, Michael and Lee, Patrick A and Wen, Xiao-Gang},
  journal   = {Physical review letters},
  number    = {11},
  pages     = {117205},
  publisher = {APS},
  title     = {Projected-wave-function study of the spin-1/2 {Heisenberg} model on the kagom{\'e} lattice},
  volume    = {98},
  year      = {2007}
}

@article{anderson1987resonating,
  author    = {Anderson, Philip W},
  journal   = {science},
  number    = {4793},
  pages     = {1196--1198},
  publisher = {American Association for the Advancement of Science},
  title     = {The resonating valence bond state in {La$_2$CuO$_4$} and superconductivity},
  volume    = {235},
  year      = {1987}
}

@article{zhang1988renormalised,
  author  = {Zhang, Fu-Chen and Gros, Claudius and Rice, T Maurice and Shiba, Hiroyuki},
  journal = {Superconductor Science and Technology},
  number  = {1},
  pages   = {36--46},
  title   = {A renormalised {Hamiltonian} approach to a resonant valence bond wavefunction},
  volume  = {1},
  year    = {1988}
}

@article{edegger2007gutzwiller,
  author    = {Edegger, Bernhard and Muthukumar, Vangal N and Gros, Claudius},
  journal   = {Advances in Physics},
  number    = {6},
  pages     = {927--1033},
  publisher = {Taylor \& Francis},
  title     = {Gutzwiller-RVB theory of high-temperature superconductivity: Results from renormalized mean-field theory and variational {Monte Carlo} calculations},
  volume    = {56},
  year      = {2007}
}

@article{gutzwiller1963effect,
  author    = {Gutzwiller, Martin C},
  journal   = {Physical Review Letters},
  number    = {5},
  pages     = {159},
  publisher = {APS},
  title     = {Effect of correlation on the ferromagnetism of transition metals},
  volume    = {10},
  year      = {1963}
}

@article{gutzwiller1965correlation,
  author    = {Gutzwiller, Martin C},
  journal   = {Physical Review},
  number    = {6A},
  pages     = {A1726},
  publisher = {APS},
  title     = {Correlation of electrons in a narrow $s$ band},
  volume    = {137},
  year      = {1965}
}

@article{barvinok1995new,
  author    = {Barvinok, Alexander I},
  journal   = {Mathematical Programming},
  number    = {1},
  pages     = {449--470},
  publisher = {Springer},
  title     = {New algorithms for linear k-matroid intersection and matroid k-parity problems},
  volume    = {69},
  year      = {1995}
}

@incollection{gelfand1994hyperdeterminants,
  author    = {Gelfand, Israel M and Kapranov, Mikhail M and Zelevinsky, Andrei V},
  booktitle = {Discriminants, Resultants, and Multidimensional Determinants},
  pages     = {444--479},
  publisher = {Springer},
  title     = {Hyperdeterminants},
  year      = {1994}
}

@inproceedings{gurvits2005complexity,
  author       = {Gurvits, Leonid},
  booktitle    = {International Symposium on Mathematical Foundations of Computer Science},
  organization = {Springer},
  pages        = {447--458},
  title        = {On the complexity of mixed discriminants and related problems},
  year         = {2005}
}

@article{hillar2013most,
  author    = {Hillar, Christopher J and Lim, Lek-Heng},
  journal   = {Journal of the ACM (JACM)},
  number    = {6},
  pages     = {1--39},
  publisher = {ACM New York, NY, USA},
  title     = {Most tensor problems are {NP}-hard},
  volume    = {60},
  year      = {2013}
}

@article{wang2010schwinger,
  author    = {Wang, Fa},
  journal   = {Physical Review B—Condensed Matter and Materials Physics},
  number    = {2},
  pages     = {024419},
  publisher = {APS},
  title     = {Schwinger boson mean field theories of spin liquid states on a honeycomb lattice: Projective symmetry group analysis and critical field theory},
  volume    = {82},
  year      = {2010}
}

@article{wang2006spin,
  author    = {Wang, Fa and Vishwanath, Ashvin},
  journal   = {Physical Review B—Condensed Matter and Materials Physics},
  number    = {17},
  pages     = {174423},
  publisher = {APS},
  title     = {Spin-liquid states on the triangular and Kagom{\'e} lattices: A projective-symmetry-group analysis of {Schwinger} boson states},
  volume    = {74},
  year      = {2006}
}

@article{balram2019parton,
  author    = {Balram, Ajit C and Barkeshli, Maissam and Rudner, Mark S},
  journal   = {Physical Review B},
  number    = {24},
  pages     = {241108},
  publisher = {APS},
  title     = {Parton construction of particle-hole-conjugate Read-{Rezayi} parafermion fractional quantum {Hall} states and beyond},
  volume    = {99},
  year      = {2019}
}

@article{balram2018parton,
  author    = {Balram, Ajit C and Barkeshli, Maissam and Rudner, Mark S},
  journal   = {Physical Review B},
  number    = {3},
  pages     = {035127},
  publisher = {APS},
  title     = {Parton construction of a wave function in the anti-{Pfaffian} phase},
  volume    = {98},
  year      = {2018}
}

@article{kolda2009tensor,
  author    = {Kolda, Tamara G and Bader, Brett W},
  journal   = {SIAM review},
  number    = {3},
  pages     = {455--500},
  publisher = {SIAM},
  title     = {Tensor decompositions and applications},
  volume    = {51},
  year      = {2009}
}

@article{sorella1998green,
  author    = {Sorella, Sandro},
  journal   = {Physical review letters},
  number    = {20},
  pages     = {4558},
  publisher = {APS},
  title     = {Green function {Monte Carlo} with stochastic reconfiguration},
  volume    = {80},
  year      = {1998}
}

@article{sorella2001generalized,
  author    = {Sorella, Sandro},
  journal   = {Physical Review B},
  number    = {2},
  pages     = {024512},
  publisher = {APS},
  title     = {Generalized {Lanczos} algorithm for variational quantum {Monte Carlo}},
  volume    = {64},
  year      = {2001}
}

@article{chen2024empowering,
  author    = {Chen, Ao and Heyl, Markus},
  journal   = {Nature Physics},
  number    = {9},
  pages     = {1476--1481},
  publisher = {Nature Publishing Group UK London},
  title     = {Empowering deep neural quantum states through efficient optimization},
  volume    = {20},
  year      = {2024}
}

@article{park2020geometry,
  author    = {Park, Chae-Yeun and Kastoryano, Michael J},
  journal   = {Physical Review Research},
  number    = {2},
  pages     = {023232},
  publisher = {APS},
  title     = {Geometry of learning neural quantum states},
  volume    = {2},
  year      = {2020}
}

@article{wen1995topological,
  author    = {Wen, Xiao-Gang},
  journal   = {Advances in Physics},
  number    = {5},
  pages     = {405--473},
  publisher = {Taylor \& Francis},
  title     = {Topological orders and edge excitations in fractional quantum {Hall} states},
  volume    = {44},
  year      = {1995}
}

@article{foulkes2001quantum,
  author    = {Foulkes, William MC and Mitas, Lubos and Needs, RJ and Rajagopal, Guna},
  journal   = {Reviews of Modern Physics},
  number    = {1},
  pages     = {33},
  publisher = {APS},
  title     = {Quantum {Monte Carlo} simulations of solids},
  volume    = {73},
  year      = {2001}
}

@article{gros1989physics,
  author    = {Gros, Claudius},
  journal   = {Annals of Physics},
  number    = {1},
  pages     = {53--88},
  publisher = {Elsevier},
  title     = {Physics of projected wavefunctions},
  volume    = {189},
  year      = {1989}
}

@article{barkeshli2019symmetry,
  author    = {Barkeshli, Maissam and Bonderson, Parsa and Cheng, Meng and Wang, Zhenghan},
  journal   = {Physical Review B},
  number    = {11},
  pages     = {115147},
  publisher = {APS},
  title     = {Symmetry fractionalization, defects, and gauging of topological phases},
  volume    = {100},
  year      = {2019}
}

@article{essin2013classifying,
  author    = {Essin, Andrew M and Hermele, Michael},
  journal   = {Physical Review B—Condensed Matter and Materials Physics},
  number    = {10},
  pages     = {104406},
  publisher = {APS},
  title     = {Classifying fractionalization: Symmetry classification of gapped $\mathbb Z_2$ spin liquids in two dimensions},
  volume    = {87},
  year      = {2013}
}

@article{orus2014practical,
  author    = {Or{\'u}s, Rom{\'a}n},
  journal   = {Annals of physics},
  pages     = {117--158},
  publisher = {Elsevier},
  title     = {A practical introduction to tensor networks: Matrix product states and projected entangled pair states},
  volume    = {349},
  year      = {2014}
}

@article{eisert2010colloquium,
  author    = {Eisert, Jens and Cramer, Marcus and Plenio, Martin B},
  journal   = {Reviews of modern physics},
  number    = {1},
  pages     = {277--306},
  publisher = {APS},
  title     = {Colloquium: Area laws for the entanglement entropy},
  volume    = {82},
  year      = {2010}
}

@article{hastings2006spectral,
  author    = {Hastings, Matthew B and Koma, Tohru},
  journal   = {Communications in mathematical physics},
  number    = {3},
  pages     = {781--804},
  publisher = {Springer},
  title     = {Spectral gap and exponential decay of correlations},
  volume    = {265},
  year      = {2006}
}

@article{wolf2006violation,
  author    = {Wolf, Michael M},
  journal   = {Physical review letters},
  number    = {1},
  pages     = {010404},
  publisher = {APS},
  title     = {Violation of the entropic area law for fermions},
  volume    = {96},
  year      = {2006}
}

@article{gioev2006entanglement,
  author    = {Gioev, Dimitri and Klich, Israel},
  journal   = {Physical review letters},
  number    = {10},
  pages     = {100503},
  publisher = {APS},
  title     = {Entanglement entropy of fermions in any dimension and the {Widom} conjecture},
  volume    = {96},
  year      = {2006}
}

@article{swingle2010entanglement,
  author    = {Swingle, Brian},
  journal   = {Physical review letters},
  number    = {5},
  pages     = {050502},
  publisher = {APS},
  title     = {Entanglement entropy and the {Fermi} surface},
  volume    = {105},
  year      = {2010}
}

@article{hastings2007area,
  author  = {Hastings, Matthew B},
  journal = {Journal of statistical mechanics: theory and experiment},
  number  = {08},
  pages   = {P08024--P08024},
  title   = {An area law for one-dimensional quantum systems},
  volume  = {2007},
  year    = {2007}
}

@article{shao2015entanglement,
  author    = {Shao, Junping and Kim, Eun-Ah and Haldane, FDM and Rezayi, Edward H},
  journal   = {Physical review letters},
  number    = {20},
  pages     = {206402},
  publisher = {APS},
  title     = {Entanglement entropy of the $\nu=1/2$ composite fermion non-{Fermi} liquid state},
  volume    = {114},
  year      = {2015}
}

@article{mishmash2016entanglement,
  author    = {Mishmash, Ryan V and Motrunich, Olexei I},
  journal   = {Physical Review B},
  number    = {8},
  pages     = {081110},
  publisher = {APS},
  title     = {Entanglement entropy of composite {Fermi} liquid states on the lattice: In support of the {Widom} formula},
  volume    = {94},
  year      = {2016}
}

@article{halperin1993theory,
  author    = {Halperin, Bertrand I and Lee, Patrick A and Read, Nicholas},
  journal   = {Physical Review B},
  number    = {12},
  pages     = {7312},
  publisher = {APS},
  title     = {Theory of the half-filled {Landau} level},
  volume    = {47},
  year      = {1993}
}

@article{rezayi1994fermi,
  author    = {Rezayi, E and Read, N},
  journal   = {Physical review letters},
  number    = {6},
  pages     = {900},
  publisher = {APS},
  title     = {{Fermi}-liquid-like state in a half-filled {Landau} level},
  volume    = {72},
  year      = {1994}
}

@article{levin2007tensor,
  author    = {Levin, Michael and Nave, Cody P},
  journal   = {Physical review letters},
  number    = {12},
  pages     = {120601},
  publisher = {APS},
  title     = {Tensor renormalization group approach to two-dimensional classical lattice models},
  volume    = {99},
  year      = {2007}
}

@article{zhang2026projected,
  author  = {Zhang, Hang and Armegioiu, Victor and Carrasquilla, Juan and Mishra, Siddhartha and M{\"u}ller, Johannes and Nys, Jannes and Zeinhofer, Marius},
  journal = {arXiv preprint arXiv:2606.07825},
  title   = {Projected Inverse Iteration: An Eigenvalue Approach to Ground-State Computation with Neural Quantum States},
  year    = {2026}
}

@article{shi2026effects,
  author    = {Shi, Jingtian and Cano, Jennifer and Morales-Dur{\'a}n, Nicol{\'a}s},
  journal   = {Physical Review Research},
  number    = {2},
  pages     = {L022045},
  publisher = {APS},
  title     = {Effects of {Berry} curvature on ideal fractional {Chern} insulator many-body gaps},
  volume    = {8},
  year      = {2026}
}

@article{wu2013bloch,
  author    = {Wu, Yang-Le and Regnault, Nicolas and Bernevig, B Andrei},
  journal   = {Physical review letters},
  number    = {10},
  pages     = {106802},
  publisher = {APS},
  title     = {Bloch model wave functions and pseudopotentials for all fractional Chern insulators},
  volume    = {110},
  year      = {2013}
}

@article{wu2012gauge,
  author    = {Wu, Yang-Le and Regnault, Nicolas and Bernevig, B Andrei},
  journal   = {Physical Review B—Condensed Matter and Materials Physics},
  number    = {8},
  pages     = {085129},
  publisher = {APS},
  title     = {Gauge-fixed Wannier wave functions for fractional topological insulators},
  volume    = {86},
  year      = {2012}
}

@article{han2025signatures,
  author    = {Han, Tonghang and Lu, Zhengguang and Hadjri, Zach and Shi, Lihan and Wu, Zhenghan and Xu, Wei and Yao, Yuxuan and Cotten, Armel A and Sharifi Sedeh, Omid and Weldeyesus, Henok and others},
  journal   = {Nature},
  number    = {8072},
  pages     = {654--661},
  publisher = {Nature Publishing Group UK London},
  title     = {Signatures of chiral superconductivity in rhombohedral graphene},
  volume    = {643},
  year      = {2025}
}

@article{bernevig2008model,
  author    = {Bernevig, B Andrei and Haldane, FDM},
  journal   = {Physical review letters},
  number    = {24},
  pages     = {246802},
  publisher = {APS},
  title     = {Model fractional quantum {Hall} states and {Jack} polynomials},
  volume    = {100},
  year      = {2008}
}

@article{sterdyniak2011extracting,
  author    = {Sterdyniak, A and Regnault, N and Bernevig, B Andrei},
  journal   = {Physical review letters},
  number    = {10},
  pages     = {100405},
  publisher = {APS},
  title     = {Extracting excitations from model state entanglement},
  volume    = {106},
  year      = {2011}
}

@article{chandran2011bulk,
  author    = {Chandran, Anushya and Hermanns, M and Regnault, N and Bernevig, B Andrei},
  journal   = {Physical Review B—Condensed Matter and Materials Physics},
  number    = {20},
  pages     = {205136},
  publisher = {APS},
  title     = {Bulk-edge correspondence in entanglement spectra},
  volume    = {84},
  year      = {2011}
}

@article{williamson2014matrix,
  author  = {Williamson, Dominic J and Bultinck, Nick and Mari{\"e}n, Michael and Sahinoglu, Mehmet B and Haegeman, Jutho and Verstraete, Frank},
  journal = {arXiv preprint arXiv:1412.5604},
  title   = {Matrix product operators for symmetry-protected topological phases: Gauging and edge theories},
  year    = {2014}
}

@article{chen2013symmetry,
  author    = {Chen, Xie and Gu, Zheng-Cheng and Liu, Zheng-Xin and Wen, Xiao-Gang},
  journal   = {Physical Review B—Condensed Matter and Materials Physics},
  number    = {15},
  pages     = {155114},
  publisher = {APS},
  title     = {Symmetry protected topological orders and the group cohomology of their symmetry group},
  volume    = {87},
  year      = {2013}
}

@article{chen2011classification,
  author    = {Chen, Xie and Gu, Zheng-Cheng and Wen, Xiao-Gang},
  journal   = {Physical Review B—Condensed Matter and Materials Physics},
  number    = {3},
  pages     = {035107},
  publisher = {APS},
  title     = {Classification of gapped symmetric phases in one-dimensional spin systems},
  volume    = {83},
  year      = {2011}
}

@article{perez2008string,
  author    = {P{\'e}rez-Garc{\'\i}a, David and Wolf, Michael M and Sanz, Mikel and Verstraete, Frank and Cirac, J Ignacio},
  journal   = {Physical review letters},
  number    = {16},
  pages     = {167202},
  publisher = {APS},
  title     = {String order and symmetries in quantum spin lattices},
  volume    = {100},
  year      = {2008}
}

@article{schuch2007computational,
  author    = {Schuch, Norbert and Wolf, Michael M and Verstraete, Frank and Cirac, J Ignacio},
  journal   = {Physical review letters},
  number    = {14},
  pages     = {140506},
  publisher = {APS},
  title     = {Computational complexity of projected entangled pair states},
  volume    = {98},
  year      = {2007}
}

@article{xiao2011interface,
  author    = {Xiao, Di and Zhu, Wenguang and Ran, Ying and Nagaosa, Naoto and Okamoto, Satoshi},
  journal   = {Nature communications},
  number    = {1},
  pages     = {596},
  publisher = {Nature Publishing Group UK London},
  title     = {Interface engineering of quantum {Hall} effects in digital transition metal oxide heterostructures},
  volume    = {2},
  year      = {2011}
}

@article{abouelkomsan2024band,
  author    = {Abouelkomsan, Ahmed and Reddy, Aidan P and Fu, Liang and Bergholtz, Emil J},
  journal   = {Physical Review B},
  number    = {12},
  pages     = {L121107},
  publisher = {APS},
  title     = {Band mixing in the quantum anomalous {Hall} regime of twisted semiconductor bilayers},
  volume    = {109},
  year      = {2024}
}

@article{hou2025stabilizing,
  author  = {Hou, Run and Nevidomskyy, Andriy H},
  journal = {arXiv preprint arXiv:2511.16641},
  title   = {Stabilizing Fractional Chern States in Twisted {MoTe$_2$}: Multi-band Correlations via Non-perturbative Renormalization Group},
  year    = {2025}
}

@article{dong2024anomalous,
  author    = {Dong, Junkai and Wang, Taige and Wang, Tianle and Soejima, Tomohiro and Zaletel, Michael P and Vishwanath, Ashvin and Parker, Daniel E},
  journal   = {Physical Review Letters},
  number    = {20},
  pages     = {206503},
  publisher = {APS},
  title     = {Anomalous {Hall} crystals in rhombohedral multilayer graphene. I. Interaction-driven Chern bands and fractional quantum {Hall} states at zero magnetic field},
  volume    = {133},
  year      = {2024}
}

@article{yu2025moireIV,
  author    = {Yu, Jiabin and Herzog-Arbeitman, Jonah and Kwan, Yves H and Regnault, Nicolas and Bernevig, B Andrei},
  journal   = {Physical Review B},
  number    = {7},
  pages     = {075110},
  publisher = {APS},
  title     = {Moir{\'e} fractional {Chern} insulators. IV. Fluctuation-driven collapse in multiband exact diagonalization calculations on rhombohedral graphene},
  volume    = {112},
  year      = {2025}
}

@article{li2025multiband,
  author    = {Li, Heqiu and Bernevig, B Andrei and Regnault, Nicolas},
  journal   = {Physical Review B},
  number    = {7},
  pages     = {075130},
  publisher = {APS},
  title     = {Multiband exact diagonalization and an iteration approach to search for fractional {Chern} insulators in rhombohedral multilayer graphene},
  volume    = {112},
  year      = {2025}
}

@article{xu2024maximally,
  author    = {Xu, Cheng and Li, Jiangxu and Xu, Yong and Bi, Zhen and Zhang, Yang},
  journal   = {Proceedings of the National Academy of Sciences},
  number    = {8},
  pages     = {e2316749121},
  publisher = {National Academy of Sciences},
  title     = {Maximally localized Wannier functions, interaction models, and fractional quantum anomalous {Hall} effect in twisted bilayer {MoTe$_2$}},
  volume    = {121},
  year      = {2024}
}

@article{he2025fractional,
  author  = {He, Yuchi and Simon, SH and Parameswaran, SA},
  journal = {arXiv preprint arXiv:2505.06354},
  title   = {Fractional {Chern} insulators and competing states in a twisted {MoTe$_2$} lattice model},
  year    = {2025}
}

@article{gonccalves2026spinless,
  author    = {Gon{\c{c}}alves, Miguel and Mendez-Valderrama, Juan Felipe and Herzog-Arbeitman, Jonah and Yu, Jiabin and Xu, Xiaodong and Xiao, Di and Bernevig, B Andrei and Regnault, Nicolas},
  journal   = {Physical review letters},
  number    = {19},
  pages     = {196503},
  publisher = {APS},
  title     = {Spinless and spinful charge excitations in moir{\'e} fractional {Chern} insulators},
  volume    = {136},
  year      = {2026}
}

@article{claassen2015position,
  author    = {Claassen, Martin and Lee, Ching Hua and Thomale, Ronny and Qi, Xiao-Liang and Devereaux, Thomas P},
  journal   = {Physical review letters},
  number    = {23},
  pages     = {236802},
  publisher = {APS},
  title     = {Position-momentum duality and fractional quantum {Hall} effect in {Chern} insulators},
  volume    = {114},
  year      = {2015}
}

@article{jian2013crystal,
  author    = {Jian, Chao-Ming and Qi, Xiao-Liang},
  journal   = {Physical Review B—Condensed Matter and Materials Physics},
  number    = {16},
  pages     = {165134},
  publisher = {APS},
  title     = {Crystal-symmetry preserving Wannier states for fractional {Chern} insulators},
  volume    = {88},
  year      = {2013}
}

@article{zhang2025beyond,
  author  = {Zhang, Yitong and Sarkar, Siddhartha and Wan, Xiaohan and Parker, Daniel E and Lin, Shi-Zeng and Sun, Kai},
  journal = {arXiv preprint arXiv:2510.22831},
  title   = {Beyond the Lowest {Landau} Level: Unlocking More Robust Fractional States Using Flat {Chern} Bands with Higher Vortexability},
  year    = {2025}
}

@article{hastings2011topological,
  author    = {Hastings, Matthew B and Loring, Terry A},
  journal   = {Annals of Physics},
  number    = {7},
  pages     = {1699--1759},
  publisher = {Elsevier},
  title     = {Topological insulators and {$C^*$}-algebras: Theory and numerical practice},
  volume    = {326},
  year      = {2011}
}

@article{hastings2010almost,
  author    = {Hastings, Matthew B and Loring, Terry A},
  journal   = {Journal of Mathematical Physics},
  number    = {1},
  publisher = {AIP Publishing},
  title     = {Almost commuting matrices, localized Wannier functions, and the quantum {Hall} effect},
  volume    = {51},
  year      = {2010}
}

@article{sarkar2026similar,
  author  = {Sarkar, Siddhartha and Zhang, Yitong and Sun, Kai},
  journal = {arXiv preprint arXiv:2606.07323},
  title   = {How Similar Can Fractional {Chern} Insulators Be to Fractional Quantum {Hall} States? Moir{\'e}-Enhanced Gaps and Excitation-Spectrum Correspondence},
  year    = {2026}
}

@article{grushin2012enhancing,
  author    = {Grushin, Adolfo G and Neupert, Titus and Chamon, Claudio and Mudry, Christopher},
  journal   = {Physical Review B—Condensed Matter and Materials Physics},
  number    = {20},
  pages     = {205125},
  publisher = {APS},
  title     = {Enhancing the stability of a fractional Chern insulator against competing phases},
  volume    = {86},
  year      = {2012}
}

@article{goldman2023zero,
  author    = {Goldman, Hart and Reddy, Aidan P and Paul, Nisarga and Fu, Liang},
  journal   = {Physical review letters},
  number    = {13},
  pages     = {136501},
  publisher = {APS},
  title     = {Zero-field composite Fermi liquid in twisted semiconductor bilayers},
  volume    = {131},
  year      = {2023}
}

@article{seo2026family,
  author    = {Seo, Junseok and Cotten, Armel A and Ye, Shenyong and Xu, Mingchi and Sedeh, Omid Sharifi and Weldeyesus, Henok and Han, Tonghang and Lu, Zhengguang and Wu, Zhenghan and Xu, Wei and others},
  journal   = {Nature},
  pages     = {1--3},
  publisher = {Nature Publishing Group UK London},
  title     = {Family of magnetic field-boosted superconductors in rhombohedral graphene},
  year      = {2026}
}

@article{hua2026multi,
  author  = {Hua, Zhenqi and Ye, Shenyong and Pattanakanvijit, Phatthanon and Shi, Gang and Han, Tonghang and Aitken, Emily and Yang, Jixiang and Seo, Junseok and Liu, Haoyang and Hao, Ran and others},
  journal = {arXiv preprint arXiv:2607.06520},
  title   = {Multi-Knob Switchable Chiral Superconductivity Quartet in Rhombohedral Graphene},
  year    = {2026}
}

@article{nguyen2025hierarchy,
  author  = {Nguyen, Ron Q and Wu, Hai-Tian and Morissette, Erin and Zhang, Naiyuan J and Qin, Peiyu and Watanabe, Kenji and Taniguchi, Takashi and Hui, Aaron W and Feldman, Dima E and Li, JIA},
  journal = {arXiv preprint arXiv:2507.22026},
  title   = {A hierarchy of superconductivity and topological charge density wave states in rhombohedral graphene},
  year    = {2025}
}

@article{qin2025stripe,
  author  = {Qin, Peiyu and Wu, Hai-Tian and Nguyen, Ron Q and Morissette, Erin and Zhang, Naiyuan J and Watanabe, K and Taniguchi, T and Li, JIA},
  journal = {arXiv preprint arXiv:2504.05129},
  title   = {Stripe order in the metallic and superconducting phases of rhombohedral hexalayer graphene},
  year    = {2025}
}

@article{tucker1966some,
  author    = {Tucker, Ledyard R},
  journal   = {Psychometrika},
  number    = {3},
  pages     = {279--311},
  publisher = {Springer},
  title     = {Some mathematical notes on three-mode factor analysis},
  volume    = {31},
  year      = {1966}
}

@misc{lin2026superfluid-FCI,
  author = {Lin, Guan-Lin and Hu, Xiaodong and Xiao, Di and Ran, Ying},
  note   = {to be submitted}
}

@article{lin2026hyperdeterminant,
  author  = {Lin, Guan-Lin and Xiao, Di and Ran, Ying},
  journal = {arXiv preprint arXiv:2607.23392},
  title   = {Hyperdeterminant wavefunctions},
  year    = {2026}
}

@misc{Appendix,
  note = {See appendix at {\color{mypurple}\textsf{URL\_to\_be\_Inserted\_by\_Publisher}} for (1) HyperDet wavefunction construction details, (2) fusion tensor reconstruction and HyperDet wavefunction determinant expansion, (3) statistical theorem, (4) built-in $\mathrm{GL}(N;\mathbb C)\rtimes S_m$ gauge structure, (5) gauge fixing of stochastic reconfiguration updates, (6) hyperparameter settings and example loss curves, (7) VMC Performance for Other Geometries, (8) virtual-lift fitting result, (9) parton Chern extraction, and (10) exact-diagonalization diagnostics of fermionic phase diagrams. The appendix also cites Ref.~\cite{sorella1998green,sorella2001generalized,chen2024empowering,goldshlager2024kaczmarz,lu2025continuous,sun2011nearly,sheng2011fractional,regnault2011fractional,bernevig2008model,sterdyniak2011extracting,chandran2011bulk,hastings2010almost,hastings2011topological,barkeshli2014continuous,barkeshli2015continuous}.}
}

\clearpage
\appendix

\onecolumngrid

\begin{center}
	{\large\bfseries Appendices for ``HyperDet Wavefunction: A Phase-Agnostic Ansatz for \\ Strongly Correlated Systems''}\\
	\vspace{1.2em}
	{Xiaodong Hu$^{1}$, Guan-Lin Lin$^{2}$, Ying Ran$^{2}$, and Di Xiao$^{1,3}$}\\
	\vspace{0.6em}
	{\small\itshape $^{1}$Department of Materials Science and Engineering, University of Washington, Seattle, WA 98195, USA}\\
	{\small\itshape $^{2}$Department of Physics, Boston College, Chestnut Hill, Massachusetts 02467, USA}\\
	{\small\itshape $^{3}$Department of Physics, University of Washington, Seattle, Washington 98195, USA}\\[0.1em]
	{\small(Dated: \today)}
	\vspace{1.5em}
\end{center}

\twocolumngrid

\appendixtableofcontents 


\section{Definition of Anchored Hyperdeterminant}\label{app:anchored_hyperdet}
For completeness, we give the precise convention used throughout the main text. For a physical configuration $\bm s=(s_1,\ldots,s_N)$, define the configuration-restricted fusion tensor
\begin{equation}
	[\mathcal F_{\bm s}]_{\ell;\ell_1,\ldots,\ell_m}:=\mathcal F_{s_\ell;\ell_1,\ldots,\ell_m},
	\qquad \ell,\ell_\alpha=1,\ldots,N.
\end{equation}
This is an order-$(m+1)$ tensor with one \emph{anchored} physical-slot index and $m$ parton indices. We define its anchored combinatorial hyperdeterminant by
\begin{align}
	\mathop{\mathsf{HyperDet}}[\mathcal F_{\bm s}] &= \sum_{P_1,\ldots,P_m\in S_N} \left[\prod_{\alpha=1}^m \sgn(P_\alpha)\right]\nonumber\\
	&\quad\times\prod_{\ell=1}^N [\mathcal F_{\bm s}]_{\ell;P_1(\ell),\ldots,P_m(\ell)}.
	\label{eqA:anchored_hyperdet_def}
\end{align}
The first index is \emph{not antisymmetrized} in this definition. It labels the physical particle slot, or equivalently the ordered entries of $\bm s$. The wavefunction amplitude is therefore
\begin{equation}
	\Psi_{\mathcal F}^{\text{phys}}(\bm s)=\mathop{\mathsf{HyperDet}}[\mathcal F_{\bm s}] .
	\label{eqA:psi_hyperdet}
\end{equation}
This convention differs significantly from \emph{Cayley's first hyperdeterminant} widely used in the literature~\cite{gelfand1994hyperdeterminants}, where \emph{all} tensor dimensions are antisymmetrized, so the hyperdeterminant vanishes identically for odd dimensions. Here only the parton legs are antisymmetrized, while the physical leg is anchored by the ordered physical configuration.

\section{Koszul Signs of the Particle-Slot Reordering}\label{app:koszul_signs}
In the main text we derive the HyperDet wavefunction in a first-quantized particle-slot convention. In that convention, each parton species $\alpha$ is described by an ordinary Slater determinant over $N$ particle slots, and the relative fermionic signs inside that species are already carried by the Slater antisymmetrizer $\sum_{P_\alpha}(-1)^{P_\alpha}$. The different parton species are then combined by an ordinary tensor product. The map $U_\pi^{\text{slot}}$ used in the main text is simply the linear reordering map that converts this species-major tensor product into a particle-slot-major tensor product. No additional fermionic sign is assigned to this reordering in the main-text convention.

Some readers may nevertheless wonder whether a hidden sign is missing when the parton tensor factors are rearranged. This worry actually comes from a different, equally legitimate, bookkeeping convention --- \emph{the graded fermionic tensor convention} used in second-quantized language. The purpose of this appendix section is to explain the difference between these two conventions. We show that,
\begin{itemize}
	\item \emph{Even if the whole construction is embedded into the graded fermionic convention, the residual extra sign (the Koszul sign) is a global, configuration-independent normalization factor and therefore can be ignored}.
\end{itemize}
Therefore it does not modify the HyperDet wavefunction derived in the main text.

Let us first define the ordinary particle-slot reordering used in the main text. For each parton species $\alpha=1,\ldots,m$ and particle slot $\ell=1,\ldots,N$, let $V_{\alpha,\ell}$ denote the one-particle parton Hilbert space assigned to that species and slot. The species-major tensor product is
\begin{equation}
	\mathcal H_{\text{species}}=\bigotimes_{\alpha=1}^{m}\bigotimes_{\ell=1}^{N}V_{\alpha,\ell},
\end{equation}
while the particle-slot-major tensor product is
\begin{equation}
	\mathcal H_{\text{slot}}=\bigotimes_{\ell=1}^{N}\bigotimes_{\alpha=1}^{m}V_{\alpha,\ell}.
\end{equation}
The map $U_\pi^{\text{slot}}:\mathcal H_{\text{species}}\rightarrow\mathcal H_{\text{slot}}$ is the canonical linear isomorphism induced by the permutation of tensor factors
\begin{equation}
	(\alpha,\ell)\mapsto(\ell,\alpha).
\end{equation}
On a simple tensor it acts as
\begin{align}
	U_\pi^{\text{slot}}\bigg(\bigotimes_{\alpha=1}^{m}\bigotimes_{\ell=1}^{N}v_{\alpha,\ell}\bigg)=\bigotimes_{\ell=1}^{N}\bigotimes_{\alpha=1}^{m}v_{\alpha,\ell}.
\end{align}
Here ``canonical'' means that no basis choice or dynamical assumption is involved; the map only changes the order in which the same tensor factors are written. In the main-text convention, this is an ordinary tensor-product isomorphism and hence carries no extra sign.

Now we contrast this with the graded fermionic convention. In second quantization, the fermionic Fock space is naturally $\mathbb Z_2$-graded by fermion parity. A state or operator is called even if it contains an even number of fermionic creation or annihilation operators, and odd if it contains an odd number. In a graded tensor convention, exchanging two homogeneous objects with parities $|a|,|b|\in\{0,1\}$ produces the so-called \emph{Koszul sign}
\begin{equation}
	a\otimes b\mapsto(-1)^{|a||b|}b\otimes a.
\end{equation}
Since every one-particle parton factor is fermion-odd, exchanging two such factors produces a minus sign. Equivalently, if $f_{\alpha,q}^{\dagger}$ creates the $q$-th occupied orbital of parton species $\alpha$, then all fermionic creation operators anticommute:
\begin{equation}
	f_{\alpha,p}^{\dagger}f_{\beta,q}^{\dagger}=-f_{\beta,q}^{\dagger}f_{\alpha,p}^{\dagger}
\end{equation}
whenever $(\alpha,p)\neq(\beta,q)$.

This graded convention is not the convention used in the first-quantized derivation in the main text. In the main text, the antisymmetry of each species is explicitly imposed by the Slater signs $(-1)^{P_\alpha}$, and different species are organized as separate tensor factors. In the graded second-quantized convention, by contrast, all parton creation operators are embedded into one fermionic algebra, so reordering parton factors from different species also produces signs. The two descriptions are just different bookkeeping conventions. We now show that their difference is a single global sign.

For fixed permutations $\bm P=(P_1,\ldots,P_m)$, the species-major order appearing in the parton mean-field state is represented in second quantization as
\begin{equation}
	|\bm P\rangle_{\text{species}}=\prod_{\alpha=1}^{m}\prod_{\ell=1}^{N}f_{\alpha,P_\alpha(\ell)}^{\dagger}|0\rangle,
\end{equation}
where products are ordered with increasing $\alpha$ first and increasing $\ell$ second. The particle-slot-major order is
\begin{equation}
	|\bm P\rangle_{\text{slot}}=\prod_{\ell=1}^{N}\prod_{\alpha=1}^{m}f_{\alpha,P_\alpha(\ell)}^{\dagger}|0\rangle,
\end{equation}
where products are ordered with increasing $\ell$ first and increasing $\alpha$ second. Passing from the first product to the second product requires a number of adjacent swaps of fermion-odd factors. The corresponding sign is $(-1)^I$, where $I$ is the number of inversions between the two orderings.

A concrete example is useful. For $m=2$ species and $N=2$ particles, the species-major order is
\begin{equation*}
	(1,1),(1,2),(2,1),(2,2),
\end{equation*}
while the particle-slot-major order is
\begin{equation*}
	(1,1),(2,1),(1,2),(2,2).
\end{equation*}
Only one transposition is needed, namely exchanging $(1,2)$ and $(2,1)$. Therefore
\begin{align*}
	&f_{1,P_1(1)}^{\dagger}f_{1,P_1(2)}^{\dagger}f_{2,P_2(1)}^{\dagger}f_{2,P_2(2)}^{\dagger}=\nonumber\\
	&\quad-f_{1,P_1(1)}^{\dagger}f_{2,P_2(1)}^{\dagger}f_{1,P_1(2)}^{\dagger}f_{2,P_2(2)}^{\dagger}.
\end{align*}
This minus sign is a Koszul sign. It is \emph{not} any of the Slater-determinant signs $(-1)^{P_\alpha}$. Instead, it comes only from changing the tensor-product ordering convention.

For general $m$ and $N$, the same counting is straightforward. Consider two factors $(\alpha,\ell)$ and $(\beta,k)$ with $\alpha<\beta$. In species-major order, $(\alpha,\ell)$ always appears before $(\beta,k)$. In particle-slot-major order, the comparison is controlled first by the slot label. Therefore $(\beta,k)$ appears before $(\alpha,\ell)$ precisely when $k<\ell$. Thus each pair of species $\alpha<\beta$ contributes one inversion for each pair of slots $k<\ell$. The total number of inversions is
\begin{equation}
	I_{m,N}=\binom{m}{2}\binom{N}{2}.
\end{equation}
Consequently,
\begin{equation}
	\prod_{\alpha=1}^{m}\prod_{\ell=1}^{N}f_{\alpha,P_\alpha(\ell)}^{\dagger}=(-1)^{\binom{m}{2}\binom{N}{2}}\prod_{\ell=1}^{N}\prod_{\alpha=1}^{m}f_{\alpha,P_\alpha(\ell)}^{\dagger}.
\end{equation}
Equivalently, the graded version of the particle-slot reordering map is related to the ordinary map used in the main text by
\begin{equation}
	U_{\pi,\text{graded}}^{\text{slot}}=(-1)^{\binom{m}{2}\binom{N}{2}}U_\pi^{\text{slot}}.
	\label{eqA:graded fermion sign for particle-slot convention}
\end{equation}

The crucial point is that this sign depends only on a fixed parton species $m$ and a fixed filling physical particles $N$. It does not depend on the physical configuration $\bm s=(s_1,\ldots,s_N)$, and it does not depend on the permutation choices $P_\alpha$. The permutations $P_\alpha$ only decide which occupied parton orbital is placed into each slot; they do not change how many one-particle tensor factors must be crossed when the total ordering is changed from species-major to particle-slot-major.

Applying the particle-slot-wise fusion map in the graded convention therefore gives
\begin{align}
	&\langle\bm s|\left(\bigotimes_{\ell=1}^{N}\hat F_\ell\right)U_{\pi,\text{graded}}^{\text{slot}}|\Psi^{\text{parton}}\rangle\nonumber\\
	&=(-1)^{\binom{m}{2}\binom{N}{2}}\sum_{P_1,\ldots,P_m}\Big[\prod_{\alpha=1}^{m}(-1)^{P_\alpha}\Big]\prod_{\ell=1}^{N}\mathcal F_{s_\ell;P_1(\ell),\ldots,P_m(\ell)}\nonumber\\
	&=(-1)^{\binom{m}{2}\binom{N}{2}}\mathop{\mathsf{HyperDet}}[\mathcal F_{\bm s}].
\end{align}
Thus the first-quantized ordinary tensor-product derivation in the main text and the second-quantized graded fermionic convention produce \emph{the same} physical HyperDet wavefunction up to an unimportant normalization factor. Since this Koszul sign is independent of $\bm s$, it can be safely ignored in the normalized physical state $|\Psi^{\text{phys}}\rangle$.

\section{Proof of the Statistical Theorem}\label{app:statistical_theorem}
We prove the statistical theorem of the main text. The HyperDet wavefunction of configuration $|\bm s\rangle$ reads
\begin{equation}\label{eqA:hyperdet wavefunction}
	\Psi_{\mathcal F}(\bm s)=\sum_{P_1,\ldots,P_m\in S_N}\left[\prod_{\alpha=1}^{m}\sgn(P_\alpha)\right]\prod_{\ell=1}^{N}\mathcal F_{s_\ell;P_1(\ell),\ldots,P_m(\ell)}.
\end{equation}
For a physical-slot permutation $\sigma\in S_N$ that sends the physical configuration to $\sigma(\bm s)=(s_{\sigma(1)},\ldots,s_{\sigma(N)})$, we have
\begin{equation*}
	\Psi_{\mathcal F}(\sigma(\bm s))=\sum_{P_1,\ldots,P_m}\left[\prod_{\alpha=1}^{m}\sgn(P_\alpha)\right]\prod_{\ell=1}^{N}\mathcal F_{s_{\sigma(\ell)};P_1(\ell),\ldots,P_m(\ell)} .
\end{equation*}
Relabel the particle index by $r=\sigma(\ell)$, or $\ell=\sigma^{-1}(r)$, then
\begin{align*}
	\Psi_{\mathcal F}(\sigma(\bm s))&=\sum_{P_1,\ldots,P_m}\left[\prod_{\alpha=1}^{m}\sgn(P_\alpha)\right]\\
	&\quad\times\prod_{r=1}^{N}\mathcal F_{s_r;P_1(\sigma^{-1}(r)),\ldots,P_m(\sigma^{-1}(r))}.
\end{align*}
Define $P'_\alpha=P_\alpha\circ\sigma^{-1}$, then $P_\alpha=P'_\alpha\circ\sigma$, and the permutation signature satisfies
\begin{equation}
	\sgn(P_\alpha)=\sgn(P'_\alpha)\sgn(\sigma).
\end{equation}
Since this occurs for every parton species, we get
\begin{align}
	&\Psi_{\mathcal F}(\sigma(\bm s))=\sgn(\sigma)^m\sum_{P'_1,\ldots,P'_m}\left[\prod_{\alpha=1}^{m}\sgn(P'_\alpha)\right]\nonumber\\
	&\quad\times\prod_{r=1}^{N}\mathcal F_{s_r;P'_1(r),\ldots,P'_m(r)}=\sgn(\sigma)^m \Psi_{\mathcal F}(\bm s).
\end{align}
And we are done with the proof that:
\begin{itemize}
	\item \emph{The physical HyperDet wavefunction~\cref{eqA:hyperdet wavefunction} is symmetric for even $m$ and antisymmetric for odd $m$}.
\end{itemize}

\section{Determinant Expansion of HyperDet Wavefunction}
\subsection{Mode-$\alpha$ Unfolding and Simultaneous Tucker Reconstruction}
To simplify the notation, We will use $s_\ell\equiv s$ denote the physical orbital if there is no ambiguity. The bipartite fusion matrix Eq.~\eqref{eq:bipartite fusion matrix} used in the main text is exactly the mode-$\alpha$ unfolding:
\begin{equation}\label{eqA:mode-alpha unfolding}
	[\bm F_s^{(\alpha)}]_{\ell_\alpha,\ell_{\bar\alpha}} \equiv \mathcal F_{s_\ell;\ell_\alpha\mid\ell_{\bar\alpha}}.
\end{equation}
Note: the mode-$\alpha$ unfolding just reshapes the data of the fusion tensor but never alter it --- the fusion tensor thus can always be uniquely recovered from it. Therefore, below we may interchangeably use the two (seemingly different) notations: unfolded matrix $\bm F_s^{(\alpha)}$, and the original fusion tensor $\mathcal F_s$.

We can always take the mode-$\alpha$ unfolding matrix's site-$s$ SVD
\begin{equation}\label{eqA:SVD of mode-alpha unfolding}
	\bm F_s^{(\alpha)} = \bm U^{(\alpha)}_s \bm\Lambda_s^{(\alpha)} \bm V_s^{(\alpha)\dagger},
\end{equation}
with species-resolved ranks $\bm\Lambda_s^{(\alpha)} = \mathop{\mathrm{diag}}\{\lambda_s^{(1)},\ldots,\lambda_s^{(r_s^{(\alpha)})}\}$, and introduce a support projector
\begin{equation}\label{eqA:support projector}
	\bm\Pi_s^{(\alpha)} := \bm U_s^{(\alpha)} \bm U_s^{(\alpha)\dagger}
\end{equation}

A nice property of the mode-$\alpha$ unfolding~\cref{eqA:mode-alpha unfolding} is that the mode-$\alpha$ product
\begin{equation*}
	(\mathcal F_s\times_\alpha\bm M)_{\ell_1,\ldots,\ell_m} := \sum_{j_\alpha=1}^N \mathcal F_{s; \ell_1,\ldots,j_\alpha,\ldots,\ell_m} [\bm M]_{\ell_\alpha,j_\alpha},
\end{equation*}
reduces to the left matrix product
\begin{equation}
	(\mathcal F_s\times_\alpha\bm M)_{(\alpha)} \equiv \bm M\bm F_s^{(\alpha)}.
\end{equation}
As a result, the fusion tensor is invariant under the mode-$\alpha$ product with the support projector:
\begin{equation}\label{eqA:mode-alpha product invariance under support projector}
	\mathcal F_s \times_\alpha \bm\Pi_s^{(\alpha)} = \bm F_s^{(\alpha)}.
\end{equation}
Applying mode product with support projector to every species leaves $\mathcal F_s$ unchanged. And these actions commute because their sums concern independent indices.

In the main text, we take the site-$s$ SVD channel $c$ as a local auxiliary degrees of freedom to label the channel state $|s,c\rangle$ as a composite coordinate $(s,c)$. This is \emph{not} an additional occupied-orbital index but spans a local \emph{reconstructed} Hilbert space $\mathcal H_{\text{rec}}^{(\alpha)} = \bigoplus_s \mathbb C^{r_s^{(\alpha)}}$. Then with the $\alpha$-species \emph{occupied orbital amplitudes}
\begin{equation}
	\bm A^{(\alpha)}_{(s,c),\ell_\alpha} := \sqrt{\lambda_{s;c}^{(\alpha)}}\,[\bm U_s^{(\alpha)}]_{\ell_\alpha,c}\in\mathbb C^{M_\alpha\times N}
\end{equation}
where $M_\alpha:=\sum_{s=1}^{N_s}r_s^{(\alpha)}\geq N$, as well as the invariance under the mode-$\alpha$ product with the support projector~\cref{eqA:mode-alpha product invariance under support projector}, we can prove the fusion tensor reconstruction straightforwardly:
\begin{align}
	&\mathcal F_{s;\ell_1,\ldots,\ell_m} \equiv \Big(\mathcal F_s \times_1 \bm\Pi_s^{(1)}\cdots\times_m\bm\Pi_s^{(m)}\Big)_{\ell_1,\ldots,\ell_m} \nonumber\\
	&\equiv \sum_{j_1,\ldots,j_m} \mathcal F_{s; j_1,\ldots,j_m} \prod_{\alpha=1}^m [\bm\Pi_s^{(\alpha)}]_{\ell_\alpha,j_\alpha} \nonumber\\
	&\equiv \sum_{c_1,\ldots,c_m} \bigg[\sum_{j_1,\ldots,j_m}\mathcal F_{s;j_1,\ldots,j_m}\prod_{\alpha=1}^m \dfrac{1}{\sqrt{\lambda_{s;c_\alpha}^{(\alpha)}}}[\bm U_s^{(\alpha)}]_{j_\alpha,c_\alpha}^*\bigg]\nonumber\\
	&\quad\quad\times\prod_{\alpha=1}^m \sqrt{\lambda_{s;c_\alpha}^{(\alpha)}} [\bm U_s^{(\alpha)}]_{\ell_\alpha,c_\alpha}\nonumber\\
	&\equiv \sum_{c_1,\ldots,c_m} \mathcal T_{s;c_1,\ldots,c_m} \prod_{\alpha=1}^m \bm A_{(s,c_\alpha);\ell_\alpha}^{(\alpha)},
\end{align}
where we recognize the \emph{weighted Tucker core}
\begin{equation}\label{eqA:weighted Tucker core}
	\mathcal T_{s; c_1,\ldots, c_m} := \sum_{\ell_1,\ldots,\ell_m} \mathcal F_{s;\ell_1,\ldots,\ell_m} \prod_{\alpha=1}^m \dfrac{[\bm U_s^{(\alpha)}]_{\ell_\alpha,c_\alpha}^*}{\sqrt{\lambda_{s;c_\alpha}^{(\alpha)}}}.
\end{equation}

\subsection{Determinant Expansion}
The Tucker reconstruction~\cref{eqA:weighted Tucker core} can be inserted into the HyperDet wavefunction definition. This gives a re-interpretation of the wavefunction evaluation with channel contractions:
\begin{align}
	&\psi_{\mathcal F}^{\text{phys}}(\bm s) \equiv \mathsf{HyperDet}[\mathcal F_{\bm s}]\nonumber\\
	&\equiv \sum_{P_1,\ldots,P_m} \Big[\prod_{\alpha=1}^m \mathrm{sgn}(P_\alpha)\Big] \prod_{\ell=1}^N \mathcal F_{s_\ell;P_1(\ell),\ldots,P_m(\ell)}\nonumber\\
	&= \sum_{P_1,\ldots,P_m}\Big[\prod_{\alpha=1}^m \mathrm{sgn}(P_\alpha)\Big] \nonumber\\
	&\quad\quad \times\prod_{\ell=1}^N \bigg[\sum_{c_1,\ldots,c_m} \mathcal T_{s_\ell;c_1,\ldots,c_m} \prod_{\alpha=1}^m \bm A_{(s_\ell,c_\alpha);P_\alpha(\ell)}^{(\alpha)}\bigg]\nonumber\\
	&= \sum_{P_1,\ldots,P_m} \Big[\prod_{\alpha=1}^m \mathrm{sgn}(P_\alpha)\Big] \nonumber\\
	&\quad\quad \times\sum_{c_1(\ell),\ldots,c_m(\ell)}\prod_{\ell=1}^N \Big[\mathcal T_{s_\ell;c_1(\ell),\ldots,c_m(\ell)} \prod_{\alpha=1}^m \bm A_{(s_\ell,c_\alpha(\ell));P_\alpha(\ell)}^{(\alpha)}\bigg]\nonumber\\
	&\equiv \sum_{\{c_\alpha(\ell)\}} \Big[\prod_{\ell=1}^N \mathcal T_{s_\ell;c_1(\ell),\ldots,c_m(\ell)}\Big] \prod_{\alpha=1}^m \sum_{P_\alpha} (-1)^{P_\alpha} \prod_{\ell=1}^N \bm A_{(s_\ell,c_\alpha(\ell));P_\alpha(\ell)}^{(\alpha)}\nonumber\\
	& \equiv \sum_{\{c_\alpha(\ell)\}} \Big[\prod_{\ell=1}^N \mathcal T_{s_\ell;c_1(\ell),\ldots,c_m(\ell)}\Big] \prod_{\alpha=1}^m \det\big[\bm A_{(s_\ell,c_\alpha(\ell)); j}^{(\alpha)}\big]_{\ell,j=1}^N,\label{eqA:determinant expansion of HyperDet wavefunction}
\end{align}
where in the third line we exchange the product and sum by choosing, for each physical slot $\ell$ and each species $\alpha$, one channel index $c_\alpha(\ell)\in\{1,\ldots,r^{(\alpha)}_{s_\ell}\}$. The summation over channel configurations and the corresponding $N\times N$ minors are defined around Eq.~\eqref{eq:channel configuration sum}. A non-vanishing physical many-body wavefunction requires that $\mathop{\mathrm{rank}}\bm A^{(\alpha)}=N$ for every species.

\section{Gauge Structure of the HyperDet Ansatz}\label{app:gauge_structure}
The fusion tensor $\mathcal F$ is a variational parametrization of the physical wavefunction, not a unique physical observable by itself; it therefore naturally carries gauge redundancies. In this section, we will explore the gauge structure of our HyperDet wavefunction ansatz. Here, by gauge transformation, we mean a reparametrization $\mathcal F\mapsto\mathcal F'$ that leaves every physical amplitude unchanged up to an unimportant configuration-independent normalization factor:
\begin{equation}
	\Psi_{\mathcal F'}^{\text{phys}}(\bm s)=C\Psi_{\mathcal F}^{\text{phys}}(\bm s),\quad\forall\bm s.
\end{equation}
Such a $\mathbb C$-valued scalar $C$ has no effect on the normalized physical many-body state.

For each physical orbital $s_\ell$, let us define the physical-orbital-restricted order-$m$ fusion tensor $\mathcal F(s_\ell)_{\ell_1,\ldots,\ell_m}:=\mathcal F_{s_\ell;\ell_1,\ldots,\ell_m}$, and consider physical-orbital-independent invertible transformations $\bm G^{(\alpha)}\in\mathrm{GL}(N;\mathbb C)$ acting on the virtual leg of each parton species $\alpha=1,\ldots,m$. In component notation, the transformed fusion tensor reads
\begin{equation*}
	\mathcal F'_{s_\ell;\ell_1,\ldots,\ell_m}=\sum_{\ell_1',\ldots,\ell_m'=1}^{N}G^{(1)}_{\ell_1,\ell_1'}G^{(2)}_{\ell_2,\ell_2'}\cdots G^{(m)}_{\ell_m,\ell_m'}\mathcal F_{s_\ell;\ell_1',\ldots,\ell_m'}.
\end{equation*}
Equivalently, using mode-$\alpha$ tensor-matrix multiplication $\times_\alpha$ to denote the contraction of $\bm G^{(\alpha)}$ with the $\alpha$-th virtual leg of the fusion tensor, this can be written in a neat compact form
\begin{equation}
	\mathcal F'(s_\ell)=\mathcal F(s_\ell)\times_1\bm G^{(1)}\times_2\bm G^{(2)}\cdots\times_m\bm G^{(m)}.
\end{equation}
We now prove that these transformations $\{\bm G^{(\alpha)}\in\mathrm{GL}(N;\mathbb C)\}$ constitute the gauge transformations of the HyperDet wavefunction ansatz.

It is enough to first consider a transformation acting on one parton species, say species $\alpha$, because transformations on different species can then be applied successively. For later comparison with physical-orbital-dependent transformations, let us temporarily allow the transformation matrix to depend on the physical orbital $s_\ell$:
\begin{equation}
	\mathcal F'_{s_\ell;\ell_1,\ldots,\ell_m}=\sum_{\ell_\alpha'=1}^{N}G^{(\alpha)}_{\ell_\alpha,\ell_\alpha'}(s_\ell)\mathcal F_{s_\ell;\ell_1,\ldots,\ell_{\alpha-1},\ell_\alpha',\ell_{\alpha+1},\ldots,\ell_m}.
	\label{eqA:site_dependent_virtual_transformation}
\end{equation}
We will see that \cref{eqA:site_dependent_virtual_transformation} becomes a gauge transformation only when $\bm G^{(\alpha)}(s_\ell)$ is independent of $s_\ell$.

For each physical configuration $\bm s=(s_1,\ldots,s_N)$ and a fixed parton permutation $P_\beta$ for all $\beta\neq\alpha$, let us introduce the $N\times N$ matrix
\begin{equation}
	\mathcal X_{\ell,\ell_\alpha}^{\{P_\beta\}_{\beta\neq\alpha}}:=\mathcal F_{s_\ell;P_1(\ell),\ldots,P_{\alpha-1}(\ell),\ell_\alpha,P_{\alpha+1}(\ell),\ldots,P_m(\ell)}.
	\label{eqA:X_matrix_definition}
\end{equation}
With this notation, the HyperDet can be reorganized as an alternating sum of ordinary determinants:
\begin{align}
	&\mathop{\mathsf{HyperDet}}[\mathcal F_{\bm s}]=\sum_{P_1,\ldots,P_m}\Big[\prod_{\beta=1}^{m}(-1)^{P_\beta}\Big]\prod_{\ell=1}^{N}\mathcal F_{s_\ell;P_1(\ell),\ldots,P_m(\ell)}\nonumber\\
	&\qquad=\sum_{\{P_\beta\}_{\beta\neq\alpha}}\Big[\prod_{\beta\neq\alpha}(-1)^{P_\beta}\Big]\sum_{P_\alpha}(-1)^{P_\alpha}\prod_{\ell=1}^{N}\mathcal X_{\ell,P_\alpha(\ell)}^{\{P_\beta\}_{\beta\neq\alpha}}\nonumber\\
	&\qquad=\sum_{\{P_\beta\}_{\beta\neq\alpha}}\Big[\prod_{\beta\neq\alpha}(-1)^{P_\beta}\Big]\det\bm{\mathcal X}^{\{P_\beta\}_{\beta\neq\alpha}}.
	\label{eqA:hyperdet_as_alternating_determinants}
\end{align}
The last equality follows from the definition of the determinant
\begin{equation*}
	\det\bm{\mathcal A}=\sum_{P\in S_N}(-1)^P\prod_{\ell=1}^{N}\mathcal A_{\ell,P(\ell)}
\end{equation*}
applied to the $\alpha$-th parton leg while keeping the other parton permutations fixed.

Under the general transformation in \cref{eqA:site_dependent_virtual_transformation}, the matrix $\bm{\mathcal X}^{\{P_\beta\}_{\beta\neq\alpha}}$ transforms row by row as
\begin{equation}\label{eqA:X_matrix_site_dependent_transformation}
	\left(\mathcal X_{\ell,\ell_\alpha}^{\{P_\beta\}_{\beta\neq\alpha}}\right)'=\sum_{\ell_\alpha'=1}^{N}G^{(\alpha)}_{\ell_\alpha,\ell_\alpha'}(s_\ell)\mathcal X_{\ell,\ell_\alpha'}^{\{P_\beta\}_{\beta\neq\alpha}}.
\end{equation}
If $\bm G^{(\alpha)}(s_\ell)$ is physical-orbital independent, namely $\bm G^{(\alpha)}(s_\ell)=\bm G^{(\alpha)}$ for all $s_\ell$, then the above transformation reduces to the ordinary matrix multiplication
\begin{equation}
	\left(\bm{\mathcal X}^{\{P_\beta\}_{\beta\neq\alpha}}\right)'=\bm{\mathcal X}^{\{P_\beta\}_{\beta\neq\alpha}}\left[\bm G^{(\alpha)}\right]^T,
\end{equation}
and the transformed determinant satisfies
\begin{equation}
	\det\left[\left(\bm{\mathcal X}^{\{P_\beta\}_{\beta\neq\alpha}}\right)'\right]=\det\bm{\mathcal X}^{\{P_\beta\}_{\beta\neq\alpha}}\det\bm G^{(\alpha)}.
\end{equation}
Substituting this into \cref{eqA:hyperdet_as_alternating_determinants}, we find
\begin{align}
	&\mathop{\mathsf{HyperDet}}[\mathcal F'_{\bm s}]=\sum_{\{P_\beta\}_{\beta\neq\alpha}}\Big[\prod_{\beta\neq\alpha}(-1)^{P_\beta}\Big]\det\left[\left(\bm{\mathcal X}^{\{P_\beta\}_{\beta\neq\alpha}}\right)'\right]\nonumber\\
	&\qquad=\det\bm G^{(\alpha)}\sum_{\{P_\beta\}_{\beta\neq\alpha}}\Big[\prod_{\beta\neq\alpha}(-1)^{P_\beta}\Big]\det\bm{\mathcal X}^{\{P_\beta\}_{\beta\neq\alpha}}\nonumber\\
	&\qquad=\det\bm G^{(\alpha)}\mathop{\mathsf{HyperDet}}[\mathcal F_{\bm s}].
\end{align}
Applying the same argument to all $m$ parton species gives the full transformation rule
\begin{equation}\label{eqA:HyperDet single gauge factor}
	\mathop{\mathsf{HyperDet}}[\mathcal F'_{\bm s}]=\left[\prod_{\alpha=1}^{m}\det\bm G^{(\alpha)}\right]\mathop{\mathsf{HyperDet}}[\mathcal F_{\bm s}].
\end{equation}
The prefactor in \cref{eqA:HyperDet single gauge factor} is independent of the physical configuration $\bm s$. Hence the normalized physical wavefunction is unchanged. The physical-orbital-independent transformations $\bm G^{(\alpha)}\in\mathrm{GL}(N;\mathbb C)$ therefore constitute a gauge redundancy of the HyperDet ansatz:
\begin{equation}\label{eqA:HyperDet gauge structure}
	\mathcal F(s_\ell)\sim\mathcal F(s_\ell)\times_1\bm G^{(1)}\times_2\bm G^{(2)}\cdots\times_m\bm G^{(m)}.
\end{equation}

There is one more gauge redundancy that is \emph{discrete} rather than continuous. Since the $m$ parton species are auxiliary copies with equal occupied dimension $N$, their labels can be permuted. For $\pi\in S_m$, define
\begin{equation}
	[\mathcal P_\pi\mathcal F]_{s_\ell;\ell_1,\ldots,\ell_m}:=\mathcal F_{s_\ell;\ell_{\pi^{-1}(1)},\ldots,\ell_{\pi^{-1}(m)}}.
\end{equation}
Then
\begin{align}
	&\mathop{\mathsf{HyperDet}}[(\mathcal P_\pi\mathcal F)_{\bm s}]\nonumber\\
	&\qquad=\sum_{P_1,\ldots,P_m}\Big[\prod_{\alpha=1}^{m}(-1)^{P_\alpha}\Big]\prod_{\ell=1}^{N}
	\mathcal F_{s_\ell;P_{\pi^{-1}(1)}(\ell),\ldots,P_{\pi^{-1}(m)}(\ell)}\nonumber\\
	&\qquad=\sum_{Q_1,\ldots,Q_m}\Big[\prod_{\alpha=1}^{m}(-1)^{Q_\alpha}\Big]\prod_{\ell=1}^{N}
	\mathcal F_{s_\ell;Q_1(\ell),\ldots,Q_m(\ell)}\nonumber\\
	&\qquad=\mathop{\mathsf{HyperDet}}[\mathcal F_{\bm s}],
\end{align}
where the second line follows by relabeling the dummy summation variables as $Q_\alpha=P_{\pi^{-1}(\alpha)}$. Thus parton-species relabeling leaves the physical wavefunction exactly unchanged. Combining it with the continuous transformations above, the gauge redundancy contains the semidirect product
\begin{equation}
	\mathrm{GL}(N;\mathbb C)^m\rtimes S_m,
\end{equation}
where $S_m$ permutes the $m$ continuous gauge factors. This $S_m$ is a discrete gauge redundancy of the parametrization, not a physical global symmetry. Consequently it cannot be interpreted as spontaneously broken or restored by the optimized wavefunction. It only means that labeled parton species are coordinate choices on the same physical state.

By contrast, if $\bm G^{(\alpha)}(s_\ell)$ depends on the physical orbital $s_\ell$, then different rows $\ell$ of the matrix $\bm{\mathcal X}^{\{P_\beta\}_{\beta\neq\alpha}}$ are acted on by different matrices $\bm G^{(\alpha)}(s_\ell)$. In that case \cref{eqA:X_matrix_site_dependent_transformation} cannot be recast as simple matrix multiplication $\bm{\mathcal X}\mapsto\bm{\mathcal X}\bm{\mathcal M}$. Explicitly, the determinant becomes
\begin{align}
	&\det\left[\left(\bm{\mathcal X}^{\{P_\beta\}_{\beta\neq\alpha}}\right)'\right] \nonumber\\
	&\quad=\sum_{P_\alpha}(-1)^{P_\alpha}\prod_{\ell=1}^{N}\sum_{\ell_\alpha'=1}^{N}G^{(\alpha)}_{P_\alpha(\ell),\ell_\alpha'}(s_\ell)\mathcal X_{\ell,\ell_\alpha'}^{\{P_\beta\}_{\beta\neq\alpha}}.
\end{align}
There is generally no configuration-independent scalar that can be factored out of this expression. Consequently, a physical-orbital-dependent virtual-leg transformation changes the relative amplitudes between different physical configurations and is not, in general, a gauge transformation of the HyperDet ansatz. It should instead be regarded as a genuine variational deformation of the fusion tensor.

This distinction is important for interpreting the SVD diagnostics of the bipartite fusion matrix. Tensor ranks and matrix ranks are invariant under the physical-orbital-independent invertible transformations above, but singular values are invariant only under unitary transformations, not under the full non-unitary $\mathrm{GL}(N;\mathbb C)$ gauge redundancy. Moreover, a labeled SVD spectrum of the bipartite fusion matrix for the bipartition $\alpha|\bar\alpha$ is relabeled under the discrete $S_m$ gauge action. Therefore SVD spectra of physical-orbital-restricted or bipartite fusion tensors should be reported only after
\begin{itemize}
	\item a specified continuous $\mathrm{GL}(N;\mathbb C)^m$ gauge fixing;
	\item a reduction to $S_m$ permutation-invariant form such as their unordered collection or species average.
\end{itemize}

\section{Stochastic Reconfiguration with Symmetry Pinning and Symmetry Projection}\label{app:sr_symmetry}
In this appendix we derive the stochastic-reconfiguration (SR) update~\cite{sorella1998green,sorella2001generalized} used to optimize the fusion tensor.

We first introduce a unified notation. Let $|\Psi_{\mathcal F}\rangle$ denote the state actually sampled in VMC, and let $K$ denote the operator used for imaginary-time evolution. The three cases used in this work are
\begin{equation}
	\begin{array}{c|c|c}
		\text{Case} & |\Psi_{\mathcal F}\rangle & K\\
		\hline
		\text{ordinary SR} & |\Psi_{\mathcal F}^{\text{phys}}\rangle & H\\
		\text{irrep-projected SR} & |\Psi_{\mathcal F;\chi}^{\text{phys}}\rangle=P_\chi|\Psi_{\mathcal F}^{\text{phys}}\rangle & H\\
		\text{symmetry-pinned SR} & |\Psi_{\mathcal F}^{\text{phys}}\rangle & H_\kappa=H+\kappa Q_\chi
	\end{array}
\end{equation}
For the symmetry-pinned case, let $G$ be an Abelian symmetry group and $\chi$ a one-dimensional irrep. A state in the $\chi$ sector satisfies
\begin{equation}
	U_g|\Psi\rangle=\chi(g)|\Psi\rangle.
\end{equation}
For a set of generators $g_a$ of $G$, define the positive-semidefinite penalty operator
\begin{equation}
	Q_\chi=\sum_a Q_{\chi,a},\quad Q_{\chi,a}:=(U_{g_a}-\chi(g_a)I)^\dagger(U_{g_a}-\chi(g_a)I).
\end{equation}
Then $\langle\Psi|Q_\chi|\Psi\rangle=0$ if and only if $|\Psi\rangle$ lies in the desired symmetry sector. For translation symmetry $G=\mathbb Z_{L_x}\times\mathbb Z_{L_y}$ with generators $T_\mu$ and characters $\lambda_\mu=\chi(T_\mu)$, this becomes
\begin{equation}
	Q_\chi=\sum_{\mu=x,y}Q_{\chi,\mu},\quad Q_{\chi,\mu}=(U_{T_\mu}-\lambda_\mu I)^\dagger(U_{T_\mu}-\lambda_\mu I).
\end{equation}
Using $|\lambda_\mu|=1$, one may also write
\begin{equation}
	Q_{\chi,\mu}=2I-\lambda_\mu U_{T_\mu}^\dagger-\lambda_\mu^*U_{T_\mu}.
\end{equation}
If $[H,U_g]=0$ for all $g\in G$, then $[H,Q_\chi]=0$. Hence $H_\kappa=H+\kappa Q_\chi$ does not deform the eigenvectors of $H$; it only lifts states outside the chosen sector. The parameter $\kappa>0$ controls the strength of this soft symmetry pinning.

We now derive the SR update for the given unified pair $(|\Psi_{\mathcal F}\rangle,K)$. Imaginary-time evolution gives
\begin{equation}
	e^{-\delta\tau K}|\Psi_{\mathcal F}\rangle=|\Psi_{\mathcal F}\rangle-\delta\tau K|\Psi_{\mathcal F}\rangle+\mathcal O(\delta\tau^2).
\end{equation}
\emph{TDVP searches for the infinitesimal parameter update $\delta\mathcal F$ whose tangent-space variation best approximates this imaginary-time-evolved vector}. Because the HyperDet wavefunction ansatz is holomorphic with respect to the fusion tensor:
\begin{equation}
	\partial_{\mathcal F_\mu^*}\Psi_{\mathcal F}^{\text{phys}}(\bm s)=0,
\end{equation}
the tangent space is spanned only by $\partial_{\mathcal F_\mu}|\Psi_{\mathcal F}^{\text{phys}}\rangle$ and not by independent $\partial_{\mathcal F_\mu^*}|\Psi_{\mathcal F}^{\text{phys}}\rangle$ directions. Here, holomorphicity refers to the variational parameters $\mathcal F$, \emph{not} to the physical coordinates of the many-body system. Therefore we can write the first-order expansion as
\begin{equation}
	|\Psi_{\mathcal F+\delta\mathcal F}\rangle=a_0|\Psi_{\mathcal F}\rangle+\sum_\nu\delta\mathcal F_\nu\cdot\partial_{\mathcal F_\nu}|\Psi_{\mathcal F}\rangle+|\Psi^\perp\rangle,
\end{equation}
where $|\Psi^\perp\rangle$ is orthogonal to $|\Psi_{\mathcal F}\rangle$ and to all tangent vectors $\partial_{\mathcal F_\nu}|\Psi_{\mathcal F}\rangle$. The coefficient $a_0$ dictates the normalization change.

Define the diagonal log-derivative operator
\begin{equation}
	\hat{\mathcal O}_\mu=\sum_{\bm s}|\bm s\rangle\mathcal O_\mu(\bm s)\langle\bm s|
\end{equation}
with $\mathcal O_\mu(\bm s)=\partial_{\mathcal F_\mu}\ln\Psi_{\mathcal F}(\bm s)=\partial_{\mathcal F_\mu}\Psi_{\mathcal F}(\bm s)/\Psi_{\mathcal F}(\bm s)$, then by construction,
\begin{equation}
	\hat{\mathcal O}_\mu|\Psi_{\mathcal F}\rangle=\partial_{\mathcal F_\mu}|\Psi_{\mathcal F}\rangle.
\end{equation}
Projecting the TDVP equation onto $\langle\Psi_{\mathcal F}|$ and defining the normalized expectation value $\langle\!\langle\hat A\rangle\!\rangle=\frac{\langle\Psi_{\mathcal F}|\hat A|\Psi_{\mathcal F}\rangle}{\langle\Psi_{\mathcal F}|\Psi_{\mathcal F}\rangle}$, we get,
\begin{equation}
	a_0+\sum_\nu\delta\mathcal F_\nu\langle\!\langle\hat{\mathcal O}_\nu\rangle\!\rangle=1-\delta\tau\langle\!\langle K\rangle\!\rangle.
\end{equation}
Alternatively, projecting onto $\langle\partial_{\mathcal F_\mu}\Psi_{\mathcal F}|=\langle\Psi_{\mathcal F}|\hat{\mathcal O}_\mu^\dagger$ gives
\begin{equation}
	a_0\langle\!\langle\hat{\mathcal O}_\mu^\dagger\rangle\!\rangle+\sum_\nu\delta\mathcal F_\nu\langle\!\langle\hat{\mathcal O}_\mu^\dagger\hat{\mathcal O}_\nu\rangle\!\rangle=\langle\!\langle\hat{\mathcal O}_\mu^\dagger\rangle\!\rangle-\delta\tau\langle\!\langle\hat{\mathcal O}_\mu^\dagger K\rangle\!\rangle.
\end{equation}
Eliminating $a_0$ yields the standard \emph{SR update rule}:
\begin{equation}
	\bm{\mathcal S}\cdot\delta\mathcal F=-\delta\tau \bm f^K,
	\label{eqA:sr_normal_equation}
\end{equation}
where
\begin{equation}
	\mathcal S_{\mu\nu}=\langle\!\langle\hat{\mathcal O}_\mu^\dagger\hat{\mathcal O}_\nu\rangle\!\rangle-\langle\!\langle\hat{\mathcal O}_\mu^\dagger\rangle\!\rangle\langle\!\langle\hat{\mathcal O}_\nu\rangle\!\rangle
\end{equation}
is the quantum-geometric covariance matrix, and
\begin{equation}
	f_\mu^{K}=\langle\!\langle\hat{\mathcal O}_\mu^\dagger K\rangle\!\rangle-\langle\!\langle\hat{\mathcal O}_\mu^\dagger\rangle\!\rangle\langle\!\langle K\rangle\!\rangle
\end{equation}
is the force vector associated with the operator $K$. In finite Monte Carlo sampling one solves the Tikhonov-regularized system
\begin{equation}
	\delta\mathcal F=-\delta\tau(\bm{\mathcal S}+\lambda I)^{-1}\bm f^{K},
\end{equation}
or applies an equivalent iterative solver such as MinSR~\cite{chen2024empowering}, SPRING~\cite{goldshlager2024kaczmarz}, or related stabilized SR solvers.

The modification to Monte Carlo estimators is minimal. Indeed, the sampled probability distribution is still
\begin{equation}
	p_{\mathcal F}(\bm s)=\dfrac{|\Psi_{\mathcal F}(\bm s)|^2}{\sum_{\bm s'}|\Psi_{\mathcal F}(\bm s')|^2}.
\end{equation}
The local energy associated with $K$ is
\begin{equation}
	E_{\text{loc}}^{K}(\bm s)=\dfrac{\langle\bm s|K|\Psi_{\mathcal F}\rangle}{\langle\bm s|\Psi_{\mathcal F}\rangle}=\sum_{\bm s'}K_{\bm s,\bm s'}\dfrac{\Psi_{\mathcal F}(\bm s')}{\Psi_{\mathcal F}(\bm s)}.
\end{equation}
So
\begin{align}
	\mathcal S_{\mu\nu}&=\mathbb E_{\bm s\sim p_{\mathcal F}}\left[\left(\mathcal O_\mu(\bm s)^*-\overline{\mathcal O_\mu^*}\right)\left(\mathcal O_\nu(\bm s)-\overline{\mathcal O_\nu}\right)\right],\\
	f_\mu^{K}&=\mathbb E_{\bm s\sim p_{\mathcal F}}\left[\left(\mathcal O_\mu(\bm s)^*-\overline{\mathcal O_\mu^*}\right)E_{\text{loc}}^{K}(\bm s)\right].
\end{align}
Equivalently, one may replace $E_{\text{loc}}^{K}(\bm s)$ by $E_{\text{loc}}^{K}(\bm s)-\overline{E_{\text{loc}}^{K}}$ in the last equation.

For symmetry-pinned SR, the only modification is the local energy:
\begin{equation}
	E_{\text{loc}}^{K}(\bm s)=E_{\text{loc}}^{H}(\bm s)+\kappa E_{\text{loc}}^{Q_\chi}(\bm s),
\end{equation}
where for translation generators,
\begin{equation}
	E_{\text{loc}}^{Q_{\chi,\mu}}(\bm s)=2-\lambda_\mu\dfrac{\langle\bm s|U_{T_\mu}^\dagger|\Psi_{\mathcal F}^{\text{phys}}\rangle}{\langle\bm s|\Psi_{\mathcal F}^{\text{phys}}\rangle}-\lambda_\mu^*\dfrac{\langle\bm s|U_{T_\mu}|\Psi_{\mathcal F}^{\text{phys}}\rangle}{\langle\bm s|\Psi_{\mathcal F}^{\text{phys}}\rangle}.
\end{equation}
Thus the penalty contribution can be evaluated from wavefunction ratios between translated configurations. This soft-pinning scheme keeps sampling the unprojected wavefunction $|\Psi_{\mathcal F}^{\text{phys}}\rangle$, but energetically guides it toward the desired symmetry sector.

For irrep-projected SR, the sampled state becomes
\begin{equation}
	|\Psi_{\mathcal F;\chi}^{\text{phys}}\rangle=P_\chi|\Psi_{\mathcal F}^{\text{phys}}\rangle,\quad P_\chi=\dfrac{1}{|G|}\sum_{g\in G}\chi(g)^*U_g.
\end{equation}
Let the symmetry act on the configuration basis as
\begin{equation}
	U_g|\bm s\rangle=\alpha_g(\bm s)|g\cdot\bm s\rangle.
\end{equation}
Then the projected amplitude is
\begin{equation}
	\Psi_{\mathcal F;\chi}^{\text{phys}}(\bm s)=\langle\bm s|P_\chi|\Psi_{\mathcal F}^{\text{phys}}\rangle=\dfrac{1}{|G|}\sum_{g\in G}\chi(g)\alpha_g^*(\bm s)\Psi_{\mathcal F}^{\text{phys}}(g\cdot\bm s).
\end{equation}
The overall factor $1/|G|$ cancels in Metropolis ratios, local energies, and log derivatives. Since $P_\chi$ is linear and independent of $\mathcal F$, the projected ansatz remains holomorphic $\partial_{\mathcal F_\mu^*}\Psi_{\mathcal F;\chi}^{\text{phys}}(\bm s)=0$, and the projected log derivative entering SR becomes
\begin{align}
	\mathcal O_{\mu;\chi}(\bm s) &\equiv \partial_{\mathcal F_\mu} \ln\Big(\Psi_{\mathcal F;\chi}(\bm s)\Big) \equiv \dfrac{1}{\Psi_{\mathcal F;\chi}(\bm s)} \partial_{\mathcal F_\mu}\Big(P_\chi\Psi_{\mathcal F}(\bm s)\Big)\nonumber\\
	&=\dfrac{1}{\Psi_{\mathcal F;\chi}(\bm s)} \langle P_\chi^\dagger\bm s|\partial_{\mathcal F_\mu}\Psi_{\mathcal F}\rangle\nonumber\\
	&=\dfrac{1}{\Psi_{\mathcal F;\chi}(\bm s)} \sum_{g\in G} \chi(g)\alpha_g^*(\bm s) \partial_{\mathcal F_\mu}\Psi_{\mathcal F}(g\cdot\bm s).
\end{align}
The projected local energy is
\begin{equation}
	E_{\text{loc};\chi}^{H}(\bm s)=\dfrac{\langle\bm s|H|\Psi_{\mathcal F;\chi}^{\text{phys}}\rangle}{\langle\bm s|\Psi_{\mathcal F;\chi}^{\text{phys}}\rangle}.
\end{equation}
The SR matrix and force are then computed from the same covariance formulas above, with $\Psi_{\mathcal F}(\bm s)$ replaced by $\Psi_{\mathcal F;\chi}^{\text{phys}}(\bm s)$, $\mathcal O_\mu(\bm s)$ replaced by $\mathcal O_{\mu;\chi}(\bm s)$, and $E_{\text{loc}}^K(\bm s)$ replaced by $E_{\text{loc};\chi}^{H}(\bm s)$.

The difference between the two symmetry strategies is therefore conceptual rather than algebraic. Symmetry projection exactly enforces the target irrep at every optimization step and samples from probability distribution proportional to $|\Psi_{\mathcal F;\chi}^{\text{phys}}|^2$. Symmetry pinning samples the unprojected state and uses $H_\kappa=H+\kappa Q_\chi$ to energetically suppress components outside the target sector. Projection is usually best for obtaining the most accurate sector-resolved energy and wavefunction overlap, while pinning can be useful when one wants the full unprojected fusion tensor to retain information about how physical symmetries act before projection.

\section{Gauge Fixing of SR Updates and Bipartite Fusion Matrix Diagnostics}\label{app:sr_gauge_fixing}
The HyperDet parametrization contains both continuous and discrete gauge redundancies. As shown in \cref{app:gauge_structure}, a physical-orbital-independent transformation $\bm G^{(\alpha)}\in\mathrm{GL}(N;\mathbb C)$ on the $\alpha$-th parton virtual legs sends
\begin{equation}
	\mathcal F(s_\ell)\mapsto\mathcal F'(s_\ell)=\mathcal F(s_\ell)\times_1\bm G^{(1)}\times_2\bm G^{(2)}\cdots\times_m\bm G^{(m)}
\end{equation}
for every physical orbital $s_\ell$, and changes every physical amplitude only by the configuration-independent normalization factor $\prod_{\alpha=1}^m\det\bm G^{(\alpha)}$. Therefore the normalized physical state depends not on the fusion tensor $\mathcal F$ itself, but on its equivalence class under the global gauge group $\mathrm{GL}(N;\mathbb C)^m$.

This redundancy affects VMC optimization in two related but distinct ways. First, the SR tangent space contains directions that only move $\mathcal F$ along its \emph{gauge orbits} and hence do not change the normalized physical wavefunction. These directions should be projected out to prevent numerical gauge drift. Second, quantities such as singular values of physical-orbital-restricted or bipartite fusion tensors are not invariant under the full non-unitary gauge group. Therefore the SVD spectra of the bipartite fusion matrix must be computed \emph{after} choosing a definite continuous gauge representative and then interpreted modulo the discrete species labels. In this section we separate these operations. We use a \emph{horizontal projection} to remove continuous gauge components from SR updates, a \emph{balanced canonical gauge} to define the SVD spectra of the bipartite fusion matrix, and permutation-invariant summaries to remove the discrete species-label ambiguity.

\subsection{Infinitesimal gauge directions}
Let $\bm K^{(\alpha)}\in\mathfrak{gl}(N;\mathbb C)$ denote the infinitesimal generators of the global gauge group. Setting $\bm G^{(\alpha)}=I+\epsilon\bm K^{(\alpha)}+\mathcal O(\epsilon^2)$ gives the infinitesimal gauge variation
\begin{equation}
	\delta_{\bm K}\mathcal F_{s_\ell;\ell_1,\ldots,\ell_m}=\sum_{\alpha=1}^{m}\sum_{a=1}^{N}K^{(\alpha)}_{\ell_\alpha a}\mathcal F_{s_\ell;\ell_1,\ldots,\ell_{\alpha-1},a,\ell_{\alpha+1},\ldots,\ell_m}.
\end{equation}
Equivalently, in compact form we have
\begin{equation}
	\delta_{\bm K}\mathcal F(s_\ell)=\sum_{\alpha=1}^{m}\mathcal F(s_\ell)\times_\alpha\bm K^{(\alpha)}.
\end{equation}
The finite gauge transformation multiplies the HyperDet amplitude by $\prod_\alpha\det\bm G^{(\alpha)}$. Expanding this determinant factor to first order in $\epsilon$, we obtain
\begin{equation}
	\delta_{\bm K}\Psi_{\mathcal F}^{\text{phys}}(\bm s)=\left[\sum_{\alpha=1}^{m}\tr\bm K^{(\alpha)}\right]\Psi_{\mathcal F}^{\text{phys}}(\bm s).
\end{equation}
Thus the logarithmic derivative associated with a pure gauge direction is
\begin{equation}
	\mathcal O_{\bm K}(\bm s)=\dfrac{\delta_{\bm K}\Psi_{\mathcal F}^{\text{phys}}(\bm s)}{\Psi_{\mathcal F}^{\text{phys}}(\bm s)}=\sum_{\alpha=1}^{m}\tr\bm K^{(\alpha)}.
\end{equation}
This is independent of the configuration $\bm s$. Consequently its centered value vanishes:
\begin{equation}
	\overline{\mathcal O}_{\bm K}(\bm s)=\mathcal O_{\bm K}(\bm s)-\langle\mathcal O_{\bm K}\rangle=0.
\end{equation}
Therefore exact SR has zero quantum-geometric norm along pure gauge directions.

In an exact calculation, the covariance matrix $\bm{\mathcal S}$ is singular on these directions and the force vector has no component along them. In finite Monte Carlo sampling, however, gauge components can still appear because of sampling noise, Tikhonov regularization, finite precision, and momentum-based optimizers. Removing these components stabilizes the tensor parametrization without changing the physical state.

\subsection{Frobenius-horizontal projection of SR updates}
At a fixed fusion tensor $\mathcal F$, the gauge variation defines the vertical gauge tangent space
\begin{equation}
	\mathcal V_{\mathcal F}:=\left\{\delta_{\bm K}\mathcal F:\bm K^{(\alpha)}\in\mathfrak{gl}(N;\mathbb C),\alpha=1,\ldots,m\right\}.
\end{equation}
We choose the horizontal subspace to be the Frobenius-orthogonal complement of $\mathcal V_{\mathcal F}$:
\begin{equation}
	\mathcal H_{\mathcal F}:=\left\{\delta\mathcal F:\Re\sum_{s_\ell=1}^{N_s}\langle\delta\mathcal F(s_\ell),\delta_{\bm K}\mathcal F(s_\ell)\rangle_F=0,\quad\forall\bm K^{(\alpha)}\right\}.
\end{equation}
Here $\langle A,B\rangle_F=\sum_{\bm\ell}A_{\bm\ell}^*B_{\bm\ell}$ is the Frobenius inner product on each physical-orbital-restricted tensor. This choice is not a new physical assumption. It is a numerical gauge choice specifying how to represent a physical TDVP update inside the redundant fusion-tensor coordinate system.

The projection can be implemented as a real least-squares problem. Choose a real basis $\{\bm B_p\}_{p=1}^{2N^2}$ of $\mathfrak{gl}(N;\mathbb C)$, for example
\begin{equation}
	\{\bm B_p\}=\{\bm E_{ab},i\bm E_{ab}\}_{a,b=1}^{N}.
\end{equation}
For each parton species $\alpha$ and basis generator $\bm B_p$, define the corresponding vertical vector $\bm v_{\alpha p}:=\mathrm{vec}_{\mathbb R}(\delta_{\bm B_p}^{(\alpha)}\mathcal F)$, where $\mathrm{vec}_{\mathbb R}$ means that the real and imaginary parts of the complex tensor are stacked into a real vector. Let $\bm A$ be the real design matrix whose columns are $\bm v_{\alpha p}$. For a raw SR update $\delta\mathcal F$, define $\bm b:=\mathrm{vec}_{\mathbb R}(\delta\mathcal F)$. The vertical component is the best Frobenius fit of $\bm b$ inside the gauge tangent space:
\begin{equation}
	\bm x_\star=\arg\min_{\bm x}\|\bm A\bm x-\bm b\|_2^2.
\end{equation}
Using a truncated singular-value pseudoinverse, the solution is
\begin{equation}
	\bm x_\star=\bm A^+\bm b.
\end{equation}
The horizontal update is then
\begin{equation}
	\mathrm{vec}_{\mathbb R}(\delta\mathcal F_{\mathrm H})=\bm b-\bm A\bm x_\star.
\end{equation}
Equivalently, $\delta\mathcal F_{\mathrm H}$ is the orthogonal projection of the raw update onto $\mathcal H_{\mathcal F}$.

The design matrix $\bm A$ is generally rank deficient. This is expected and harmless. For example, scalar generators $\bm K^{(\alpha)}=c_\alpha I$ with $\sum_\alpha c_\alpha=0$ can act trivially on the fusion tensor itself. The pseudoinverse should therefore be understood as a projection onto the column space of $\bm A$, not as a unique reconstruction of the gauge generators.

Now let us explain how to incorporate this Frobenius-horizontal projection into SR, MinSR~\cite{chen2024empowering}, and SPRING~\cite{goldshlager2024kaczmarz} optimizers. The notation below follows Ref.~\cite{goldshlager2024kaczmarz}. Let $\bar{\bm O}$ denote the centered logarithmic-derivative matrix sampled in VMC, with rows labeled by Monte Carlo samples and columns labeled by flattened fusion-tensor parameters. Let $\bar{\bm E}$ denote the centered local-energy vector, including the symmetry-pinning contribution if used. The ordinary SR equation can be written as
\begin{equation}
	\delta\mathcal F=-\delta\tau(\bm{\mathcal S}+\lambda I)^{-1}\bm f.
\end{equation}
After solving this equation, we replace the raw update by its \emph{horizontal projection} explained above:
\begin{equation}
	\delta\mathcal F\mapsto\delta\mathcal F_{\mathrm H}.
\end{equation}
Since the subtracted component is a pure gauge tangent, this operation changes only the tensor representative and not the normalized physical wavefunction to first order.

For MinSR-type dual solvers, the same projection is applied to the final primal update. For momentum-based methods, the momentum vector should also be kept in the current horizontal space. Operationally, before using a previous-step momentum vector $\bm\phi_{n-1}$ at the current tensor $\mathcal F_n$, we project it as
\begin{equation}
	\bm\phi_{n-1}\mapsto P_{\mathrm H}(\mathcal F_n)\bm\phi_{n-1}.
\end{equation}
For the SPRING-style update used in this work, define
\begin{equation}
	\bm T=\dfrac{1}{N_{\mathrm{MC}}}\bar{\bm O}\bar{\bm O}^\dagger+\lambda I.
\end{equation}
After transporting the previous momentum to the current horizontal space, the intermediate residual is
\begin{equation}
	\bm\zeta_n=\bar{\bm E}-\mu\dfrac{1}{N_{\mathrm{MC}}}\bar{\bm O}\bm\phi_{n-1}.
\end{equation}
The new unprojected momentum is
\begin{equation}
	\widetilde{\bm\phi}_n=\dfrac{1}{N_{\mathrm{MC}}}\bar{\bm O}^\dagger\bm T^{-1}\bm\zeta_n+\mu\bm\phi_{n-1}.
\end{equation}
We then project
\begin{equation}
	\bm\phi_n=P_{\mathrm H}(\mathcal F_n)\widetilde{\bm\phi}_n.
\end{equation}
The parameter update is
\begin{equation}
	\delta\mathcal F_n=-\eta_n\bm\phi_n,
\end{equation}
possibly with the norm cap used in the optimizer. Since multiplication by a scalar preserves horizontality, an additional projection of $\delta\mathcal F_n$ is optional once $\bm\phi_n$ has already been projected. In practice, keeping the final projection is a useful safety check, while removing it reduces the computational cost.

\subsection{Discrete species-permutation gauge in SR}

The horizontal projection above removes only the connected $\mathrm{GL}(N;\mathbb C)^m$ gauge tangent. The species-permutation factor $S_m$ in \cref{app:gauge_structure} is discrete and has no infinitesimal generator, so there is no vertical tangent direction to subtract. This does not mean that the discrete gauge redundancy is broken by SR. In exact arithmetic, the TDVP/SR equations are equivariant under relabeling parton species: if $\delta\mathcal F$ is the update at $\mathcal F$, then $\mathcal P_\pi\delta\mathcal F$ is the corresponding update at $\mathcal P_\pi\mathcal F$, provided the optimizer treats all species by the same numerical rule.

In a finite Monte Carlo run, random initialization, sampling noise, and linear-solver roundoff can select one labeled representative of the $S_m$ gauge orbit. Such a labeled representative need not be invariant under species exchange, but reflects only gauge fixing induced by the numerical trajectory, not physical symmetry breaking. For this reason we should not enforce species permutation by replacing $\mathcal F$ with an average over its permutations; that average is generally not gauge-equivalent to the original tensor and can change the physical wavefunction.

For diagnostics, the correct treatment is instead to quotient the labels. After balancing the continuous $\mathrm{GL}$ gauge, the SVD spectrum of the bipartite fusion matrix for a labeled bipartition $\alpha|\bar\alpha$ transforms as
\begin{equation}
	\lambda^{(\alpha)}(\mathcal P_\pi\mathcal F)=\lambda^{(\pi^{-1}(\alpha))}(\mathcal F).
\end{equation}
Thus an individual labeled panel is a useful consistency check, but not a physical observable. The species-averaged spectrum or another symmetric function of $\{\lambda^{(\alpha)}\}_{\alpha=1}^m$ is invariant under this discrete gauge redundancy. Therefore, similarity among the species-resolved spectra is more than a consistency check: for $m\ge3$, their near coincidence constitutes a nontrivial permutation \emph{isotropy} of the balanced one-leg spectral diagnostics. This spectral isotropy should not, however, be identified with a physical parton-permutation symmetry or with a permutation-valued stabilizer of the full fusion tensor, which would require a direct tensor-level stabilizer analysis.

\subsection{Balanced gauge for bipartite-fusion-matrix SVD spectra}
Horizontal projection removes gauge drift from the SR updates but does not select a definite representative of the $\mathrm{GL}(N;\mathbb C)^m$ gauge orbit. Since the singular values of the bipartite fusion matrix are invariant only under unitary, rather than general non-unitary gauge transformations, the SVD spectra reported in the main text are meaningful only after the continuous gauge freedom has been fixed in a reproducible way. We do this with a post-processing step that we call the \emph{balanced gauge}, which selects, for every parton leg, a representative on which that leg's marginal Gram matrix is proportional to the identity.

For each parton leg $\alpha$, define the $N\times N$ one-leg Gram matrix by contracting all legs except the parton virtual leg-$\alpha$ for the fusion tensor with its complex conjugate:
\begin{align}\label{eqA:one_leg_gram}
	[\bm R^{(\alpha)}]_{ab}&=\sum_{s_\ell=1}^{N_s}\sum_{\{\ell_\beta:\beta\neq\alpha\}}\mathcal F_{s_\ell;\ell_1,\ldots,\ell_{\alpha-1},a,\ell_{\alpha+1},\ldots,\ell_m}\nonumber\\
	&\qquad\qquad\times\mathcal F^*_{s_\ell;\ell_1,\ldots,\ell_{\alpha-1},b,\ell_{\alpha+1},\ldots,\ell_m}.
\end{align}
It is Hermitian and positive semidefinite, and it plays the role of the reduced density matrix of virtual leg $\alpha$: the spread of its eigenvalues measures how anisotropic the gauge of that leg is. Under a gauge transformation $\bm G^{(\alpha)}\in\mathrm{GL}(N;\mathbb C)$ acting on leg $\alpha$, the Gram matrix transforms covariantly,
\begin{equation}
	\bm R^{(\alpha)}\mapsto \bm G^{(\alpha)}\bm R^{(\alpha)}\bm G^{(\alpha)\dagger}.
\end{equation}

The balancing step diagonalizes the Gram matrix, $\bm R^{(\alpha)}=\bm V^{(\alpha)}\bm\Lambda^{(\alpha)}\bm V^{(\alpha)\dagger}$ with $\bm V^{(\alpha)}$ unitary, and applies to leg $\alpha$ the \emph{whitening transformation}
\begin{equation}
	\bm W^{(\alpha)}=\bm V^{(\alpha)}\mathrm{diag}\Big[\Big(\frac{\gamma^{(\alpha)}}{\max(\lambda^{(\alpha)}_i,\epsilon\lambda^{(\alpha)}_{\max})}\Big)^{1/2}\Big]\bm V^{(\alpha)\dagger},
	\label{eqA:whitener}
\end{equation}
where $\lambda^{(\alpha)}_i$ are the eigenvalues of $\bm R^{(\alpha)}$, $\lambda^{(\alpha)}_{\max}$ is the largest one, $\epsilon>0$ is a small relative floor introduced so that (nearly) vanishing eigenvalues do not cause the inverse square root to diverge --- in practice it is only activated for eigenvalues far below the dominant scale, and the results are insensitive to its precise value (we use $\epsilon\sim10^{-12}$) --- and
\begin{equation}
	\gamma^{(\alpha)}=\exp\Big\langle \ln\max\big(\lambda^{(\alpha)}_i,\epsilon\lambda^{(\alpha)}_{\max}\big)\Big\rangle_i
\end{equation}
is the geometric mean of the floored eigenvalues.

Since $\bm W^{(\alpha)}\bm R^{(\alpha)}\bm W^{(\alpha)\dagger}=\gamma^{(\alpha)}\bm I$, this whitening transformation makes the Gram matrix of leg-$\alpha$ proportional to the identity: after the transformation, every direction in the virtual-leg space carries, on average, the same weight, in direct analogy with the standard whitening of a covariance matrix in multivariate statistics. Crucially, this is a pure gauge choice: the transformed tensor lies on the same $\mathrm{GL}(N;\mathbb C)^m$ orbit and represents the same physical state. The overall scale $\gamma^{(\alpha)}$ is itself irrelevant, because multiplying a virtual leg by a scalar only rescales every physical amplitude by a configuration-independent factor. Here fixing it to the geometric mean simply keeps the procedure numerically stable.

Whitening leg-$\alpha$ generically changes the Gram matrices of the other legs, so the procedure must be repeated. Starting from the raw optimized tensor, we cycle over $\alpha=1,\ldots,m$, applying the whitener~\cref{eqA:whitener} to each leg in turn, until
\begin{equation}\label{eqA:balanced_condition}
	\bm R^{(\alpha)}\propto \bm I,\qquad \alpha=1,\ldots,m,
\end{equation}
holds within numerical tolerance. A convenient convergence measure is the
\emph{anisotropy}
\begin{equation}
	\eta^{(\alpha)}=\frac{\big\|\bm R^{(\alpha)}/\tr \bm R^{(\alpha)}-\bm I/N\big\|_F}{\|\bm I/N\|_F},
\end{equation}
which vanishes precisely when $\bm R^{(\alpha)}\propto \bm I$; we iterate until $\max_\alpha\eta^{(\alpha)}$ falls below a tolerance (typically $10^{-10}$). Between sweeps we also rescale the entire fusion tensor so that its average slice norm stays $\mathcal O(1)$; this common scalar is again gauge-inert.

\subsection{Completeness of the Balanced Gauge}
The purpose of balanced gauge is to make the diagnostics extracted from the optimized fusion tensor \emph{well-defined}. The two diagnostic objects used in the main text are invariant only under the unitary part of the gauge group $\mathrm{GL}(N;\mathbb C)^m$: the normalized singular-value spectrum of the bipartite fusion matrix, and the occupied-orbital projector $\bm P^{(\alpha)}$ whose Bott index yields the parton tuple. The lemma below shows that the postprocessing balanced gauge completely removes the non-unitary part of the gauge freedom. The two corollaries then guarantee that both diagnostics depend only on the gauge orbit of the optimized tensor.
\begin{lemma}[Completeness of the Balanced Gauge]
	Exact balanced representatives on the same $\mathrm{GL}(N;\mathbb C)^m$ gauge orbit differ only by virtual-leg unitary transformations and an overall scalar.
\end{lemma}
\begin{proof}
	Let $\mathcal F$ and $\mathcal F'$ be two finite, balanced representatives on the same gauge orbit. Write each invertible virtual-leg transformation in polar form and separate out its scalar part, so that
	\begin{equation*}
		\mathcal F'=z\left(\bigotimes_{\alpha=1}^{m}\bm L^{(\alpha)}\right)e^{\widehat S}\mathcal F,
	\end{equation*}
	where $z=\prod_{\alpha=1}^m z_\alpha\in\mathbb C^\times$ collects the scalar (determinant) parts of the $m$ leg transformations, $\bm L^{(\alpha)}\in U(N)$ is the unitary part, and $\widehat S = \sum_\alpha\widehat S^{(\alpha)}$ is the non-unitary part with traceless Hermitian generators $\widehat S^{(\alpha)}$ that balancing is designed to remove. $\widehat S^{(\alpha)}$ acts on the virtual space as a Hermitian matrix $\bm S^{(\alpha)}$. Since scalars and virtual-leg unitaries preserve the balancing condition~\cref{eqA:balanced_condition}, both $\mathcal F$ and $e^{\widehat S}\mathcal F$ are balanced.

	Consider the squared Frobenius norm $f(t)=|e^{t\widehat S}\mathcal F|_F^2$ for $t\in[0,1]$, where $t$ interpolates between the two balanced tensors $\mathcal F$ and $\mathcal F'$. Its derivatives satisfy
	\begin{align*}
		f'(t) &=2\langle e^{t\widehat S}\mathcal F,\widehat S e^{t\widehat S}\mathcal F\rangle,\\
		f''(t) &=4|\widehat S e^{t\widehat S}\mathcal F|_F^2\geq0,
	\end{align*}
	so map $f$ is convex. For either balanced endpoint, we have $\operatorname{Tr}[\bm S^{(\alpha)}\bm R^{(\alpha)}]=0$ because $\bm R^{(\alpha)}\propto\bm I$ and $S^{(\alpha)}$ is traceless. Hence $f'(0)=f'(1)=0$. Convexity then forces $f'$ to vanish on $[0,1]$, so $f''\equiv0$ and in particular $\widehat S\mathcal F=0$. Therefore
	\begin{equation}
		\mathcal F'=z\left(\bigotimes_{\alpha=1}^{m}\bm L^{(\alpha)}\right)\mathcal F.
	\end{equation}
	and we are done with the proof.\hfill\qed
\end{proof}

\begin{corollary}
	Bipartite SVD spectra are well-defined up to the gauge orbit of the optimized fusion tensor.
\end{corollary}
\begin{proof}
	Under a residual transformation $\mathcal F\mapsto z(\bigotimes_\alpha\bm L^{(\alpha)})\mathcal F$, the site-resolved bipartite unfolding of any species transforms by left and right multiplication with unitary matrices together with an overall scalar $z$. Its singular values therefore change only by the common factor $|z|$, and the normalized spectra are unchanged. Consequently, the SVD spectra reported in the main text are independent of the initial gauge representative and are only functions of the gauge orbit of the optimized tensor.
\end{proof}
\begin{corollary}
	Occupied-orbital projectors and Bott tuples are well-defined up to the gauge orbit of the optimized fusion tensor.
\end{corollary}
\begin{proof}
	The reconstructed occupied-orbital amplitudes and projectors transform as
	\begin{align}
		\bm A'^{(\alpha)}&=\sqrt{|z|}\,\bm W^{(\alpha)}\bm A^{(\alpha)}\bm L^{(\alpha)T},\nonumber\\
		\bm P'^{(\alpha)}&=\bm W^{(\alpha)}\bm P^{(\alpha)}\bm W^{(\alpha)\dagger},
		\qquad \bm W^{(\alpha)}=\bigoplus_s\bm W_s^{(\alpha)},
	\end{align}
	where the $\bm W_s^{(\alpha)}$ are unitary channel transformations accounting for the local SVD frames. The exponentiated position operators $\bm X,\bm Y$ are scalar within each physical-orbital block and therefore commute with $\bm W^{(\alpha)}$. The compressed position operators (see below for details), their polar factors, and the Bott commutator thus transform by unitary conjugation, and every Bott eigenphase --- hence the Bott index --- is unchanged whenever the Bott construction is well defined. Species permutations only relabel the resulting tuple. Therefore the global balancing followed by the compact local-SVD gauge defines a parton Bott tuple up to species permutation.
\end{proof}

\section{Hyperparameter Settings and Energy Loss Curves}\label{app:hyperparameter_settings}
For both bosonic and fermionic extended Hubbard models introduced in the main text, the VMC optimization for the fusion tensor in our HyperDet wavefunction ansatz is performed using hyperparameters listed in~\cref{tabA:hyperparameter_settings}.
\begin{table}[H]
	\centering
	\begin{tabular}{cc}
		\hline\hline
		\textbf{VMC Hyperparameters} & \\
		\textsf{n\_burn\_in} & $4000$ \\
		\textsf{n\_samples} & $1000$ \\
		\textsf{n\_thinnings} & $1$ \\
		\textsf{n\_reburn\_in} & $0$ \\
		\textsf{n\_iters} & $800\sim1600$ \\
		\textbf{Optimizer Hyperparameters for SPRING~\cite{goldshlager2024kaczmarz}} & \\
		\textsf{learning\_rate} & $\mathbf{0.1\sim0.5}$ \\
		\textsf{learning\_rate\_decay} & $0\sim500$ \\
		\textsf{Tikhonov\_regularization} & $0.005\sim0.2$ \\
		KFAC-like norm constraint $C$~\cite{goldshlager2024kaczmarz} & $0.1\sim1.0$ \\
		momentum decay factor $\mu$~\cite{goldshlager2024kaczmarz} & $0.1\sim0.9$\\
		\hline\hline
	\end{tabular}
	\caption{Hyperparameter Settings used in VMC Optimization}
	\label{tabA:hyperparameter_settings}
\end{table}

\begin{figure*}
	\centering
	\includegraphics[width=1.0\textwidth]{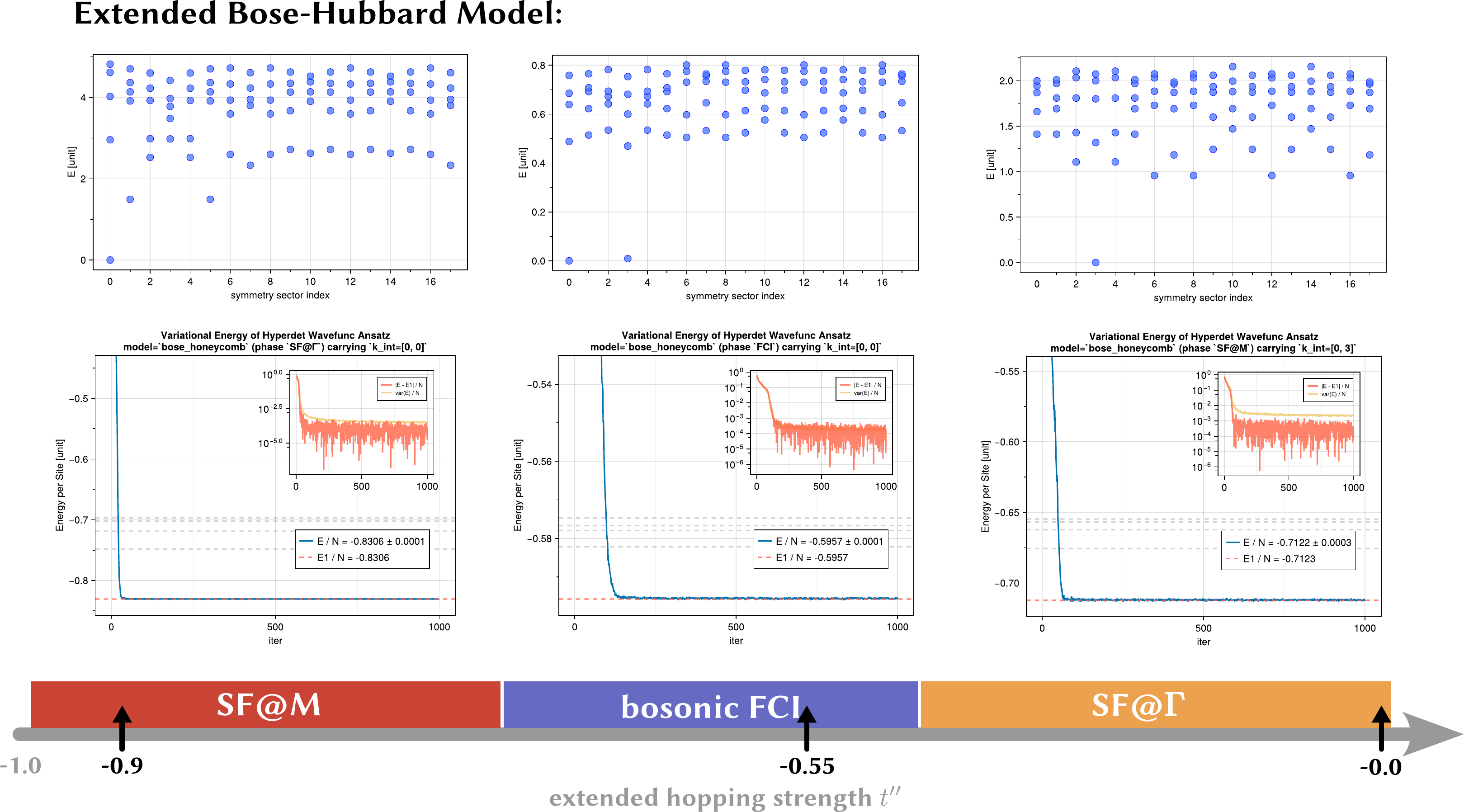}\\[2em]
	\includegraphics[width=1.0\textwidth]{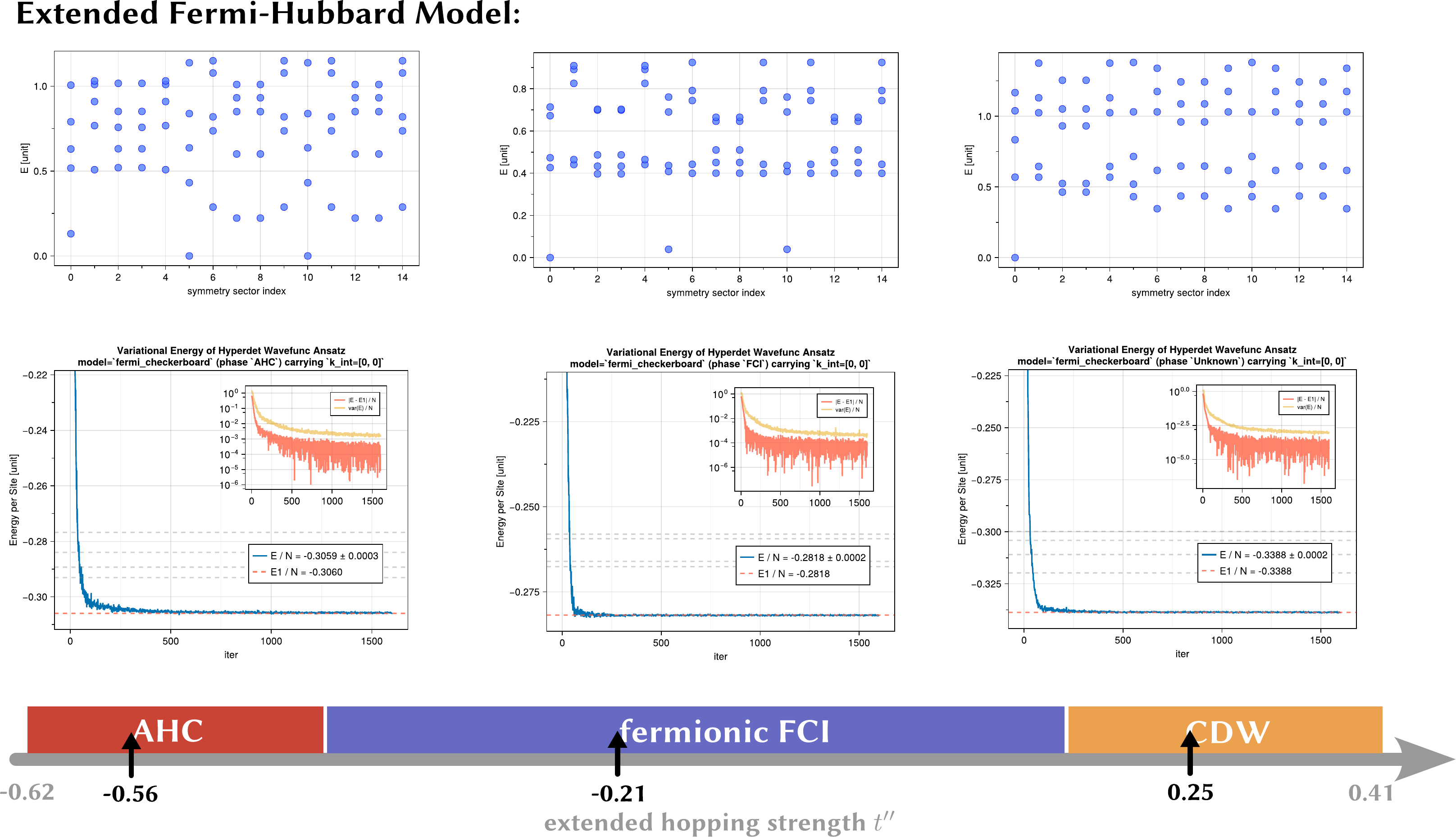}
	\caption{Selected ED Spectra and VMC Loss Curves for the Extended Bose-Hubbard (Top Panel) and the Fermi-Hubbard Model (Bottom Panel). For loss curves, we also add higher energy states at the target momentum sector as the gray dashed lines. The lowest-energy state with energy $E_1$ is marked red dashed lines. For the inset logarithmic plot of the loss curves, we add the energy variance and the error with respect to $E_1$ at the target sector. For each panel, we also add illustrative phase diagrams to mark the position of the three chosen values of the extended hopping strengths, highlighted with black arrows.}
	\label{figA:combined_ED_and_VMC_loss}
\end{figure*}

Here we list several ED spectra and VMC loss curves for both cases in~\cref{figA:combined_ED_and_VMC_loss}. Specifically, for the extended Bose-Hubbard model on the $3\times6\times2$ geometry, we take three typical values of the extended hopping strengths for the three phases determined by DMRG in Ref.~\cite{lu2025continuous}: $t''=-0.9$ for SF@M, $t''=-0.55$ for bosonic FCI, and $t''=-0.0$ for SF@$\Gamma$. For the extended Fermi-Hubbard model over $3\times5\times2$ geometry, we also take three typical values of the extended hopping strengths for the three phases identified using our ED calculations: $t''=-0.58$ for AHC, $t''=-0.21$ for fermionic FCI, and $t''=0.25$ for CDW.

It turns out that, under such aggressive hyperparameters in~\cref{tabA:hyperparameter_settings}, especially the large learning rate $\eta=0.1\sim0.5$ and small Tikhonov regularization $\lambda=0.005$, our HyperDet ansatz can converge to the exact ED ground state in about one hundred VMC iterations, far fewer than comparable ansatzes.

\section{Extended Data}\label{app:extended_data}
\subsection{VMC Performance for Other Geometries}
\begin{figure}[H]
	\centering
	\includegraphics[width=1.0\linewidth]{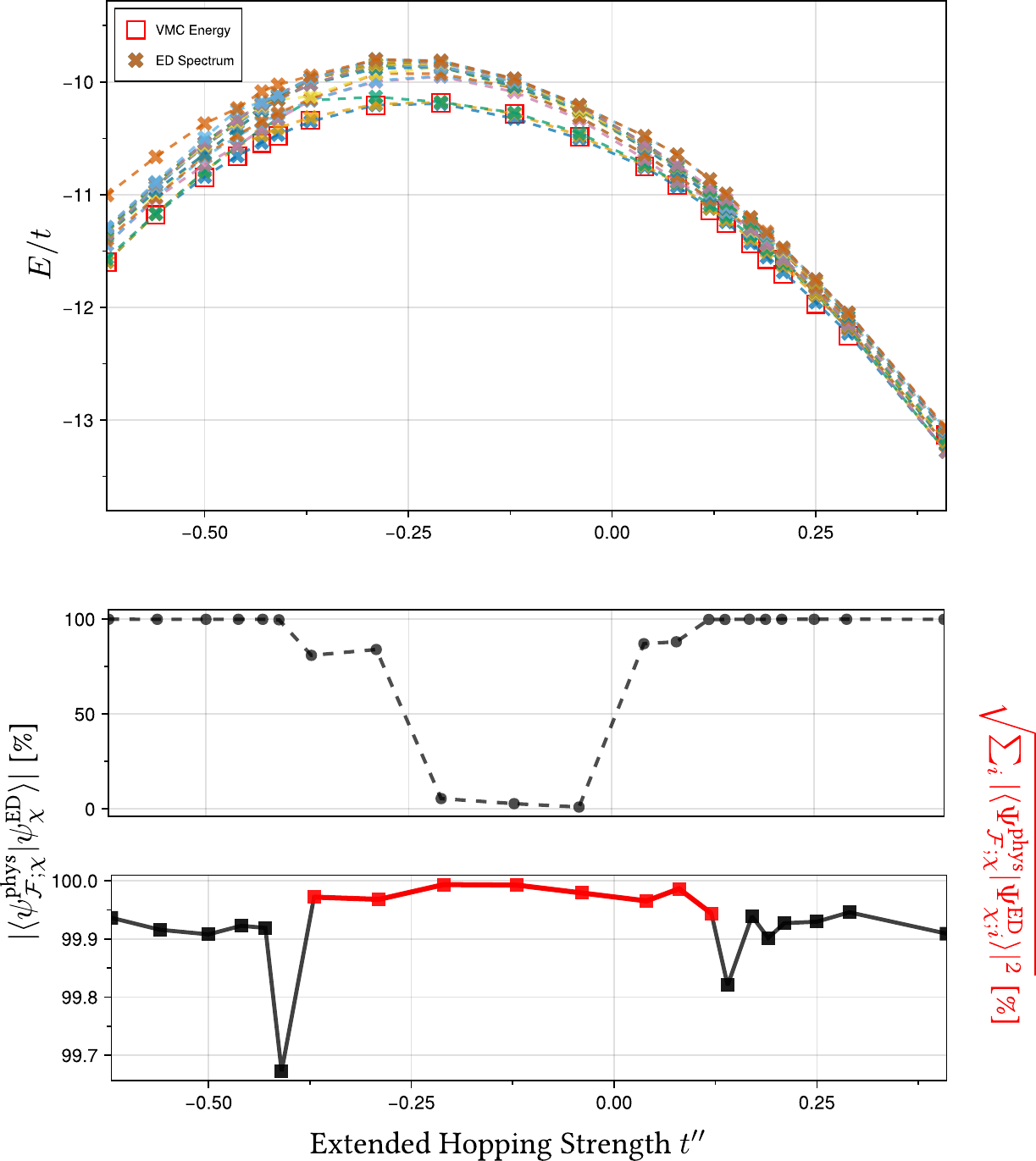}
	\caption{VMC energy vs ED spectrum (top), wavefunction overlap with respect to a single lowest-energy ground state (middle) and the ground-state manifold (bottom) on the $3\times6\times2$ geometry for the extended Fermi-Hubbard model at $\nu=1/3$. The solid black lines in the bottom plot are adapted from the data points in the middle plot for single lowest-energy ground state overlap. We only switch to the ground-state manifold overlap for parameter sweeps inside the FCI phase (the red lines).}
	\label{figA:combined_VMC_energy_and_Overlap_3_6}
\end{figure}

As a complement to the figure shown in the main text, here we present the VMC energy versus ED spectrum and the corresponding wavefunction overlap for the extended Fermi-Hubbard model~\cref{eqA:Fermi-Hubbard Hamiltonian} on a $3\times6\times2$ geometry in~\cref{figA:combined_VMC_energy_and_Overlap_3_6}. The fusion tensor optimization is performed over the same extended hopping strength sweep $t''\in[-0.62,0.41]$ as in the main text.

On this commensurate geometry, the three almost degenerate ground states collapse to the same momentum sector, and VMC will be ``trapped'' inside the ground state manifold as an artifact. To show this, we compute the ground-state manifold overlap $\sqrt{\sum_i|\langle\Psi_{\mathcal F;\chi}^{\text{phys}}|\Psi_{\chi;i}^{\text{ED}}\rangle|^2}$ inside the FCI phase, by projecting onto the lowest three ED states (red dots in~\cref{figA:combined_VMC_energy_and_Overlap_3_6}) rather than only the lowest state as we did outside of the FCI phase (black dots in~\cref{figA:combined_VMC_energy_and_Overlap_3_6}). Comparing the sharp drop of the overlap for the middle panel (black lines) with the exceptionally high overlap over $99.97\%$ for the red lines inside FCI phase in~\cref{figA:combined_VMC_energy_and_Overlap_3_6}, it is clear that such a drop is purely a limitation of the single-state diagnostic that can be removed.

The constantly high performance for the HyperDet wavefunction ansatz across different geometries thus establishes its strong expressiveness and applicability, although the phase boundaries shift due to the finite-size effects.

\subsection{Species-specific SVD Spectrum of Bipartite Fusion Matrix}
In the $m=3$ fermionic example of the main text, we present the species-averaged SVD spectrum of the bipartite fusion matrix to quotient out the $S_m$ parton species permutation redundancy. For completeness, here we show in~\cref{figA:fermionic_species_specific_bipartite_fusion_matrix_SVD_spectrum} the full three species-specific SVD spectra for the bipartite fusion matrix obtained from the parton bipartitions $\mathcal F(s_\ell)_{\ell_\alpha\mid\ell_{\bar\alpha}} \equiv \mathcal F_{s_\ell;\ell_1,\ldots,\ell_m}$.
\begin{figure*}
	\centering
	\includegraphics[width=1.0\linewidth]{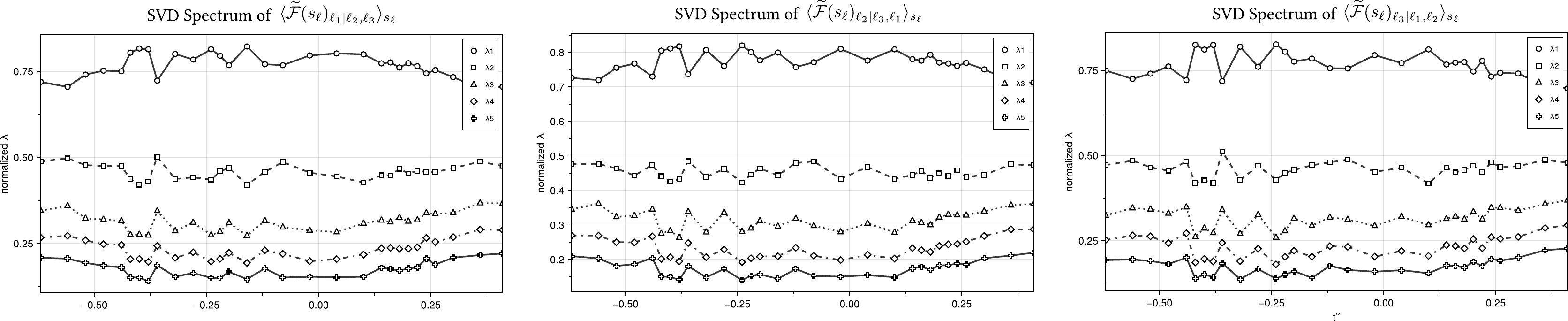}
	\caption{Species-specific SVD spectrum for the bipartite fusion matrix under the parton bipartitions $\mathcal F(s_\ell)_{\ell_1\mid\ell_2,\ell_3}$, $\mathcal F(s_\ell)_{\ell_2\mid\ell_1,\ell_3}$, and $\mathcal F(s_\ell)_{\ell_3\mid\ell_1,\ell_2}$ for the $m=3$ fermionic example.}
	\label{figA:fermionic_species_specific_bipartite_fusion_matrix_SVD_spectrum}
\end{figure*}

For $m=3$ there are three bipartitions: $\mathcal F(s_\ell)_{\ell_1\mid\ell_2,\ell_3}$, $\mathcal F(s_\ell)_{\ell_2\mid\ell_1,\ell_3}$, and $\mathcal F(s_\ell)_{\ell_3\mid\ell_1,\ell_2}$. Interestingly, we find that the three species-resolved SVD spectra nearly coincide throughout all three regimes and across independently optimized VMC tensors, which reveals a nontrivial emergent $S_3$-\emph{isotropy} of the balanced one-leg fusion spectra. Equality of these spectra alone does not establish an $S_3$-valued stabilizer of the complete fusion tensor; determining whether such an enlarged stabilizer is present requires a direct tensor-level fit.

\section{Full Virtual-lift Fitting Results}\label{app:full_translation_virtual_lifts}
In this section, we provide the complete translation virtual-lift fits underlying the results discussed in the main text. All fusion tensors analyzed below are obtained directly from the real-space VMC optimizations. Neither a projective translation algebra nor any translation-fractionalization phase is imposed during VMC optimization or the virtual-lift fitting.

\subsection{Translation Virtual-lift Fitting Protocol}
For each translation generator $T_\mu$ ($\mu=1,2$) and a fixed parton-species permutation $\sigma_{T_\mu}\in S_m$, we minimize the full-data residual
\begin{equation}\label{eqA:translation_virtual_lift_residual}
	\epsilon(\mathcal L_{T_\mu})^2=\frac{\displaystyle\sum_{s_\ell}\left\|
	\mathcal F(T_\mu\cdot s_\ell) - \mathcal L_{T_\mu}\cdot\mathcal F(s_\ell)\right\|_F^2}{\displaystyle\sum_{s_\ell}\left\|\mathcal F(T_\mu\cdot s_\ell)\right\|_F^2},
\end{equation}
where the translation-lifted action on the fusion tensor is
\begin{equation}
	\mathcal L_{T_\mu} = \Big(\bm G_{T_\mu}^{(1)},\ldots,\bm G_{T_\mu}^{(m)};\sigma_{T_\mu}\Big)\in\mathrm{GL}(N;\mathbb C)^m\rtimes S_m.
\end{equation}
For a fixed $\sigma_{T_\mu}$, the matrices $\bm G_{T_\mu}^{(\alpha)}$ are optimized by cyclic alternating least squares. Holding all other parton-leg transformations fixed reduces the update of each $\bm G_{T_\mu}^{(\alpha)}$ to an ordinary linear least-squares problem. We cycle over all parton legs from multiple initializations, rebalance scalar factors among the legs after every sweep, and enumerate all $m!$ species permutations. As a control of the non-compact $\mathrm{GL}$-freedom, we independently repeat the fitting under the restriction $\bm G_{T_\mu}^{(\alpha)}\in U(N)$. For the chosen optimized fusion tensors, the reported multi-start minima were further audited by joint Levenberg--Marquardt and L-BFGS optimizations: the three solvers agree within $2.33\times10^{-9}$ for every fixed-permutation translation branch so the minima is considered robust fits rather than local minima.

To examine translation fractionalization, we compose two fits using the physical relation
\begin{equation}
	T_1 T_2 T_1^{-1} T_2^{-1} = \bm e
\end{equation}
and evaluate
\begin{equation}\label{eqA:translation_lifted_commutator}
	\omega(\{T_1,T_2\}) = \mathcal L_{T_1}\mathcal L_{T_2}\mathcal L_{T_1}^{-1}\mathcal L_{T_2}^{-1}.
\end{equation}
The same species permutation $\sigma_{T_1}=\sigma_{T_2}=\sigma$ is used within each row of~\cref{tabA:translation_virtual_lift_fittings}, and all $\sigma\in S_m$ are enumerated. This scan over common permutations guarantees that the discrete part of~\cref{eqA:translation_lifted_commutator} closes. In particular, ``parton-identical'' means only $\sigma=\mathrm{id}$; the continuous matrices on different parton legs remain independently optimized.

Let $\bm M_\alpha$ denote the continuous matrix on parton leg $\alpha$ in the composed commutator~\cref{eqA:translation_lifted_commutator}. We extract its normalized diagonal-trace phase $\zeta_\alpha=\frac{\tr(\bm M_\alpha)}{\left|\tr(\bm M_\alpha)\right|}$. We separately quantify whether the \emph{full} matrices are close to the scalar stabilizer by
\begin{equation}\label{eqA:translation_centrality_error}
	\delta_{\mathrm{cent}}^2=\sum_{\alpha=1}^{m}\left\|\bm M_\alpha	- \frac{\tr(\bm M_\alpha)}{N}\bm I_N\right\|_F^2 \bigg/ \sum_{\alpha=1}^{m}	\left\|\bm M_\alpha\right\|_F^2.
\end{equation}
The numerical phase error $\delta_{\mathrm{diag}}$ compares only $(\zeta_1,\ldots,\zeta_m)$ with the common parton-filling target $(-1,-1)$ for $\nu=1/2$ bosons and $(e^{-2\pi i/3},e^{-2\pi i/3},e^{-2\pi i/3})$ for $\nu=1/3$ fermions. Thus off-diagonal components and nonuniform diagonal entries contribute to $\delta_{\mathrm{cent}}$, but are not counted again in $\delta_{\mathrm{diag}}$. These targets are used only to assess the fitted commutator \emph{after} optimization.

On a finite $L_1\times L_2$ torus, a scalar translation commutator $\mathcal L_{T_1}\mathcal L_{T_2}=z\mathcal L_{T_2}\mathcal L_{T_1}$ is constrained by
\begin{equation}\label{eqA:finite_torus_translation_fractionalization}
	z\in \mu_q=\Big\{e^{i2\pi k/q}:k=0,1,\ldots,q-1\Big\},
\end{equation}
where $q=\gcd(N,L_1,L_2)$ and $\mu_q$ denotes the $q$th roots of unity. This follows from the determinant constraint $z^N=1$ and the two boundary relations $z^{L_1}=z^{L_2}=1$. In addition, an ambiguity inside the scalar stabilizer must satisfy $\prod_{\alpha=1}^{m}z_\alpha=1$. Therefore the complete set of geometry-compatible translation-fractionalization tuples is
\begin{equation}\label{eqA:allowed_translation_fractionalization_tuples}
	\mathcal Z_{m,q} = \left\{(z_1,\ldots,z_m)\in\mu_q^m\;\middle|\;\prod_{\alpha=1}^{m}z_\alpha=1\right\}.
\end{equation}
The $4\times4\times2$ bosonic geometry has $N=8$ and $q=4$. For $m=2$, it permits all four tuples
\begin{equation}\label{eqA:allowed_bosonic_translation_tuples}
	\mathcal Z_{2,4} = \Big\{(1,1),(-1,-1),(i,-i),(-i,i)\Big\}.
\end{equation}
The species-identical subset contains $(1,1)$ and $(-1,-1)$, so the finite geometry alone does not select the nontrivial $\pi$-flux class. The $3\times6\times2$ fermionic geometry has $N=6$ and $q=3$. Defining $\omega=e^{-2\pi i/3}$, its nine allowed tuples are compactly written as
\begin{align}\label{eqA:allowed_fermionic_translation_tuples}
	\mathcal Z_{3,3} &= \Big\{(1,1,1),(\omega,\omega,\omega),(\omega^2,\omega^2,\omega^2)\Big\}\cup\mathsf{Perm}\{(1,\omega,\omega^2)\},
\end{align}
where the last set contains all six distinct permutations of the tuple $(1,\omega,\omega^2)$. Its species-identical subset therefore contains three different compatible classes. Importantly, neither the VMC optimization nor the virtual-lift fitting is restricted to the discrete sets~\cref{eqA:allowed_bosonic_translation_tuples,eqA:allowed_fermionic_translation_tuples}: the fitted matrices are fully complex and continuous. Consequently, recovering the expected tuple from among these geometry-allowed alternatives is a nontrivial result of energy minimization followed by the unconstrained fitting.

\begin{table*}
	\centering
	\caption{Complete common-permutation translation virtual-lift fittings for the raw VMC-optimized fusion tensors. The tuple $\sigma=(\sigma(1),\ldots,\sigma(m))$ specifies the parton-species permutation. The fitting residual is denoted as $\epsilon^{\mathrm{GL}}[\epsilon^{ U}]$, where the bracketed value is the unitary-constrained residual. The unconstrained phase tuple $\bm\zeta=(\zeta_1,\ldots,\zeta_m)$ and the two relation errors are computed from the GL fits. The phase target is $-1$ for the two bosonic tensors and $\omega=e^{-2\pi i/3}$ for the fermionic tensor. All fitted values are rounded to at most three decimal places.}
	\label{tabA:translation_virtual_lift_fittings}
	\resizebox{\textwidth}{!}{
		\begin{tabular}{c c c c c c c}
			\hline\hline
			optimized state & $\sigma\in S_m$ & $\epsilon_{T_1}^{\mathrm{GL}}[\epsilon_{T_1}^{ U}]$ & $\epsilon_{T_2}^{\mathrm{GL}}[\epsilon_{T_2}^{ U}]$ & $\bm\zeta$ & $\delta_{\mathrm{cent}}$ & $\delta_{\mathrm{diag}}$ \\
			\hline
			bosonic FCI, $t''=-0.60$ & $(1,2)$ & $0.021[0.021]$ & $0.021[0.021]$ & $(-1,-1)$ & $1.16\times10^{-4}$ & $1.36\times10^{-9}$ \\
			& $(2,1)$ & $0.287[0.290]$ & $0.287[0.290]$ & $(-0.179+0.984i,-0.179-0.984i)$ & $2.46\times10^{-3}$ & $1.28$ \\
			\hline
			bosonic SF@$\Gamma$, $t''=-0.45$ & $(1,2)$ & $0.045[0.045]$ & $0.045[0.045]$ & $(-1,-1)$ & $5.11\times10^{-4}$ & $2.32\times10^{-8}$ \\
			& $(2,1)$ & $0.493[0.515]$ & $0.494[0.515]$ & $(0.510+0.860i,0.658-0.753i)$ & $1.00$ & $1.78$ \\
			\hline
			fermionic FCI, $t''=-0.12$ & $(1,2,3)$ & $0.761[0.848]$ & $0.815[0.930]$ & $(-0.506-0.863i,-0.508-0.862i,-0.505-0.863i)$ & $0.715$ & $7.16\times10^{-3}$ \\
			& $(1,3,2)$ & $0.805[0.917]$ & $0.675[0.729]$ & $(-0.491-0.871i,0.525+0.851i,0.498+0.867i)$ & $0.671$ & $1.63$ \\
			& $(2,1,3)$ & $0.725[0.794]$ & $0.833[0.958]$ & $(-0.997-0.077i,0.431-0.902i,-0.498-0.867i)$ & $0.511$ & $0.761$ \\
			& $(2,3,1)$ & $0.795[0.897]$ & $0.764[0.849]$ & $(-0.266-0.964i,0.975-0.222i,-0.415+0.910i)$ & $0.922$ & $1.39$ \\
			& $(3,1,2)$ & $0.761[0.847]$ & $0.741[0.817]$ & $(-0.859+0.512i,-0.368+0.930i,-0.122+0.992i)$ & $0.484$ & $1.72$ \\
			& $(3,2,1)$ & $0.806[0.910]$ & $0.786[0.880]$ & $(-0.226-0.974i,-0.493-0.870i,-0.730-0.684i)$ & $0.602$ & $0.240$ \\
			\hline\hline
		\end{tabular}
	}
\end{table*}

\subsection{Translation Fractionalization Results}
For the bosonic FCI phase at $t''=-0.60$, both translations select the parton-identical branch from the optimized fusion tensor, with $\epsilon(\mathcal L_{T_1})=0.021$ and $\epsilon(\mathcal L_{T_2})=0.021$. The species-exchanging branches remain large residuals near $0.287$. Among the four tuples allowed by~\cref{eqA:allowed_bosonic_translation_tuples}, the independently fitted parton-identical lifts select
\begin{equation}
	\omega(\{T_1,T_2\})\simeq\big(-\bm I_N,-\bm I_N;\mathrm{id}\big),
\end{equation}
with $\delta_{\mathrm{cent}}=1.16\times10^{-4}$ and $\delta_{\mathrm{diag}}=1.36\times10^{-9}$. The agreement between the two species is therefore not an artifact of trace averaging: the full commutator matrices themselves are already nearly scalar. This provides a decisive numerical recovery of the expected $\pi$-flux translation fractionalization directly from the optimized fusion tensor.

The bosonic SF@$\Gamma$ phase at $t''=-0.45$ also selects parton-identical translation lifts, with residuals $0.045$ and $0.045$, while the species-exchanging branches stay near $0.493$. Its lifted commutator again selects the nontrivial allowed tuple $(-1,-1)$, with $\delta_{\mathrm{cent}}=5.11\times10^{-4}$ and $\delta_{\mathrm{diag}}=2.32\times10^{-8}$. Thus the virtual translation class persists across the studied FCI--SF@$\Gamma$ sweep, consistent with the field-theory descriptions~\cite{barkeshli2014continuous,barkeshli2015continuous}. Note: this observed parton translation fractionalization is a property of the projective symmetry group action. It does not indicate any presence of anyonic excitation or anyon fractionalization in the SF@$\Gamma$ phase, which is a trivial superfluid.

The fermionic tensor behaves qualitatively differently. The independently optimized generator minima occur at different species permutations:
\begin{align*}
	\min_{\sigma\in S_3}\epsilon(\mathcal L_{T_1})&=0.725,\quad\sigma_{T_1}=(2,1,3),\\
	\min_{\sigma\in S_3}\epsilon(\mathcal L_{T_2})&=0.675,\quad\sigma_{T_2}=(1,3,2).
\end{align*}
These two independently selected permutations are not combined into a commutator in~\cref{tabA:translation_virtual_lift_fittings}. More importantly, both residuals remain large, as do all other permutation branches. This failure is stable across the VMC history: three snapshots give $(\epsilon_{T_1},\epsilon_{T_2})=(0.735,0.680)$, $(0.727,0.677)$, and $(0.725,0.675)$, respectively. It is also reproduced by the three independent fitting algorithms described above.

Within the common parton-identical branch, the fitted commutator nevertheless gives $(\zeta_1,\zeta_2,\zeta_3)=(-0.506-0.863i,\,-0.508-0.862i,\,-0.505-0.863i)$, all close to $\omega=e^{-2\pi i/3}$. Among the nine tuples allowed by~\cref{eqA:allowed_fermionic_translation_tuples}, the diagonal traces therefore select the expected species-identical tuple $(\omega,\omega,\omega)$, with $\delta_{\mathrm{diag}}=7.16\times10^{-3}$. At the same time, the full-matrix centrality error is $\delta_{\mathrm{cent}}=0.715$, and the two residuals of the generator fits are $0.761$ and $0.815$. Therefore the diagonal phase agreement cannot be promoted to a faithful extraction of fermionic translation fractionalization from the present minimal lift. It is a suggestive trace-level feature that survives an otherwise unsuccessful structural fit.

\section{Parton Chern Numbers from the Optimized Fusion Tensor}\label{app:parton_chern_numbers}
In this section we turn the optimized fusion tensor into a quantitative diagnostic of parton-band topology. We construct an occupied-orbital projector for each parton species and extract its topological invariant, the parton Chern number. The calculation requires no translation virtual-lift fit and no assumed parton band structure, and applies to the optimized fusion tensors of both FCI phases analyzed in the main text: $t''=-0.60$ on the $4\times4\times2$ geometry for the bosonic case, and $t''=-0.12$ on the $3\times6\times2$ geometry for the fermionic case.

\subsection{Bott Index and Parton Chern Numbers}
The determinant expansion re-writing of the HyperDet~\cref{eqA:determinant expansion of HyperDet wavefunction} proves that a nonzero wavefunction requires the occupied-orbital amplitudes to have $\mathop{\mathrm{rank}}\bm A^{(\alpha)}=N$ and that its columns define an exact reconstructed Slater realization. We thus can use the orthonormal polar frame to construct the \emph{occupied-orbital projector}:
\begin{align}\label{eqA:occupied orbital projector}
	\bm Q^{(\alpha)}&=\bm A^{(\alpha)}\bigl(\bm A^{(\alpha)\dagger}\bm A^{(\alpha)}\bigr)^{-1/2},\nonumber\\
	\bm P^{(\alpha)}&=\bm Q^{(\alpha)}\bm Q^{(\alpha)\dagger}=\bm A^{(\alpha)}\big(\bm A^{(\alpha)\dagger}\bm A^{(\alpha)}\big)^{-1}\bm A^{(\alpha)\dagger}.
\end{align}
For the optimized fusion tensor in the bosonic and fermionic cases, the ranks of the constructed projector are $N=8$ and $N=6$, respectively.

The Chern number of the projector $\bm P^{(\alpha)}$ can be computed over the 2D grids of boundary twists, but that would require additional expensive simulation of VMC data over this twist grid. Instead, we use the Bott index, the standard real-space proxy for the Chern number of a lattice projector. In the real-space physical-leg basis, we define the exponentiated position operators
\begin{equation}\label{eqA:parton_position_unitaries}
	\begin{aligned}
		\bm X &= \sum_{s_\ell=1}^{N_s}\sum_{c=1}^{r^{(\alpha)}_{s_\ell}} e^{2\pi i x_{s_\ell}/L_1} |s_\ell,c\rangle\langle s_\ell,c|,\\
		\bm Y &= \sum_{s_\ell=1}^{N_s}\sum_{c=1}^{r^{(\alpha)}_{s_\ell}} e^{2\pi i y_{s_\ell}/L_2} |s_\ell,c\rangle\langle s_\ell,c|,
	\end{aligned}
\end{equation}
where $(x_{s_\ell},y_{s_\ell})$ are the crystal coordinates of physical orbital $s_\ell$ folded onto the $L_1\times L_2$ torus. The matrices $\bm X$ and $\bm Y$ are exactly unitary and, crucially, scalar within each physical-orbital block, since all components $c$ at the same physical orbital share the same phase. This property is what makes the unitary gauge freedom below harmless.

To measure the topology of the occupied subspace, we complete the compressed position operators by the identity on the unoccupied complement and require that their restrictions to the occupied subspace be nonsingular before taking their unitary polar factors:
\begin{equation}\label{eqA:parton_bott_compression}
	\begin{aligned}
		\widetilde{\bm U}_{\bm X}^{(\alpha)} &= \bm P^{(\alpha)}\bm X\bm P^{(\alpha)} + (\bm 1-\bm P^{(\alpha)}), \\
		\widetilde{\bm V}_{\bm Y}^{(\alpha)} &= \bm P^{(\alpha)}\bm Y\bm P^{(\alpha)} + (\bm 1-\bm P^{(\alpha)}),
	\end{aligned}
\end{equation}
and replace each by its unitary polar factor, i.e., the nearest unitary,
\begin{equation}
	\bm U_{\bm X}^{(\alpha)}=\mathop{\mathrm{polar}}(\widetilde{\bm U}_{\bm X}^{(\alpha)}),\quad \bm V_{\bm Y}^{(\alpha)}=\mathop{\mathrm{polar}}(\widetilde{\bm V}_{\bm Y}^{(\alpha)})
\end{equation}
where for matrix $\bm Z=\bm U_Z\bm\Sigma_Z\bm V_Z^\dagger$ its polar factor is defined as $\mathop{\mathrm{polar}}(\bm Z)=\bm U_Z\bm V_Z^\dagger$. The product of these polar factors defines a Wilson-loop-like commutator:
\begin{equation}\label{eqA:parton_bott_index}
	\bm W^{(\alpha)}=\bm U_{\bm X}^{(\alpha)}\bm V_{\bm Y}^{(\alpha)}\bm U_{\bm X}^{(\alpha)\dagger}\bm V_{\bm Y}^{(\alpha)\dagger},
\end{equation}
and the \emph{Bott index} of the occupied-orbital projector is the total winding of its principal eigenphases ($\vartheta^{(\alpha)}_j\in(-\pi,\pi]$)
\begin{equation}
	\mathop{\mathsf{Bott}}(\bm P^{(\alpha)})=\frac{1}{2\pi}\Im\mathop{\mathrm{Tr}}\log\bm W^{(\alpha)}=\frac{1}{2\pi}\sum_{j=1}^{M_\alpha}\vartheta^{(\alpha)}_j,
\end{equation}

Every factor in the Bott commutator is unitary, so its determinant is one. If no eigenphase lies on the principal-logarithm branch cut, the sum of principal eigenphases is an integer multiple of $2\pi$. This integer can change only if an eigenphase crosses $\pm\pi$, and it is therefore constant under any continuous deformation of $\bm U^{(\alpha)}$ and $\bm V^{(\alpha)}$ that keeps the polar construction well-defined, and the \emph{branch gap} $\Delta_\vartheta\equiv\min_j(\pi-|\vartheta^{(\alpha)}_j|)$ remains open.

Interpreting this finite-system integer as a bulk Chern number for the auxiliary partons requires additional locality assumptions. We assume that the reconstructed occupied subspaces admit local or sufficiently quasilocal insulating realizations with smooth, gapped families under auxiliary boundary twists (as \cref{ass:parton insulator} stated in the main text), in the same setup as Loring and Hastings has assumed~\cite{hastings2010almost,hastings2011topological} in the operator-level $C^*$-algebraic formulation of Kitaev's construction, which identifies this integral Bott index with the first Chern number of the occupied bundle $\bm P^{(\alpha)}$
\begin{equation}\label{eqA:Bott index to Chern number}
	\mathsf{Bott}(\bm P^{(\alpha)}) = \mathcal C^{(\alpha)}.
\end{equation}
Importantly, the interpretation of $\bm P^{(\alpha)}$ as the one-body correlation matrix Eq.~\eqref{eq:one-body correlation matrix meaning of occupied-orbital projector} in the main text does \emph{not} require an appeal to the factorized limit, nor does it require identifying the SVD amplitudes $\bm A^{(\alpha)}$ with the occupied parton orbitals fed into the original fusion tensor. Therefore, the correspondence~\cref{eqA:Bott index to Chern number} remains valid in the generic case under \cref{ass:parton insulator}. We therefore record, alongside the Bott index, the distance to the nearest integer and the branch gap as diagnostics of how well defined the assignment is.

\begin{table}[H]
	\centering
	\begin{tabular}{lc|cccccc}
		\hline\hline
		Example & species-$\alpha$ & $\mathop{\mathrm{rank}}\bm P^{(\alpha)}$ & $\mathop{\mathsf{Bott}}(\bm P^{(\alpha)})$ & branch gap $\Delta_\vartheta$ \\
		\hline
		bosonic & $1$ & $8$ & $1.0$ & $2.2951$ \\
		& $2$ & $8$ & $1.0$ & $2.2931$ \\
		\hline\hline
		fermionic & $1$ & $6$ & $1.0$ & $0.6610$ \\
		& $2$ & $6$ & $1.0$ & $0.5750$ \\
		& $3$ & $6$ & $1.0$ & $0.7974$ \\
		\hline\hline
	\end{tabular}
	\caption{Parton Bott-index analysis of the optimized bosonic FCI tensor at $t''=-0.60$ on the $4\times4\times2$ geometry and the fermionic FCI tensor at $t''=-0.12$ on the $3\times6\times2$ geometry.}
	\label{tabA:parton_bott_results}
\end{table}
All five Bott values, for the two bosonic and three fermionic species are reported in~\cref{tabA:parton_bott_results}, within $5.6\times10^{-16}$ of an integer. The branch gaps are all larger than $0.575$, so the principal logarithm is far from its branch cut and the assignment is unambiguous. For the chosen optimized fusion tensors, we identify a tuple of parton Chern numbers $(\mathcal C^{(1)},\mathcal C^{(2)})=(1,1)$ in the bosonic phases, and $(\mathcal C^{(1)},\mathcal C^{(2)},\mathcal C^{(3)})=(1,1,1)$ in the fermionic phase. These parton Chern numbers match the parton topological content prescribed by the low-energy field theory descriptions of the $\nu=1/2$ bosonic FCI and the $\nu=1/3$ fermionic FCI as the Chern-Simons levels, respectively. We emphasize that in those treatments the parton Chern numbers are input as an assumption, whereas here they are extracted directly from the optimized fusion tensor without theoretical input.

\subsection{Gauge Robustness Audits}
We close this section by discussing the gauge freedom of the construction, which has two distinct aspects that should not be conflated. We write $\bm A^{(\alpha)}_{s_\ell}$ for the block of rows $(s_\ell,c)$ of $\bm A^{(\alpha)}$ belonging to a fixed physical orbital $s_\ell$.
\begin{itemize}
	\item \emph{Unitary SVD gauge.} At a fixed $s_\ell$, the retained singular subspace carries a $U(r^{(\alpha)}_{s_\ell})$ freedom: a unitary $\bm Q_{s_\ell}$ acting as $\bm A^{(\alpha)}_{s_\ell}\mapsto\bm Q_{s_\ell}\bm A^{(\alpha)}_{s_\ell}$ leaves $\bm P_{s_\ell}^{(\alpha)}$ invariant. Because $\bm X$ and $\bm Y$ are scalar within each $s_\ell$-block, the occupied projector transforms covariantly and every Bott eigenphase is intrinsically unchanged.
	\item \emph{Non-unitary refactorizations.} More generally, any invertible but not necessarily unitary $\bm Q_{s_\ell}$ with $\bm A^{(\alpha)}_{s_\ell}\mapsto\bm Q_{s_\ell}\bm A^{(\alpha)}_{s_\ell}$ and $\bm B^{(\bar\alpha)}_{s_\ell}\mapsto\bm Q_{s_\ell}^{-T}\bm B^{(\bar\alpha)}_{s_\ell}$ also exactly preserves the mode-$\alpha$ unfolding. However, such a map generically changes the column spaces of $\bm A^{(\alpha)}_{s_\ell}$, hence the projectors $\bm P^{(\alpha)}$ and, in principle, the Bott index. The parton Chern numbers are therefore not invariant under the full $\mathrm{GL}$ refactorization freedom; they are defined only after fixing a convention. Our convention is the \emph{compact local-SVD gauge}, i.e., the canonical Hermitian choice obtained directly from the SVD of $\bm F_{s_\ell}^{(\alpha)}$, in which the retained columns of each $\bm A^{(\alpha)}_{s_\ell}$ are orthogonal and weighted by $\sqrt{\lambda^{(\alpha)}_c(s_\ell)}$.
\end{itemize}
This refactorization/SVD gauge freedom must be carefully distinguished from the wavefunction gauge group $\mathsf{GG}=\mathrm{GL}(N;\mathbb C)^m\rtimes S_m$. The latter leaves the physical state invariant and, after balancing, acts only by residual unitary transformations on the virtual parton legs as we have proved before. It therefore leaves the occupied-orbital projector invariant up to unitary conjugation and cannot change the Bott index. The refactorization freedom, by contrast, is a redundancy of the auxiliary-channel \emph{interpretation} for a fixed optimized physical state: a non-unitary $\bm Q_{s_\ell}$ can change the column space of $\bm A^{(\alpha)}_{s_\ell}$ and hence $\bm P^{(\alpha)}$. The reconstructed parton Chern number is therefore defined with respect to the compact local-SVD convention adopted here, and it is this convention, rather than the physical wavefunction, that fixes its value. The stress test below establishes that the reported integers are stable throughout the tested neighborhood of that convention.

We therefore test the stability of the Bott index under the non-unitary refactorizations (ii). We apply positive-definite stress maps $\bm Q_{s_\ell}=e^{s\bm H_{s_\ell}}$ to $\bm A^{(\alpha)}$ at increasing strengths $s=0.1$ to $0.5$, where $\bm H_{s_\ell}$ is a normalized random Hermitian matrix at physical site $s_\ell$. As guaranteed by the transformation law, the tensor reconstruction remains exact at machine precision, while the Bott analysis returns the pair $(1,1)$ and $(1,1,1)$ for bosonic and fermionic case in every trial under a minimum branch gap above $1.76$ and $0.46$ respectively. These reported parton Chern numbers are thus stable throughout the tested gauge neighborhood.

\section{Exact-Diagonalization Diagnostics for the Fermionic Phase Diagram}\label{app:fermion_ed_diagnostics}
Here we summarize the ED diagnostics used to identify the FCI, AHC, and the candidate CDW phases of the extended Fermi-Hubbard model on the checkerboard lattice at $\nu=1/3$ filling:
\begin{align}
	H &= -\sum_{\langle i,j\rangle}t e^{i\phi_{i,j}}c_{i}^{\dagger}c_{j}- \sum_{\langle\langle i,j\rangle\rangle}t'_{i,j} c_{i}^{\dagger}c_{j}- \sum_{\langle\langle\langle i,j\rangle\rangle\rangle}t'' c_{i}^{\dagger}c_{j} + \text{h.c.}\nonumber\\
	& + V_{1}\sum_{\langle i,j\rangle}n_{i}n_{j} + V_{2}\sum_{\langle\langle i,j\rangle\rangle}n_{i}n_{j} + V_{3}\sum_{\langle\langle\langle i,j\rangle\rangle\rangle}n_{i}n_{j}.\label{eqA:Fermi-Hubbard Hamiltonian}
\end{align}
As emphasized in the main text, unlike the flatband limit widely used in the literature~\cite{sun2011nearly,sheng2011fractional} where $t=1$ as the energy unit and
\begin{equation}\label{eqA:flatband limit parameter settings}
	t'_{1,2}=\pm1/(2+\sqrt2),\quad t''=1/(2+2\sqrt2),\quad\phi_{i,j}=\pm\pi/4,
\end{equation}
we push the model \emph{away} from this flatband limit, by choosing modified next-nearest-neighbor hoppings
\begin{equation}\label{eqA:new t_prime}
	t'_{1,2}=\pm1.4/(2+\sqrt2),
\end{equation}
and consider the strong-coupling limit with all three interacting terms present: $(V_1,V_2,V_3)=(5.6,1.26,0.56)$. This strong-coupling regime has not been studied previously and therefore provides a robust testbed for our HyperDet wavefunction ansatz. The single-particle band structure differences are shown in~\cref{figA:checkerboard_model_band_structure_comparisons}.
\begin{figure}[H]
	\centering
	\includegraphics[width=0.9\linewidth]{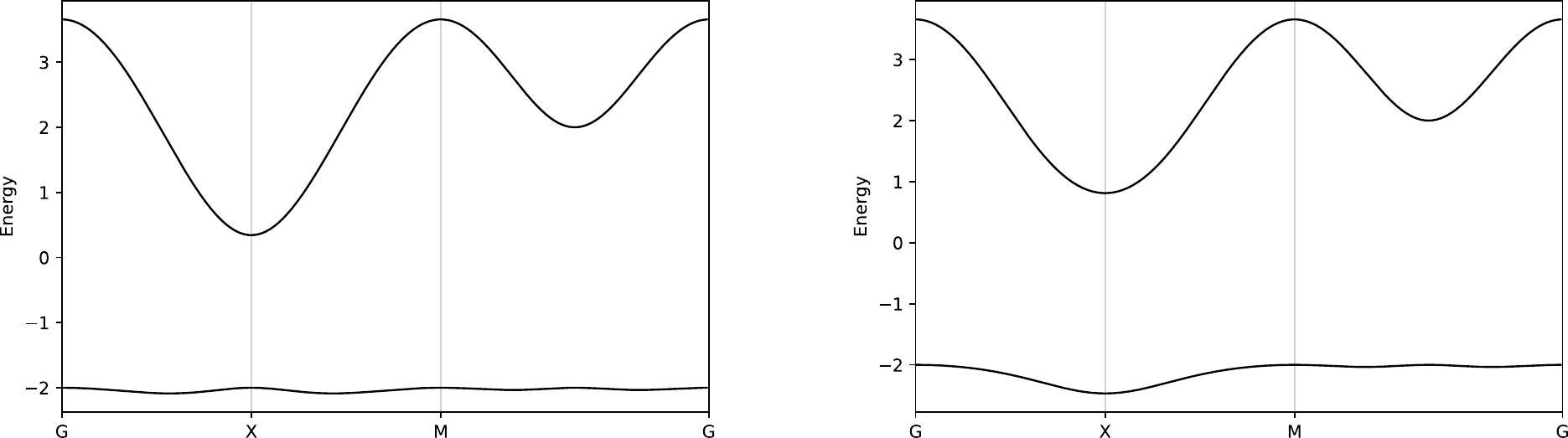}
	\caption{Single-particle Band Structure of the Extended Checkerboard Model~\cref{eqA:Fermi-Hubbard Hamiltonian}. Left: the flatband limit~\cref{eqA:flatband limit parameter settings} commonly used in FCI studies~\cite{sun2011nearly,sheng2011fractional}. Right: the parameter regime we studied here, with the modified extended hopping strengths~\cref{eqA:new t_prime}.}
	\label{figA:checkerboard_model_band_structure_comparisons}
\end{figure}

Our ED and HyperDet VMC calculations are all performed under this new unexplored parameter set. But before presenting the results, we clarify that
\begin{itemize}
	\item This section does \emph{not} aim to provide an accurate \emph{thermodynamic} phase diagram of model~\cref{eqA:Fermi-Hubbard Hamiltonian}, but only tries to provide numerical evidence sufficient to label the phases specific to the $3\times5\times2$ geometry we focused on VMC calculations in the main text. Our goal is to show that our HyperDet wavefunction ansatz is capable of capturing the ED ground states across these phases at this geometry.
\end{itemize}
In particular, different geometries can shift the positions of phase boundaries, and may even modify the phase content themselves.

\Cref{figA:fermi_Hubbard_ED_spec_sweep} shows the twelve lowest ED spectra along the sweep of the extended hopping strength $t''\in[-0.62,0.41]$. It is clear that in the middle region of the sweep, for approximately $t''\in[-0.42,0.15]$, there is a nearly three-fold degenerate ground state manifold in green, blue, and orange colors separated from higher excited states, which undergoes a same-sector level crossing at the left boundary, and a distinct-sector roton softening at the right boundary.
\begin{figure}[H]
	\centering
	\includegraphics[width=1.0\linewidth]{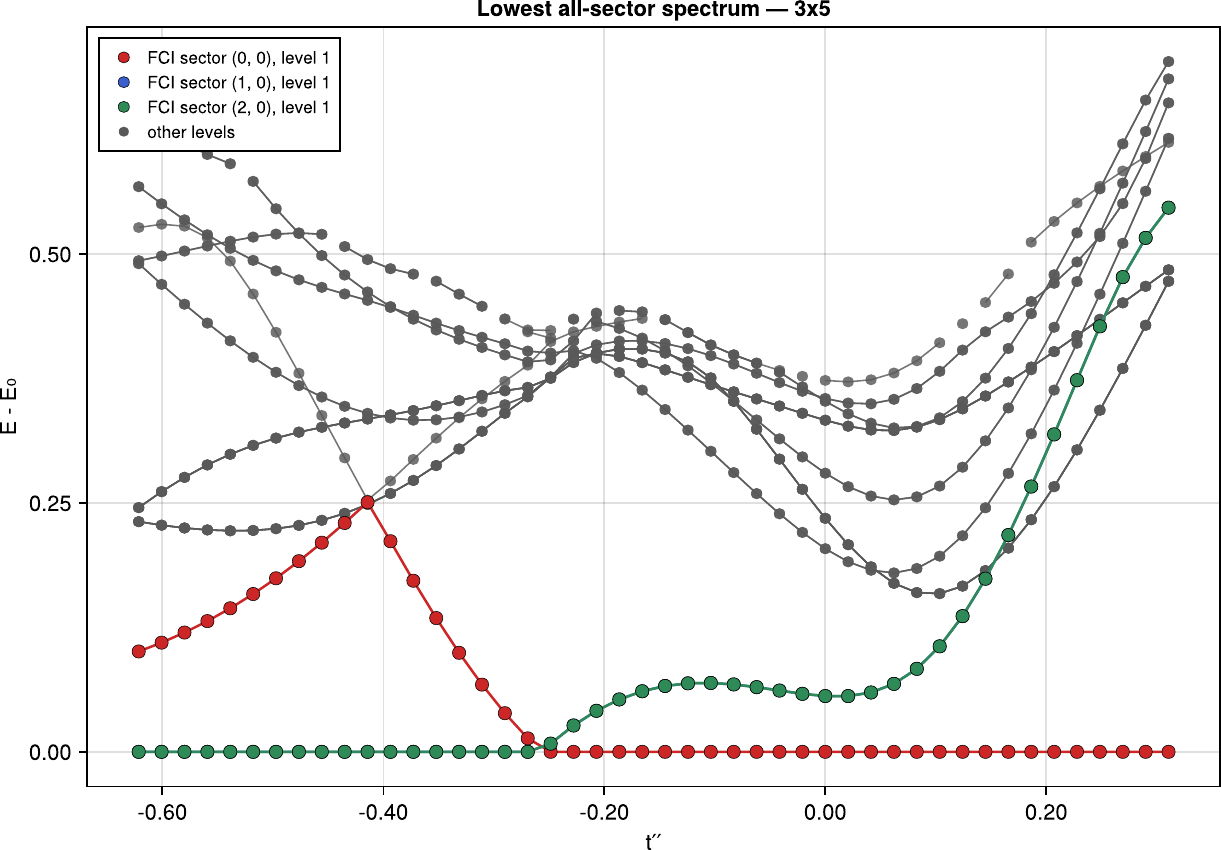}
	\caption{ED Spectra Sweep over the Extended Hopping Strengths $t''\in[-0.62,0.41]$ on a $3\times5\times2$ Geometry. Sectors $(1,0)$ (blue) and $(2,0)$ (green) are exactly degenerate due to antiunitary mirror symmetry of the model~\cref{eqA:Fermi-Hubbard Hamiltonian}. There is a clear $(0,0)$-sector level crossing at about $t''\approx-0.42$ signaling a first-order FCI--AHC transition, and a roton mode softening from other sectors at $t''\approx+0.15$ signaling a second-order FCI--CDW transition.}
	\label{figA:fermi_Hubbard_ED_spec_sweep}
\end{figure}

\subsection{Diagnostics of the FCI Phase}
To determine the fractionally quantized many-body Chern number of the FCI phase, we impose twisted boundary conditions at the representative point $t''=-0.15$ on a $3\times5\times2$ geometry and compute the spectral flow and many-body charge pump plot in~\cref{figA:fermionic_FCI_spectral_flow_and_charge_pump}.
\begin{figure}[H]
	\centering
	\includegraphics[width=1.0\linewidth]{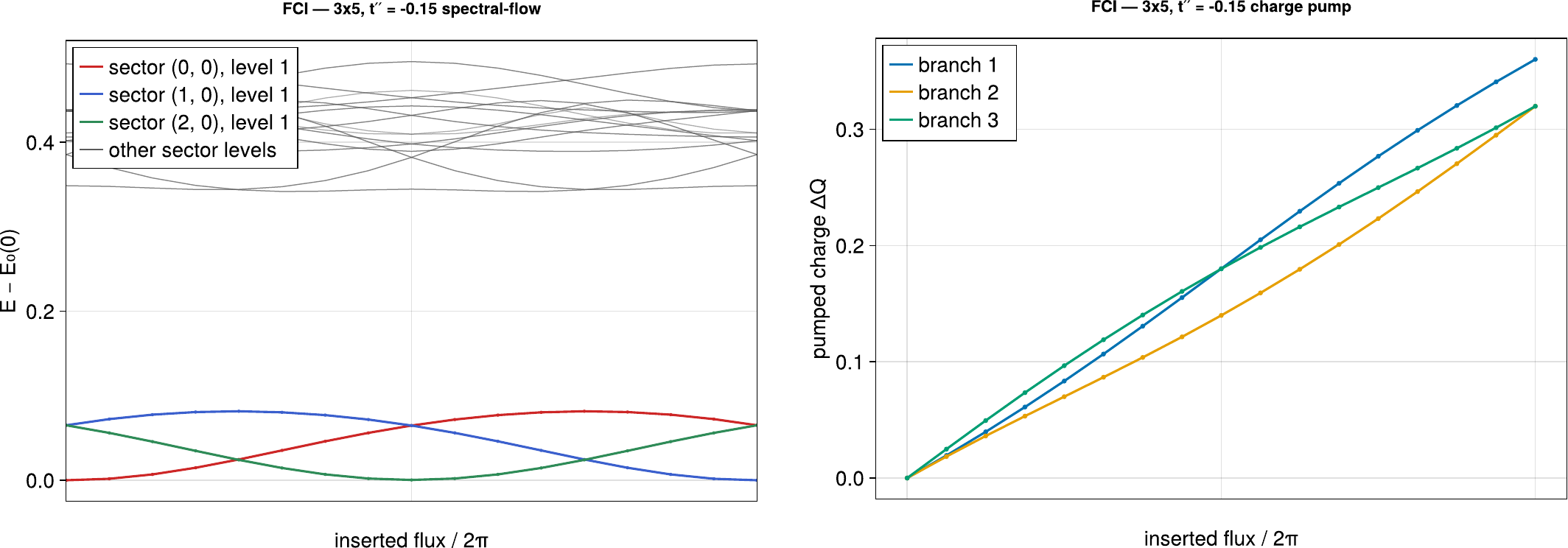}
	\caption{Spectral flow and many-body charge pump for the extended Fermi-Hubbard model at $t''=-0.15$ on a $3\times5\times2$ geometry.}
	\label{figA:fermionic_FCI_spectral_flow_and_charge_pump}
\end{figure}
Under the phase twists, the three nearly-degenerate ground states flow into each other and together carry a total many-body Chern number $C=1$ (each about $C=1/3$), as expected for a $\nu=1/3$ FCI state.

We further justify the FCI nature by computing the particle-entanglement spectra (PES) of the $t''=-0.15$ sample over both $3\times5\times2$ and $3\times6\times2$ geometries, which provide a diagnostic of ideal quasihole counting of fractional quantum Hall states~\cite{regnault2011fractional,bernevig2008model,sterdyniak2011extracting,chandran2011bulk}. For a $\nu=1/3$ FCI state in the same universality class as the $1/3$-Laughlin state on the torus, the counting obeys the so-called $(1,3)$-admissible rule as the generalized Pauli principle~\cite{regnault2011fractional}:
\begin{equation}\label{eqA:generalized_Pauli_principle}
	x_{i+1}-x_i\geq3
\end{equation}
plus the torus wrap condition $x_1+L-x_n\geq3$. Let $L$ be the number of orbitals and $N_A$ the number of retained particles after a particle bipartition, the PES counting $\mathcal N^{\text{torus}}_{(1,3)}(L,N_A)$ below the entanglement gap is determined by~\cite{regnault2011fractional}
\begin{equation}
	\mathcal N^{\text{torus}}_{(1,3)}(L,N_A)=\dfrac{L}{N_A}\binom{L-2N_A-1}{N_A-1}.
\end{equation}
For the two different $3\times5\times2$ and $3\times6\times2$ geometries, the PES counting for retained $N_A=2$ particles gives
\begin{equation*}
	\mathcal N^{\text{$3\times5\times2$ torus}}_{(1,3)}(15,2)=75,\quad \mathcal N^{\text{$3\times6\times2$ torus}}_{(1,3)}(18,2)=117.
\end{equation*}
This agrees exactly with the PES plot in~\cref{figA:fermionic_FCI_PES_combined}, confirming the FCI nature, in particular the $U(1)_3$ Laughlin topological order, of our Fermi-Hubbard model~\cref{eqA:Fermi-Hubbard Hamiltonian}.
\begin{figure*}
	\centering
	\includegraphics[width=0.8\linewidth]{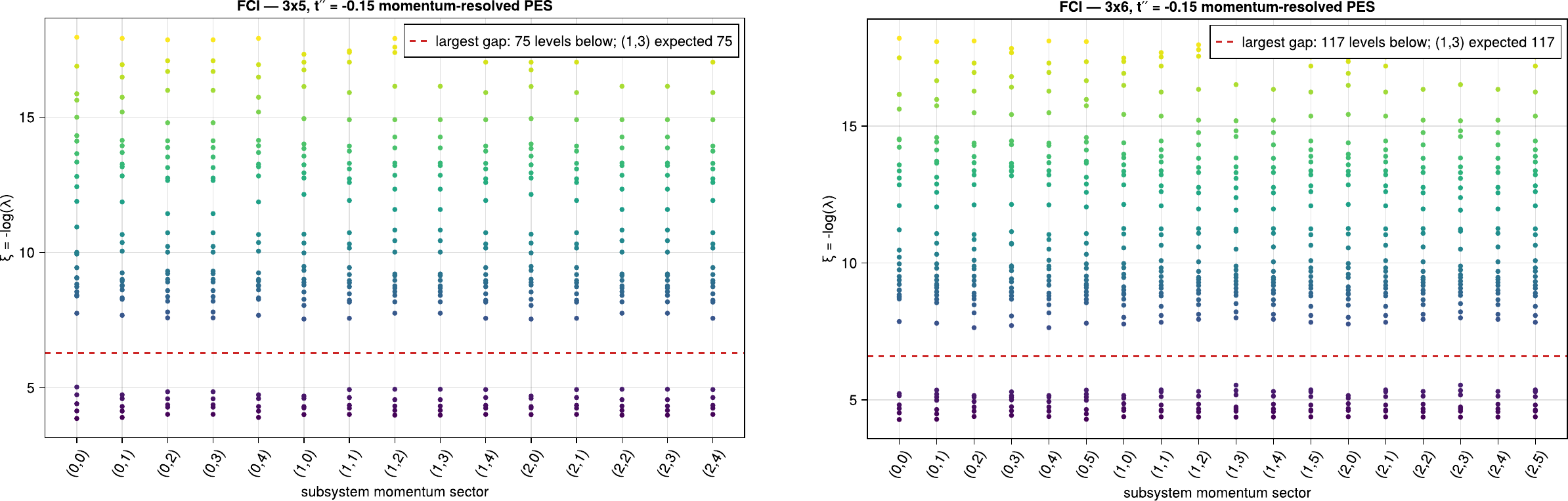}
	\caption{Particle entanglement spectrum for the extended Fermi-Hubbard model~\cref{eqA:Fermi-Hubbard Hamiltonian} at $t''=-0.15$ on $3\times5\times2$ (left) and $3\times6\times2$ (right) geometries.}
	\label{figA:fermionic_FCI_PES_combined}
\end{figure*}

\subsection{Diagnostics of the AHC Phase}
The AHC phase is identified by a symmetry-broken ground state with nonzero Hall response.

To determine the integer many-body Chern number of the AHC phase, we impose the twisted boundary condition on the representative $t''=-0.5$ sample on both $3\times5\times2$ and $3\times6\times2$ geometries and compute the spectral flow and many-body charge pump plot in~\cref{figA:fermionic_AHC_spectral_flow_and_charge_pump}.
\begin{figure}[H]
	\centering
	\includegraphics[width=1.0\linewidth]{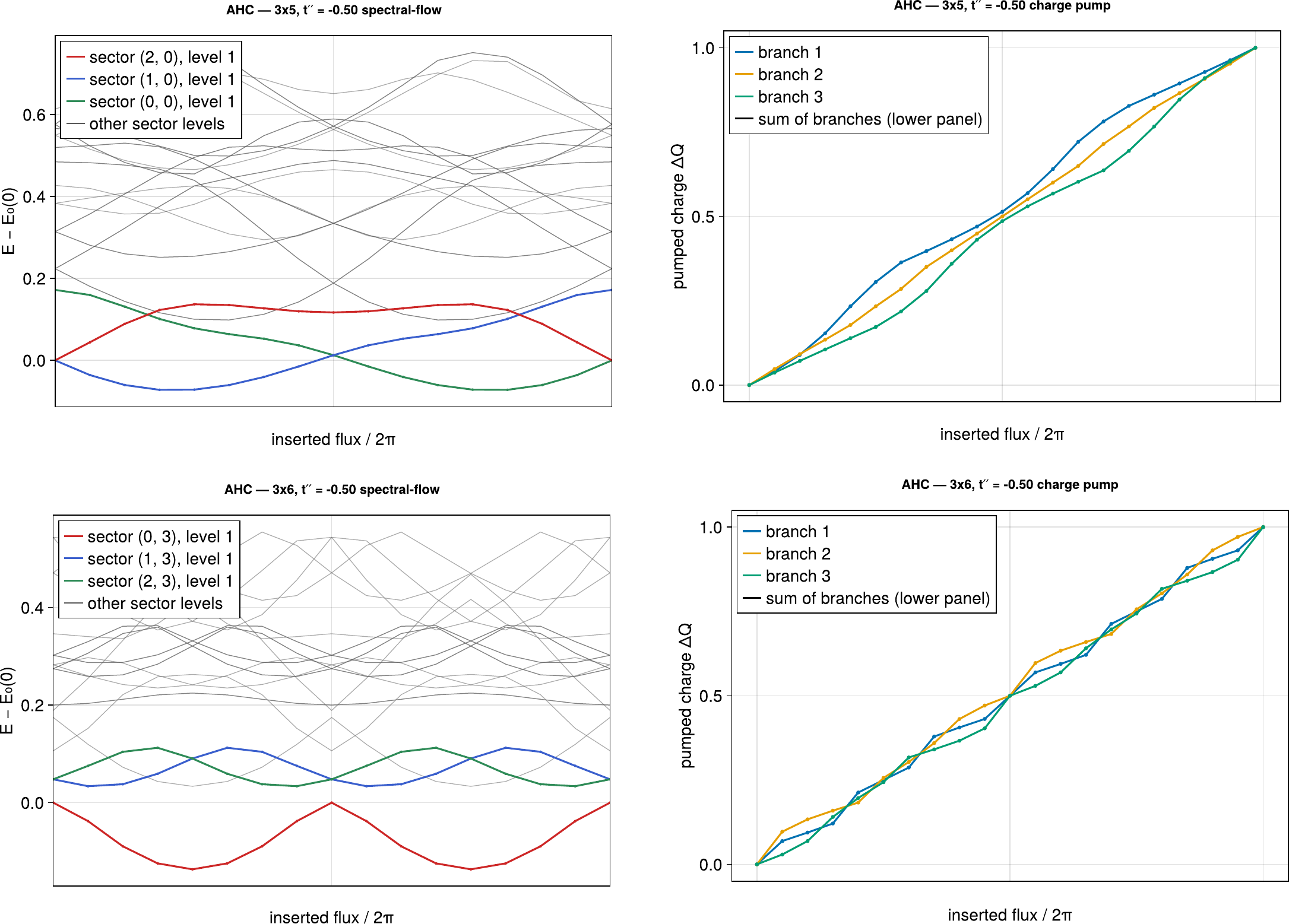}
	\caption{Spectral flow and many-body charge pump for the extended Fermi-Hubbard model at $t''=-0.5$ on both $3\times5\times2$ (top) and $3\times6\times2$ (bottom) geometries.}
	\label{figA:fermionic_AHC_spectral_flow_and_charge_pump}
\end{figure}
Although the three lowest levels get weakly mixed with higher levels under boundary phase twists, an integer-quantized many-body charge pump corresponding to $C=1$ is robust under sample-size changes.

To justify the AHC nature, we compute the connected static structure factor $S(\bm q)$ on both $3\times5\times2$ and $3\times6\times2$ geometries, which is defined as
\begin{equation}\label{eqA:connected_static_structure_factor}
	S(\bm q)=\frac{1}{N_s}\sum_{i,j}e^{i\bm q\cdot(\bm r_i-\bm r_j)}\Big[\langle n_i n_j\rangle-\langle n_i\rangle\langle n_j\rangle\Big].
\end{equation}
A Bragg peak in $S(\bm q)$ at the ordering momentum identifies the translational symmetry breaking. As shown in~\cref{figA:fermionic_AHC_S(q)}, there are clear Bragg peaks at $\bm q=(\pm4\pi/3,0)$ for both geometries, signaling a robust translation-symmetry breaking order in the AHC phase.
\begin{figure}[H]
	\centering
	\includegraphics[width=1.0\linewidth]{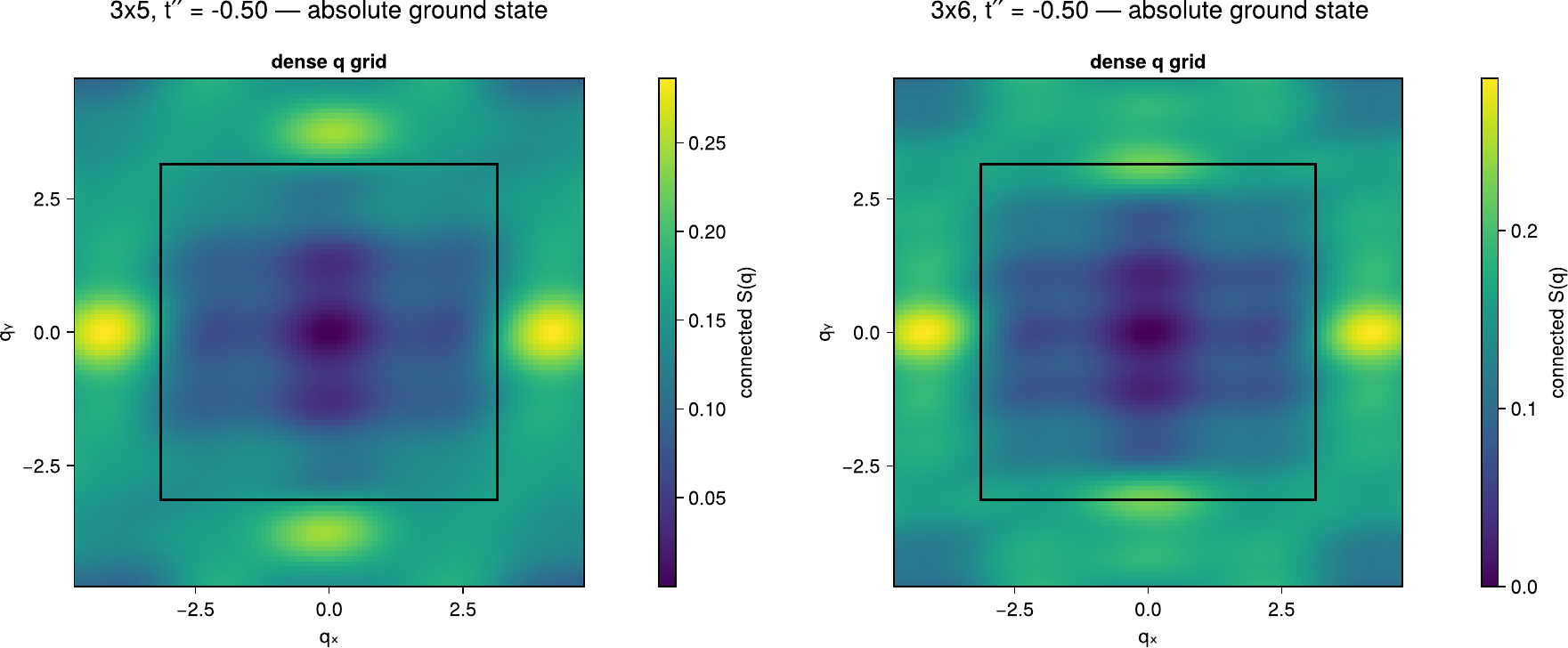}
	\caption{Connected static structure factor $S(\bm q)$ for the extended Fermi-Hubbard model at $t''=-0.5$ on both $3\times5\times2$ (left) and $3\times6\times2$ (right) geometries evaluated on the dense $\bm q$-grids.}
	\label{figA:fermionic_AHC_S(q)}
\end{figure}

We also perform a finite-size-scaling analysis of the charge gap
\begin{equation}\label{eqA:charge_gap}
	\Delta_c(N_s)=E_0(N+1)+E_0(N-1)-2E_0(N),
\end{equation}
in~\cref{figA:fermionic_AHC_charge_gap_finite_size_scaling}. A fit to the ED-accessible clusters extrapolates to $\Delta_c(\infty)\approx0.63$, consistent with a gapped AHC.
\begin{figure}[H]
	\centering
	\includegraphics[width=0.8\linewidth]{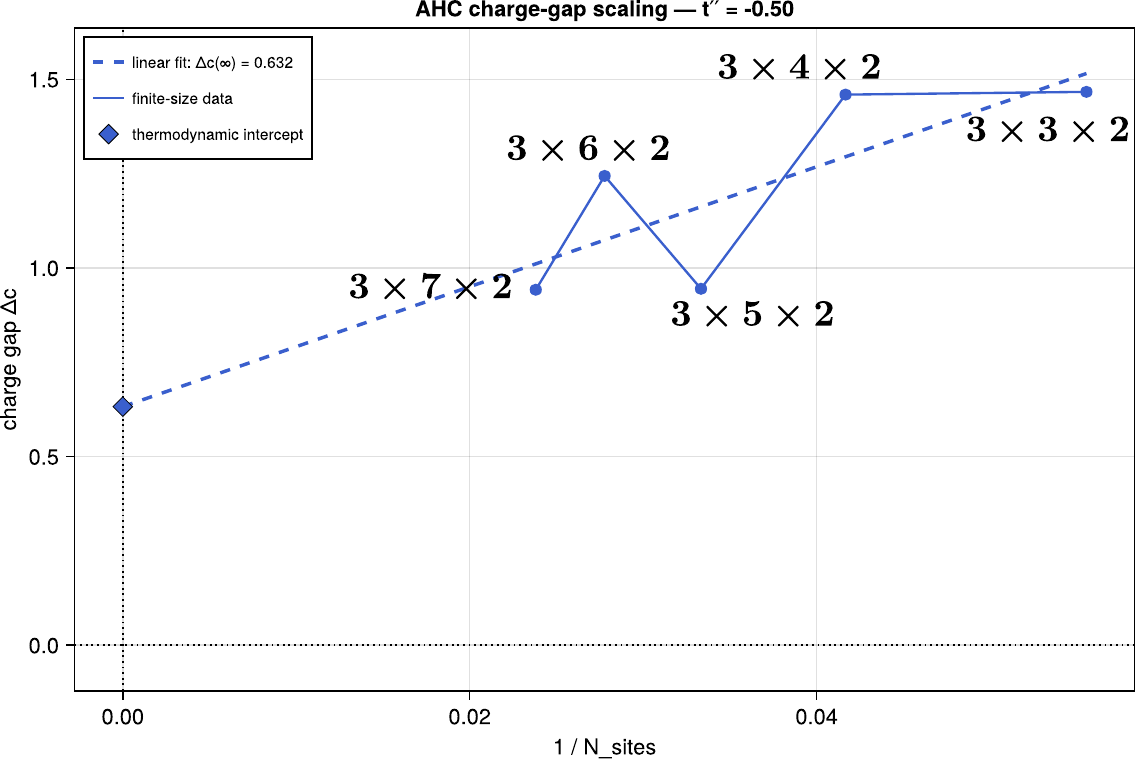}
	\caption{Finite-size scaling of the charge gap for the extended Fermi-Hubbard model at $t''=-0.5$. The data are shown for ED-feasible geometries $3\times3\times2$, $3\times4\times2$, $3\times5\times2$, $3\times6\times2$, and up to $3\times7\times2$.}
	\label{figA:fermionic_AHC_charge_gap_finite_size_scaling}
\end{figure}

\subsection{Diagnostics of the Candidate CDW Phase}
For the candidate CDW phase, we confirmed sharp Bragg peaks in the connected static structure factor~\cref{eqA:connected_static_structure_factor} across different geometries, although the Bragg peak position shifts because the chosen geometries are not mutually commensurate, as shown in~\cref{figA:fermionic_CDW_S(q)}. We emphasize that the robust presence of Bragg peaks itself already supports the translation symmetry breaking in this regime.
\begin{figure}[H]
	\centering
	\includegraphics[width=1.0\linewidth]{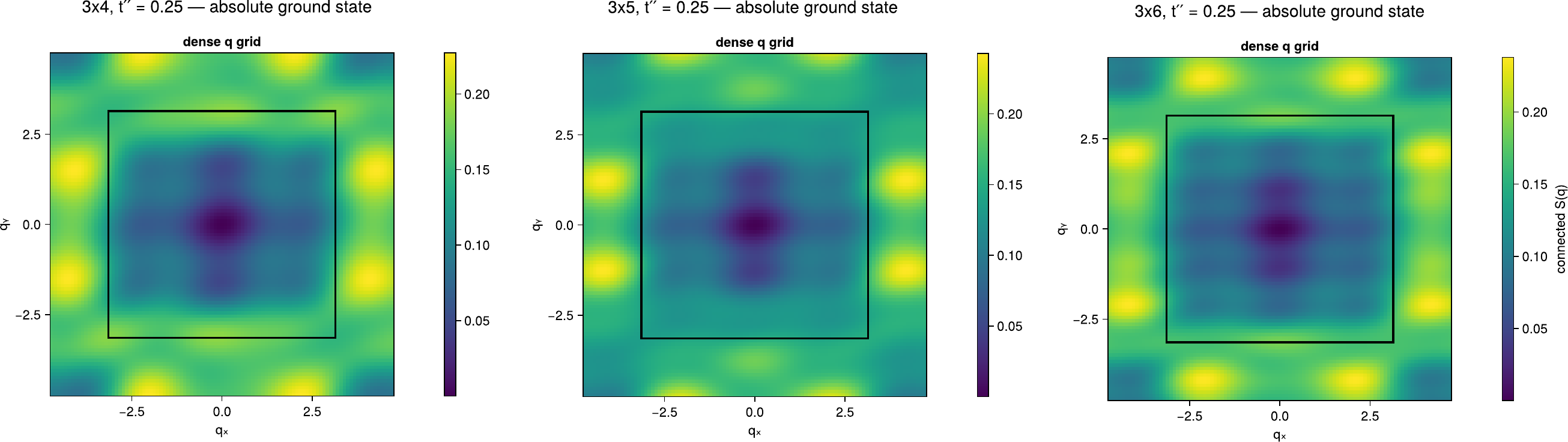}
	\caption{Connected static structure factor $S(\bm q)$ for the extended Fermi-Hubbard model at $t''=+0.25$ on $3\times4\times2, 3\times5\times2$ and $3\times6\times2$ geometries evaluated on the dense $\bm q$-grids. Although the peak position varies among these incommensurate
	geometries, a pronounced Bragg peak is present in each, indicating a robust translational symmetry breaking order.}
	\label{figA:fermionic_CDW_S(q)}
\end{figure}

We also perform a finite-size scaling of the charge gap~\cref{eqA:charge_gap}. As shown in~\cref{figA:fermionic_CDW_charge_gap_finite_size_scaling}, a linear extrapolation gives a small estimated charge gap $\Delta_c(\infty)\approx0.22$. Therefore, determining whether the phase is gapped or gapless requires more numerical evidence from reliable large-scale calculations such as DMRG, which is beyond the scope of this work. Although the metallic/insulating conclusion cannot drawn on current ED-accessible geometries, we refer to this regime as ``candidate CDW'', as the Bragg peak for connected static structure factor is sharp and robust.
\begin{figure}[H]
	\centering
	\includegraphics[width=0.8\linewidth]{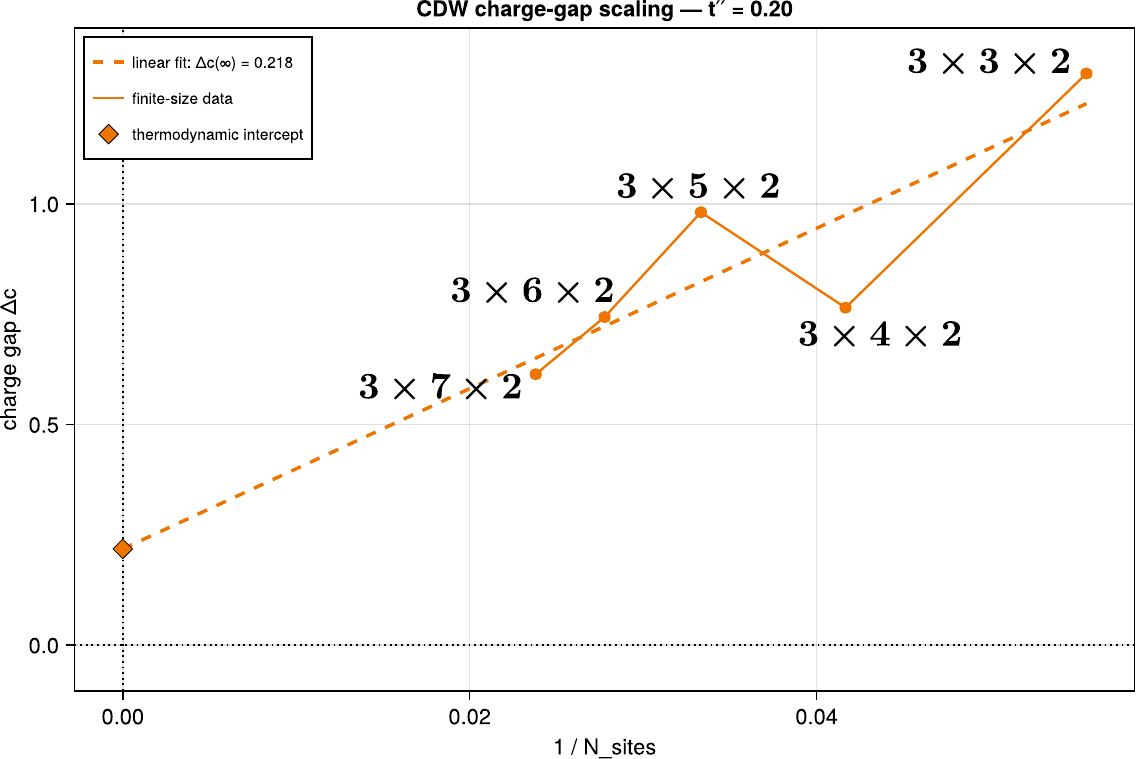}
	\caption{Finite-size scaling of the charge gap for the extended Fermi-Hubbard model at $t''=+0.2$. The data are shown for ED-feasible geometries $3\times3\times2$, $3\times4\times2$, $3\times5\times2$, $3\times6\times2$, and up to $3\times7\times2$.}
	\label{figA:fermionic_CDW_charge_gap_finite_size_scaling}
\end{figure}

\end{document}